\documentclass[conference]{IEEEtran}

\usepackage{textcomp}
\usepackage{xcolor}

\usepackage{amssymb,amsfonts,bm}
\usepackage{graphicx}
\usepackage{color}
\usepackage{upgreek,enumitem}
\usepackage{algorithm}
\usepackage{algpseudocode}

\usepackage[utf8]{inputenc} 
\usepackage[T1]{fontenc}
\usepackage{url}
\usepackage{ifthen}
\usepackage{cite}
\usepackage[cmex10]{amsmath} % Use the [cmex10] option to ensure complicance
\newcommand{\be}{\begin{equation}}
\newcommand{\ee}{\end{equation}}
\newcommand{\ba}{\begin{align}}
\newcommand{\ea}{\end{align}}

\newtheorem{theorem}{Theorem}
\newtheorem{definition}{Definition}[section]
\newtheorem{remark}{Remark}[section]
\newtheorem{corollary}{Corollary}
\newtheorem{lemma}{Lemma}[section]
\newtheorem{claim}{Claim}[section]

\newenvironment{proof}[1][Proof]{\noindent\textit{#1:} }{\hfill$\blacksquare$\par}

\begin{document}

\title{Secret key-distribution over networks with\\ node-based  adversarial errors}

\author{
\IEEEauthorblockN{Reza Sayyari}
\IEEEauthorblockA{
% Department of Electrical Engineering\\
University at Buffalo\\
% Buffalo, NY, USA\\
\texttt{rezasayy@buffalo.edu}
}
\and
\IEEEauthorblockN{Michael Langberg}
\IEEEauthorblockA{
% Department of Electrical Engineering\\
University at Buffalo\\
% Buffalo, NY, USA\\
\texttt{mikel@buffalo.edu}
}
}

\maketitle

\begin{abstract}
We study the multiple key-cast problem, in the context of network coding, under active node-based adversaries.
In multiple key-cast, a source generates independent secret keys to be securely and reliably delivered to designated terminal subsets. 
In our model of study, the network adversary can observe \(\ell_o\) nodes, inject {\em additive} or {\em overwrite} errors into \(\ell_e\) nodes, and simultaneously observe and corrupt  \(\ell_{oe}\) nodes, while having full knowledge of the network topology and coding operations.
{Adversarial models of similar nature, however, where corruption and eavesdropping is done on {\em edges} instead of nodes, have} seen previous studies in the context of secure multicast network-coding.
The work at hand builds on and extends these studies to address the challenges in node-based adversaries in the context of (multiple) key distribution.

For single-source networks in which every node is $d$-vertex connected from the source, we show that perfectly secure multiple key-cast under both additive and overwrite error models is asymptotically achievable at the key-capacity of \(d - \ell_o - \ell_e - 2\ell_{oe}\). 
We then extend our analysis to networks where only terminal nodes satisfy this connectivity requirement, while intermediate nodes may be only partially connected. For these topologies, we develop coding schemes that achieve secure and reliable multiple key-cast capacities determined by the source vertex-connectivity and additional structural properties of the network.
Finally, we show that our results generalize to multi-source settings, ensuring perfect secrecy even if the adversary observes all but one source node, and establish that our constructions apply directly to secure multicast network coding and to network secret-sharing scenarios.
As part of our studies, we improve the security guarantee of a central scheme in [Zhang et al., IEEE Trans. on Comm., 2023] addressing parallel-edge networks, from weak-security to perfect-security.
\end{abstract}

\begin{IEEEkeywords}
Multiple Key-Cast, Network Coding, Adversarial Errors, Network Secret Sharing, Perfect Secrecy.
\end{IEEEkeywords}

\section{Introduction}
In modern distributed networks, shared randomness serves as a critical resource for secure communication and computation. From coordinating federated learning algorithms to securing the Internet of Things (IoT), operations such as encryption, distributed inference, randomized coding, and privacy-preserving computations rely heavily on nodes sharing uniform, secret keys \cite{langberg2022network, 10904071, 10482871, byrd2020differentially, bonawitz2017practical}. As these networks expand in scale and geographic dispersion, they frequently operate in hostile or unverified environments. In such settings, adversaries are rarely limited to passive eavesdropping. Instead, compromised nodes can actively disrupt communications by injecting malicious data, jamming links, and overwriting legitimate network traffic. Ensuring the secure and reliable distribution of keys under such node-based adversary attacks is a critical challenge in modern communication systems.

To address this vulnerability, this work studies the task of secure and reliable key dissemination over networks, specifically within the context of network coding \cite{ahlswede2000network, li2003linear}. We consider communication settings in which source nodes generate independent randomness that is disseminated throughout the network via designed coding operations. The objective is for disjoint subsets of terminal nodes to successfully decode common (secret) keys while overcoming an active, node-based adversary. 
When a single common key is distributed securely and reliably to all terminals, we refer to the problem as (single) key-cast; when distinct keys are distributed to multiple disjoint terminal sets, we use the term multiple key-cast.

The task of multiple key-cast is studied in a prior work of the authors \cite{eavsPaper} over a {\em passive} threat model in which the adversary only observes a subset of nodes. The adversary in the work at hand is {\em active} in the sense that it can corrupt the information traversing the network.
Namely, our threat model assumes an adversary that can observe (up to) $\ell_o$ nodes, inject error into (up to) $\ell_e$ nodes, and simultaneously observe and inject errors into (up to) $\ell_{oe}$ nodes. 
{The adversary strength is guided by the network topology in the sense that it} can decide its injected errors to a given vertex $v$ based on full knowledge of the network topology, the coding operations throughout the network, {and the observed nodes with topological order less than or equal to $v$.\footnote{In Section~\ref{SysM}, we present our detailed model of communication over acyclic networks.}}
{Our model of study fits traditional network coding protocols in which communication is done according to topological order.}
Defending against active adversaries requires robust coding strategies that not only guarantee secrecy but also correct or filter out malicious errors.

In our model of study, motivated by practical settings, the active adversary's capabilities vary across different parts of the network. On some compromised nodes it might be restricted to purely passive observation, on others it might blindly inject errors without seeing the incoming data, and on a third set of nodes it might simultaneously observe and corrupt the transmissions. {As seen in \cite{zhang2023optimal,bakshi2025optimal,9611518}}, and to be discussed in detail shortly, these variations in adversarial capabilities significantly impact the communication rate and connectivity requirements of the network.
Furthermore, the nature of the adversary's error injection capabilities at unobserved nodes plays a critical role in determining system vulnerability. In scenarios where the adversary has limited power, such as in typical wireless communication environments, it may only be capable of injecting additive errors. Under this model, the adversary adds a malicious error to the legitimate node's transmission {without observing the transmission itself}, corrupting the message without entirely eliminating the underlying signal. We refer to such errors as {\em additive} errors.
On the other hand, a more powerful adversary may be capable of {\em overwrite} errors, where it completely replaces the legitimate transmission with its own malicious data. This elevated level of control may impose more severe penalties on achievable rates and demand stricter network connectivity conditions. Therefore, in this work, we distinguish between additive and overwrite error models and provide comprehensive analyses for both. Note that for nodes where the adversary can simultaneously observe and inject errors (we assume $\ell_{oe}$ nodes), this distinction collapses; by observing the stored data, the adversary can precisely tailor an additive error to effectively overwrite the transmission.

Our main contributions in this work are as follows. First, we enhance the secure coding scheme originally proposed in \cite{zhang2023optimal} for parallel edge networks, upgrading its theoretical guarantee from weak (or strong) security to perfect security. In weakly or strongly secure schemes, the mutual information between the eavesdropper’s observations and the secret message only approaches zero asymptotically as the blocklength goes to infinity. In contrast, our modified approach guarantees perfect security, meaning this mutual information is exactly zero. Utilizing our enhanced scheme as a foundational tool, we then establish a number of results addressing network key-capacities and connectivity requirements for multiple key-cast against active, node-based adversaries. While our analysis focuses on the single-source setting, our methodology can be extended to the multi-source setting. For multiple sources, one can apply our scheme to independently disseminate keys using the randomness generated at each source, and subsequently combine these individual keys to construct the final shared key, similar to the strategy used in \cite[Corollaries 1 and 2]{eavsPaper}. This construction ensures that the final key remains perfectly secure even if the adversary observes all but one of the source nodes. Note that as single-source key-cast is equivalent to secure multicast network coding \cite{10904071}, the achievability results and network conditions established in this work hold equally for the secure and reliable multicast problem. Finally, we show that the proposed method directly extends to network secret sharing under the same adversarial model.

The remainder of this paper is organized as follows. Section~\ref{RelW} reviews a collection of related works. In Section~\ref{SysM}, we define the system model and introduce the terminology used throughout the paper. In Section~\ref{dCon}, we study networks in which all nodes are $d$-vertex connected from the source and present perfectly secure multiple key-cast schemes that achieve the network key-capacity.
In Section~\ref{Pcon}, we relax this connectivity requirement to accommodate intermediate nodes that may only be partially connected, presenting coding solutions whose established key-rate depends on the network vertex-connectivity from the source and certain additional network properties. 
{The proofs of the theorems and lemmas are provided in the Appendix.}
Finally, Section~\ref{Con} concludes the paper.

\section{Related Works}\label{RelW}
In this section, we review the literature most relevant to our multiple key-cast model.
% and highlight the theoretical gaps our work addresses. 
We organize our discussion into three primary areas. First, we examine secure network coding and multiple key-cast under passive adversarial models. Second, we review the framework of Byzantine network coding to explore reliable communication against active malicious attacks. 
{Finally, we discuss prior work on different notions of causality in the context of network coding and their relation to our adversarial model.}

\subsection{Secure Network Coding and (Multiple) Key-Cast with a Passive Adversary}
The task of key-cast is closely related to secure network coding under passive eavesdropping models \cite{cai2002secure,cai2010secure,feldman2004capacity}, where the objective is to reliably communicate specific source messages while guaranteeing zero information leakage to any adversary observing up to a bounded number of network nodes or edges. 
Treating a secret message as a secret key, it is observed in \cite{10904071} that single-source key-cast is equivalent to secure multicast network coding; however, the multi-source multiple key-cast setting fundamentally differs. 
In secure network coding, terminals must reconstruct predefined source messages, and the source nodes themselves are assumed safe from observation. Multiple key-cast, however, does not require the reconstruction of source information, and only requires terminals in the same subset to decode a common secret key that may contain a mixture of the randomness generated at source nodes. Because specific message reconstruction is not required, the key-cast model accommodates scenarios where source nodes themselves might be observed by the eavesdropper. Furthermore, there exist network instances in which secure network coding is impossible, while key-cast remains achievable (e.g., see \cite[Fig. 1]{eavsPaper}). While secure network coding against edge-based adversaries is well-understood (e.g., the secrecy capacity \cite{cai2002secure,cai2010secure} and achieving optimal rates via random linear network coding \cite{feldman2004capacity}), the more challenging node-based setting remains largely open. Existing node-based models are often restricted to unicast routing \cite{che2013routing}, rely on pre-shared keys \cite{zhang2013lightweight}, or assume only single, non-colluding eavesdroppers \cite{lima2007random, wang2016optimal}. For a more detailed survey on prior works in this context, we refer the reader to \cite{eavsPaper}.

To address the problem of multiple key-cast over networks, recent works draw connections to secure regenerating codes \cite{dimakis2010network, shah2011information} and distributed secret sharing \cite{shah2015distributed}. Originally designed for efficient data repair in distributed storage and secret sharing, these frameworks allow authorized nodes to reconstruct a common secret from distributed shares without exposing information to passive adversaries. Building on this perspective, \cite{10904071} adapts distributed secret sharing to construct a multiple key-cast protocol. Their scheme successfully disseminates keys over networks with a single source, keeping the keys perfectly secure from a passive eavesdropper capable of observing any single non-terminal node (i.e., in our notation, $\ell_o = 1$).
\cite{langberg2024characterizing} characterizes the network connectivity conditions required for positive-rate key-cast under a similar single-node threat model. Extending beyond the $\ell_o=1$ regime, prior work \cite{eavsPaper} of the authors generalizes the multiple key-cast protocol to accommodate eavesdroppers that control multiple colluding network nodes for general values of $\ell_o$, given some restrictions on the network connectivity. 
While all prior studies on the problem of key-cast address the setting of passive adversaries, 
the work at hand addresses networks threatened by active disruptions, providing coding solutions against adversaries that can simultaneously observe and inject malicious data across multiple nodes.

\subsection{Byzantine Network Coding and the Limited-View Adversary}
The problem of reliable and secure communication in the presence of active adversaries has been extensively studied in the network coding literature. Existing models can be classified according to whether the adversary controls network edges or nodes. In the edge-based setting, the adversary corrupts a subset of network edges, whereas in the node-based setting it compromises a subset of nodes, thereby gaining access to all incoming information and control over all outgoing edges. Therefore, node-based adversaries are significantly stronger threat models. In the following, we briefly review the literature on both edge-based and node-based adversaries, highlighting the challenges that motivate our study.

\subsubsection*{Edge-Based Adversary}
%, but does not address error correction. 
\cite{yeung2006network, cai2006network} establish the fundamental limits of reliable communication over networks with adversarial errors and show that an omniscient adversary capable of corrupting $\ell$ edges imposes a rate penalty of $2 \ell$, yielding a network equivalent of the classical Singleton bound. 
Here, an omniscient adversary is one that has access to all information traversing the network.
\cite{koetter2008coding} considers an omniscient adversary capable of observing all network transmissions and injecting worst-case errors and erasures, and introduces an abstract subspace-coding framework that encodes information as vector spaces, requires no knowledge of the network topology, and asymptotically achieves the optimal transmission rate of $C - 2\ell$, where $C$ denotes the network capacity.
Unlike an omniscient adversary, a limited-view adversary observes only a portion of the network traffic. When an adversary is limited-view, they cannot perfectly align injected errors with the overall transmitted codewords. Under this assumption, hashing and detection mechanisms can be leveraged to explicitly identify corrupted links, allowing the receiver to treat malicious injections as erasures and thereby reducing the capacity penalty. 
\cite{ho2004byzantine} studies Byzantine adversaries and proposes a probabilistic error-detection scheme, based on polynomial hashing, that succeeds in the restricted setting that receiver nodes have access to at least one uncorrupted packet that is unknown to the adversary.
\cite{jaggi2005correction} studies Byzantine network error correction and develops a coding scheme that enables reliable communication despite adversarial corruption of network edges. For parallel-edge networks, the scheme achieves rates approaching the network capacity when the adversary controls only a sufficiently small fraction of the network resources.
\cite{jaggi2007resilient} subsequently develop schemes for local Byzantine network error correction, achieving the optimal rates of $C-2\ell_{oe}$ and $C-\ell_e$ for omniscient and the limited-knowledge adversaries of \cite{jaggi2005correction}, respectively. 
Here, we use our notation $\ell_{oe}$ and $\ell_e$ mentioned previously.
Complementing the works mentioned above, a rich line of work has studied edge-based adversaries that can both eavesdrop and inject errors, e.g.,  \cite{yao2014network, hayashi2017secrecy, guo2018some, cai2019secure, hayashi2020reduction, hayashi2021asymptotically}. These studies characterize capacity limits, including under adaptive and time-varying corruption patterns, and develop secure network coding schemes that ensure both reliability against malicious injections and secrecy against eavesdropping.
Most relevant to the setting studied in this work, \cite{zhang2023optimal} uses a general threat model in which the adversary may eavesdrop on a specific subset of edges while injecting errors into another, with the possibility that these two subsets overlap, and characterizes the optimal capacity and secrecy-capacity against such adversaries over parallel multipath networks, utilizing a combination of Maximum Distance Separable (MDS) codes and pairwise hashing.
Their construction, however, guarantees only asymptotic secrecy. \cite{9611518} extends this framework, using subspace-codes, to arbitrary network topologies with, as before, asymptotic secrecy guarantees.
The work at hand addresses node-based adversaries and requires schemes that provide perfect secrecy.
As part of our analysis, we modify the hashing mechanism of \cite{zhang2023optimal} to obtain a perfectly secure variant used as a building block for our multiple key-cast scheme.

\subsubsection*{Node-Based Adversary}
Despite the operational significance of node-based threat models in modern communication networks,
there are a limited number of works addressing active node-based adversaries.
\cite{kosut2014polytope} studies omniscient adversaries that compromise a subset of intermediate nodes at unknown locations. They introduce polytope codes, a class of non-linear codes that enable internal nodes to detect adversarial behavior through local consistency checks, and show that the cut-set bound is achievable for planar networks. 
Furthermore, \cite{tian2016arbitrarily} studies an omniscient, node-based threat model where the adversary knows all network traffic and can alter the data transmitted from compromised nodes.
To achieve network capacity, their construction combines random linear network coding with a verification mechanism that uses pre-shared randomness in each node to detect and discard corrupted packets. However, the scheme assumes the existence of these secure shared secrets and does not address how such randomness can be established in the presence of adversarial nodes.
In comparison to the above, this work addresses secure multiple key-cast against node-based active adversaries by introducing a perfectly secure scheme that does not rely on pre-shared randomness and that can operate under general vertex-connectivity conditions.

\subsection{Causal adversarial models in Network Coding}
The literature features {a number of related}  formulations of causal adversaries, differing  in how the attacker's observations and interventions are temporally constrained. {\cite{nutman2008adversarial, kosut2014generalized, bakshi2025optimal} study a causal-omniscient edge-based adversarial model in which corrupted edge transmissions at a given time $t$ may rely on the adversary's view up-to time $t$ (\cite{nutman2008adversarial} also add a delay parameter $\Delta$).
\cite{hayashi2021asymptotically} take a slightly different approach, closest to ours, and assume an edge-based active adversary constrained by a topological ordering of the network links, deciding on the corrupted edge value based on the information transmitted on edges of lesser (or equal) topological order. 
All works above study a {\em multiple-shot} setting that allows to {\em hide} random bits governing the hash functions used at communication round $t$ from the adversary by transmitting these random bits at a later round.
As outlined in detail shortly in Section~\ref{SysM}, in this work, we follow the model of \cite{hayashi2021asymptotically} in a {\em one-shot} setting in which communication is considered to be done according to topological order and corrupted information leaving a given node $v$ may depend on {\em all}  information passing through nodes of lesser (or equal) topological order.
The one-shot model does not allow to hide randomness from the adversary in a way that the multiple-shot model does.
More precisely, we strengthen the adversary when compared to the one-shot variant of \cite{hayashi2021asymptotically} by allowing adversarial actions at node $v$ to depend on the information available at nodes that are not descendants of $v$.}

\section{System Model}\label{SysM}

\paragraph{Notation}
Matrices are denoted by uppercase boldface letters (e.g., $\boldsymbol{A}$), vectors by lowercase boldface letters (e.g., $\boldsymbol{a}$), and scalars by italic letters (e.g., $a$). For a matrix $\boldsymbol{A} \in \mathbb{F}_q^{m \times n}$ and a vector $\boldsymbol{a} \in \mathbb{F}_q^n$, let $[\boldsymbol{A}]_{i:j}$ denote the submatrix of $\boldsymbol{A}$ formed by columns $i$ through $j$, and $[\boldsymbol{a}]_{i:j}$ denote the subvector of $\boldsymbol{a}$ consisting of components $i$ through $j$. The neighborhood of a node $v$, denoted by $\mathcal{N}(v)$, is the set of nodes that are direct outgoing neighbors of $v$ (i.e., nodes connected by an edge from $v$). The set of parent nodes of a node $v$, denoted by $\mathcal{P}(v)$, is the set of nodes with direct edges leading into $v$.

\paragraph{Network and Adversary Model}
The network instance $\mathcal{I}=\left(\mathcal{V}, \mathcal{E}, \mathcal{B}, \mathcal{S}, \mathcal{T}\right)$ is represented by a directed acyclic graph (DAG)
$\mathcal{G} = (\mathcal{V}, \mathcal{E})$, 
where $\mathcal{V} = \{v_1,\dots,v_{|\mathcal{V}|}\}$ is the set of nodes,
$\mathcal{E}$ is the set of directed edges, 
$\mathcal{S} \subseteq \mathcal{V}$ is the set of source nodes, where each source $s \in {\mathcal S}$ holds an unlimited number of independent random symbols ${X}_s$ drawn from $\mathbb{F}_q$,
and $\mathcal{T} \subseteq \mathcal{V}$ is the set of terminal nodes.
The set $\mathcal{T}$ is partitioned as
$\mathcal{T}=\{\mathcal{T}_i\}_{i=1}^m$ into $m$ pairwise disjoint subsets $\mathcal{T}_1,\mathcal{T}_2,\dots,\mathcal{T}_m$,
where for each $i=1,\dots, m$, terminals in $\mathcal{T}_i$ decode the same key (distinct from the keys decoded in other subsets $\mathcal{T}_j$, $j \ne i$). All operations in the network are performed over a finite field $\mathbb{F}_q$ of size $q$. Every edge $e \in \mathcal{E}$ has unit capacity, i.e., each edge can carry a single symbol from $\mathbb{F}_q$. To model higher capacity edges, multiple (unit-capacity) edges may exist between two given nodes.
The security and reliability requirements are evaluated against a computationally unbounded active adversary. The scope of the adversary is defined by a collection $\mathcal{B}$ of admissible attack configurations. Each configuration $\beta \in \mathcal{B}$ is specified by a tuple of pairwise disjoint subsets of nodes,
\(
\beta=(\beta_o,\beta_e,\beta_{oe}),
\)
where the sets $\beta_o$, $\beta_e$, and $\beta_{oe}$ denote the nodes that the adversary may observe, corrupt, and simultaneously observe and corrupt, respectively. 
For a given attack configuration $\beta\in\mathcal{B}$, let $\beta_{\mathrm{obs}}=\beta_o\cup\beta_{oe}$ denote the set of observed nodes and let $\beta_{\mathrm{jam}}=\beta_e\cup\beta_{oe}$ denote the set of nodes, in which the adversary can inject errors.

For any edge $e=(u,v) \in \mathcal{E}$, let $X_{u \to v}$ be the symbol transmitted by node $u$, and $Y_{u \to v}$ be the symbol received by node $v$. If $u \notin \beta_{\mathrm{jam}}$, the transmission is error-free ($Y_{u \to v} = X_{u \to v}$). If $u \in \beta_{\mathrm{jam}}$, the adversary can corrupt the transmission. We consider two active adversary models:
\begin{itemize}
    \item \textbf{Additive Jamming:} The adversary injects an error $Z_{u \to v} \in \mathbb{F}_q$, such that $Y_{u \to v} = X_{u \to v} + Z_{u \to v}$.
    \item \textbf{Overwrite Jamming:} The adversary replaces the transmitted symbol with $Z_{u \to v} \in \mathbb{F}_q$, so that $Y_{u \to v} = Z_{u \to v}$.
\end{itemize}
\

A defining feature of our threat model is the set of topological constraints  imposed on the adversary's actions. {The network $\mathcal{G}$} induces a natural partial order $\prec_{\mathcal{G}}$ on the vertex set $\mathcal{V}$, where $v \prec_{\mathcal{G}} u$ if and only if there exists a directed path from $v$ to $u$ ($v \ne u$).
The adversary is assumed to know the network topology, the coding scheme employed by the nodes, and can choose any compromised set $\beta \in \mathcal{B}$. 
{The adversary may} pool observations from all topologically concurrent nodes before committing to an injection.
Formally, for an error injection at a node $v \in \beta_{\mathrm{jam}}$, the adversary's valid observation set $\mathcal{O}(v)$ consists of all compromised observation nodes that are not descendants of $v$;
$ \mathcal{O}(v) \triangleq \{ u \in \beta_{\mathrm{obs}} \mid v \not\prec_{\mathcal{G}} u \}. $
The condition $v \not\prec_{\mathcal{G}} u$ ensures that $u$ 
{is not a descendant of $v$.}
When injecting an error $Z_v$ at node $v$, the adversary calculates its action solely as a function of the topology, the protocol, and the symbols observed in $\mathcal{O}(v)$.

\paragraph{Key-Codes}
For a transmission scheme, a \emph{key-code} $\mathcal{C} = (\mathcal{F}, \mathcal{G})$ consists of a collection of local encoding functions $\mathcal{F} = \{ f_e : e \in \mathcal{E} \}$ and decoding functions $\mathcal{G} = \{ g_{t} : t \in \mathcal{T}\}$. For each edge $e=(u,v) \in \mathcal{E}$, the transmitted message $X_{u \to v}$ is obtained by applying the encoding function $f_e$ to the information ${X}_u$ available at node $u$. 
More precisely, for a generic node $u \in \mathcal{V}\setminus \mathcal{S}$, the set of available inputs is given by
\(
{X}_{u} = \big( (Y_{w \to u} : (w,u) \in \mathcal{E}) \big),
\)
which captures all potentially corrupted symbols received by $u$. 
The information available at a source $s$ includes the randomness ${X}_s$ discussed previously.
Throughout the paper, we assume that source nodes have no incoming edges and terminal nodes have no outgoing edges.
Communication proceeds according to a topological ordering.

\paragraph{Secure Multiple Key-Cast Instance}
An instance ${\mathcal{I}=\left(\mathcal{V}, \mathcal{E}, \mathcal{B}, \mathcal{S}, \mathcal{T}\right)}$ allows secure multiple key-cast with error probability $\epsilon$ if there exists a key-code that, for any adversarial configuration $\beta \in \mathcal{B}$, enables a collection of keys $\{K_1, K_2, \dots, K_m\}$ to be transmitted, where each key $K_i$ is intended for terminal set $\mathcal{T}_i \subseteq \mathcal{T}$, such that the following two conditions are satisfied:
\begin{itemize}
    \item \textbf{Reliability:} For any $\beta \in \mathcal{B}$, each terminal $t \in \mathcal{T}_i$ can recover its corresponding key $K_i$ with a decoding error probability bounded by $\epsilon$. Specifically, if $\hat{K}_i = g_t(Y_t)$ is the estimated key at terminal $t$, then:
    \[
        \forall i \in \{1,\dots,m\},\ \forall t \in \mathcal{T}_i, \quad \Pr(\hat{K}_i \neq K_i) \le \epsilon.
    \]
    
    \item \textbf{Perfect Secrecy:} For each $i \in \{1,\dots,m\}$ and for any $\beta \in \mathcal{B}$, such that $\beta_{\mathrm{obs}} \cap \mathcal{T}_i =\emptyset$, the adversary obtains zero information about $K_i$ from all observed nodes. Equivalently,
    \(
        I\!\left(K_i ; \{X_v : v \in \beta_{\mathrm{obs}}\} \right) = 0 .
    \)
\end{itemize}

\begin{definition}[Blocklength]
The blocklength, denoted by $N$, is the number of independent network uses of a coding scheme. Under blocklength $N$, all transmissions over any edge $e \in \mathcal{E}$ are vectors in $\mathbb{F}_q^{N}$.
\end{definition}

\begin{definition}[Key-Rate]
    For $i \in \{1,\dots,m\}$, let $K_i$ denote the key decoded at the terminals of $\mathcal{T}_i \in \mathcal{T}$ and let $N$ denote the code blocklength. The network key-rate $R$ is defined as
    $$
    R = \min_{i \in \{1,\dots,m\}}{\frac{H(K_i)}{N}}.
    $$
    where $H(\cdot)$ denotes the Shannon entropy computed using base $q$ logarithms (i.e., measured in units of $\mathbb{F}_q$ symbols), such that a uniformly distributed element in $\mathbb{F}_q$ has an entropy of $1$.
\end{definition}

\begin{definition}[Key-Capacity]
Instance $\mathcal{I}=\left(\mathcal{V}, \mathcal{E}, \mathcal{B}, \mathcal{S}, \mathcal{T}\right)$ has key-capacity $C$ if for any $\epsilon > 0$ and $\gamma > 0$ there exists a finite blocklength $N$ and a corresponding key-code $\mathcal{C}=  (\mathcal{F}, \mathcal{G})$ achieving key-rate $R \ge C - \gamma$ with decoding error at most $\epsilon$.
\end{definition}

\begin{definition}[$d$-Vertex Connected]
    A node $v \in \mathcal{V}$ is said to be $d$-vertex connected from a source $s \in \mathcal{S}$ if for any integer $0 \leq b \leq d$,
    there exist at least $b$ vertex-disjoint paths from $s$ to $v$ and at least $d-b$ edges connecting $s$ to $v$.
    For any $v \notin \mathcal{N}(s)$, being $d$-vertex connected implies that the removal of any set of fewer than $d$ intermediate nodes (excluding $s$ and $v$) does not disconnect $v$ from $s$.
\end{definition}

Throughout, we associate with each node $v_i \in \mathcal{V}$ a distinct element $\rho_i \in \mathbb{F}_q \setminus \{0\}$ and a corresponding Vandermonde vector.
\begin{definition}[Vandermonde Vector]
    The Vandermonde vector of size $d$ assigned to node $v_i$ is defined as  
    $$
    \boldsymbol{v}_i = \big(1, \rho_i, \rho_i^2, \ldots, \rho_i^{d-1}\big) \in \mathbb{F}_q^d,
    $$  
    where $\{\rho_1, \dots, \rho_{|\mathcal{V}|}\}$ are distinct elements strictly chosen from the non-zero elements of the underlying field $\mathbb{F}_q \setminus \{0\}$.
\end{definition}

\section{Multiple Key-Cast in $d$-vertex Connected Network}\label{dCon}
In this section, we consider network instances in which every node is $d$-vertex connected from every source in $\mathcal S$.
First, building upon the framework of \cite{zhang2023optimal}, we establish that any node that is $d$-vertex connected from a source $s$ can reliably and securely receive information in the presence of an active adversary. Our construction refines the hashing of \cite{zhang2023optimal} to upgrade its asymptotic weak and strong secrecy guarantees to perfect secrecy.
We then propose a perfectly-secure multiple key-cast scheme that achieves the key-capacity $d-\ell_o-\ell_e-2\ell_{oe}$ against active adversaries.
We extend this result to the multiple-source setting. In this setting, source nodes may be observed by the adversary but are assumed to be safe from active corruption. (If an adversary were to actively corrupt a source node, it could, for example, mask all outgoing information using independent randomness and, as such, remove the corresponding source-generated information from the suggested coding scheme.)
Finally, we observe that our proposed construction extends naturally to the network secret sharing problem \cite{shah2015distributed}, implying that the capacity results derived in this section remain valid in that setting as well.
{We note that the proof of Lemma~\ref{lem:parallel_edge_active_secure} below assumes the adversarial model of \cite{zhang2023optimal} in which the adversary can base its corruptions on the information available to any node in the set $\beta_{\mathrm{obs}}$ irrelevant of its topological order. Thus, Lemma~\ref{lem:parallel_edge_active_secure} extends the results of  \cite{zhang2023optimal} to perfect-security. The theorems that follow use the adversarial model described in Section~\ref{SysM}.}

\begin{lemma}[Secure transmission to a $d$-vertex connected node]\label{lem:parallel_edge_active_secure}
Consider a network instance $\mathcal{I}=(\mathcal{V},\mathcal{E},\mathcal{B},\mathcal{S}=\{s\},\mathcal{T}=\{t\})$ with
$$
{\begin{aligned}
\mathcal{B} = \Big\{ &(\beta_o, \beta_e, \beta_{oe}) \mid \text{disjoint}\ \beta_o, \beta_e, \beta_{oe} \subset \mathcal{V} \setminus \{s\}, \\
& |\beta_o| \le \ell_o, \, |\beta_e| \le \ell_e, \, |\beta_{oe}| \le \ell_{oe}
%, \\
%& \text{$(\beta_o, \beta_e, \beta_{oe})$ mutually disjoint}
\Big\}.
\end{aligned}}
$$
The secure message capacity from the source $s$ to the terminal $t$ is given by
$$
C_{\mathrm{add}} =
\begin{cases}
d-\ell_o-\ell_e-2\ell_{oe}, & \text{if } d>\ell_o+\ell_e+2\ell_{oe},\\[2ex]
0, & \text{otherwise,}
\end{cases}$$
under the additive jamming model, and
$$
C_{\mathrm{ow}} =
\begin{cases}
d-\ell_o-\ell_e-2\ell_{oe}, & \text{if } d>\ell_o+2\ell_e+2\ell_{oe},\\[2ex]
0, & \text{otherwise,}
\end{cases}$$
under the overwrite jamming model.
\end{lemma}

\begin{proof}[Proof sketch]
{We provide only a proof sketch here and defer the full proof to Appendix~\ref{LemmaProof}.}
We first observe that reliable and secure communication over $d$ vertex-disjoint paths under a node-based adversary is structurally equivalent to communication over a $d$-parallel-edge network under an edge-based adversary. This equivalence connects our setting to the model studied in \cite{zhang2023optimal}. However, achieving the stated rates with perfect secrecy requires modifying the hashing scheme therein, which guarantees only weak security. Specifically, we refine the hashing construction so that the resulting hash reveals no information about the data carried on other edges, thereby upgrading the weak (and strong) secrecy guarantees of \cite{zhang2023optimal} to perfect secrecy. Roughly speaking, compared to the linear hashing used in \cite{zhang2023optimal}, our mechanism employs a polynomial hash coupled with a one-time pad. By operating over a sufficiently large finite blocklength, the fractional overhead of the hashing and padding becomes asymptotically negligible, allowing the transmission rate to approach the secure message capacity with decoding-error decaying with the underlying field size.
\end{proof}

\begin{remark}[Tight Capacity Characterization]
\label{rem:asymptotic_optimality}
The secure coding schemes established in Lemma~\ref{lem:parallel_edge_active_secure}, and Theorem~\ref{T1} that follows, obtain asymptotically optimal rates. Namely, slight modifications to the converses in \cite{zhang2023optimal} (to address the node-based adversarial setting studied in this work) establish 
{the required bounds.}
\end{remark}

\begin{algorithm}[t]
\caption{Multiple Key-Cast in $d$-vertex conn. networks}
\label{alg:keycast_active}
\begin{algorithmic}[1]
\State Source node $s$ generates $n$ symmetric random matrices $\boldsymbol{M}^{(1)}, \dots, \boldsymbol{M}^{(n)} \in\mathbb{F}_q^{d'\times d'}$, where each entry in the upper triangular part is chosen independently and uniformly from $\mathbb{F}_q$.

\State \textbf{Pre-distribution:} Using Lemma~\ref{lem:parallel_edge_active_secure}, for each node $j \in \mathcal{V}\setminus\{s\}$, source $s$ securely transmits a random value $\alpha_j \in \mathbb{F}_q$ and polynomial hashes $h_{j_p \to j}$ for all $j_p \in \mathcal{P}(j)$.

\ForAll{$j\in\mathcal{N}(s)$}
    \State $s$ sends $\left[ \boldsymbol{v}_j^T \boldsymbol{M}^{(1)}\boldsymbol{V}_{s\to j}, \dots, \boldsymbol{v}_j^T \boldsymbol{M}^{(n)}\boldsymbol{V}_{s\to j} \right]^T$ to node $j$
\EndFor

\ForAll{nodes $j\in\mathcal{V}\setminus\{s\}$ in topological order}
    \State Each non-source parent node $j_p\in\mathcal{P}(j)$ sends data vector $\boldsymbol{s}_{j_p \to j} = \left[ \boldsymbol{v}_j^T\boldsymbol{M}^{(1)}\boldsymbol{v}_{j_p}, \dots, \boldsymbol{v}_j^T\boldsymbol{M}^{(n)}\boldsymbol{v}_{j_p} \right]^T$ to  $j$
    \State  \textbf{Local error-correction:} Node $j$ verifies incoming vectors against $h_{j_p \to j}$ using $\alpha_j$, marking corrupted transmissions as erasures
    \State Node $j$ concatenates received symbols and recovers share $\boldsymbol{S}_j = \bigl[\boldsymbol{M}^{(1)}\boldsymbol{v}_j, \dots, \boldsymbol{M}^{(n)}\boldsymbol{v}_j\bigr]$ via MDS decoding
\EndFor

\ForAll{$i \in [m]$ and $t\in\mathcal{T}_i$}
    \State Terminal $t$ outputs
    \[
    \boldsymbol{K}_{\mathcal{T}_i}
    =
    \bigl[\boldsymbol{M}^{(1)} \boldsymbol{v}_{\mathcal{T}_i}, \dots, \boldsymbol{M}^{(n)} \boldsymbol{v}_{\mathcal{T}_i}\bigr]_{1:d'-(\ell_o+\ell_{oe})}
    \]
\EndFor
\end{algorithmic}
\end{algorithm}

\begin{theorem}[Single-source setting]\label{T1}
Consider an instance $\mathcal{I}=\bigl(\mathcal{V},\mathcal{E},\mathcal{B},\mathcal{S}=\{s\},\mathcal{T}\bigr)$ of the Secure Multiple Key-Cast problem with
$$
{\begin{aligned}
\mathcal{B} = \Big\{ &(\beta_o, \beta_e, \beta_{oe}) \mid \text{disjoint}\ \beta_o, \beta_e, \beta_{oe} \subset \mathcal{V} \setminus \{s\}, \\
& |\beta_o| \le \ell_o, \, |\beta_e| \le \ell_e, \, |\beta_{oe}| \le \ell_{oe} 
\Big\}.
\end{aligned}
}
$$
Suppose that every node in the network is $d$-vertex connected from the source node $s$. Then, the key-capacity under the additive jamming model is given by
$$C_{\mathrm{add}} =
\begin{cases}
d-\ell_o-\ell_e-2\ell_{oe}, & \text{if } d>\ell_o+\ell_e+2\ell_{oe},\\[3ex]
0, & \text{otherwise,}
\end{cases}$$
and the key-capacity under the overwrite jamming model is given by$$C_{\mathrm{ow}} =
\begin{cases}
d-\ell_o-\ell_e-2\ell_{oe}, & \text{if } d>\ell_o+2\ell_e+2\ell_{oe},\\[3ex]
0, & \text{otherwise.}
\end{cases}$$
\end{theorem}
\begin{proof}[Proof sketch]
{We provide only a proof sketch here and defer the full proof to Appendix~\ref{T1Proof}.}
At a high level, our communication scheme applies the key-cast protocol of \cite{eavsPaper} designed for passive adversaries to provide perfect secrecy, combined with hash-based local error-correction applied at each network node to provide reliability. To distribute the necessary hash values securely, we apply Lemma~\ref{lem:parallel_edge_active_secure}. The primary challenge lies in proving that this combined scheme simultaneously preserves reliability and perfect secrecy against an active, node-based adversary.

We apply Algorithm~\ref{alg:keycast_active} with $d' = d - (\ell_e + \ell_{oe})$. 
{(Note that, to simplify our presentation, Algorithm~\ref{alg:keycast_active} is described in a sequential manner, while in our model of communication transmissions are done according to topological order in which all transmissions leaving a given node are done simultaneously.)} 
Each non-terminal node $v_i\in\mathcal{V}$ is assigned a Vandermonde vector $\boldsymbol{v}_i \in \mathbb{F}_q^{d'}$ indexed by $i$, and all terminals in the same set $\mathcal{T}_i\subseteq\mathcal{T}$ share the same Vandermonde vector $\boldsymbol{v}_{\mathcal{T}_i}$. For each terminal set $\mathcal{T}_i$, the key is 
$$\boldsymbol{K}_{\mathcal{T}_i} = \bigl[\boldsymbol{M}^{(1)} \boldsymbol{v}_{\mathcal{T}_i}, \dots, \boldsymbol{M}^{(n)} \boldsymbol{v}_{\mathcal{T}_i}\bigr]_{1:d'-(\ell_o+\ell_{oe})}\; .$$
During a pre-distribution phase, since every node is $d$-vertex connected from the source $s$, Lemma~\ref{lem:parallel_edge_active_secure} ensures that $s$ can securely transmit a uniformly distributed random value $\alpha_j$ and certain polynomial hashes $\{h_{j_p \to j}\}_{j_p\in \mathcal{N}(j)}$ to each node $j$.
Following the key-cast protocol in \cite{eavsPaper}, let $0\le b_j\le d$ be the number of edges from $s$ to $j$. The source $s$ sends $b_j$ vectors to node $j$ using a Vandermonde matrix $\boldsymbol{V}_{s\to j}\in\mathbb{F}_q^{d'\times b_j}$, whose columns are Vandermonde vectors constructed from distinct field elements unused elsewhere in the network. Node $j$ receives the remaining $d-b_j$ vectors from its non-source parents $j_p\in\mathcal{P}(j)$ inductively according to a topological order on the network. 
Node $j$ uses $\alpha_j$ to verify incoming vectors against the pre-distributed hashes and
corrupted transmissions are reliably detected and marked as erasures. 
Let $\boldsymbol{V}_j\in\mathbb{F}_q^{d'\times d}$ be the matrix whose first $b_j$ columns are $\boldsymbol{V}_{s\to j}$ and whose remaining columns are $\{\boldsymbol{v}_{j_p}\}_{j_p\in\mathcal{P}(j)}$. Due to our assumptions on $d$-vertex connectivity and the distinct nature of the Vandermonde vectors, any $d'$ columns of $\boldsymbol{V}_j$ are linearly independent. Therefore, $\boldsymbol{V}_j$ serves as the generator matrix of an MDS code. 
The minimum distance of this code is $d_{\min} = d - d' + 1 = \ell_e + \ell_{oe} + 1$. As corrupted messages are converted to erasures via our local hash checks, this distance guarantees the capability to correct up to $d_{\min} - 1 = \ell_e + \ell_{oe}$ erasures \cite{macwilliams1977theory}. 

To ensure reliability, we bound the overall probability of error.
The overall error probability consists of two components: the error introduced by applying Lemma~\ref{lem:parallel_edge_active_secure} at each node, and the error caused by failing to correctly identify corrupted edges within a node's parent set. 
Both errors can be made arbitrarily small using a sufficiently large field size (detailed derivations are provided in Appendix~\ref{T1Proof}).
Perfect security follows from a delicate analysis of (the lack of) information leakage when performing multiple rounds of the protocol described in the proof of Lemma~\ref{lem:parallel_edge_active_secure} (for the pre-distribution phase of Algorithm~\ref{alg:keycast_active}) followed by the remaining  steps of Algorithm~\ref{alg:keycast_active}. As each round of  Lemma~\ref{lem:parallel_edge_active_secure}'s protocol, and the remaining steps of Algorithm~\ref{alg:keycast_active}, use independent randomness for masking,  the combined procedure is shown to be perfectly-secure (details in Appendix~\ref{T1Proof}).
\end{proof}

{
{\begin{remark}\label{R0}
    The construction in Algorithm~\ref{alg:keycast_active} used in the proof of Theorem~\ref{T1} does not require the full network $\mathcal{G}$ itself to have every node $d$-vertex connected from the source. For a general network, it suffices that there exists a subgraph $\mathcal{G}' \subseteq \mathcal{G}$, containing the source, all terminal nodes, and possibly only a subset of the intermediate nodes, within which every terminal node and every intermediate node is $d$-vertex connected from the source. Whenever such a subgraph $\mathcal{G}'$ exists, the coding scheme summarized in Algorithm 1 can be run directly on $\mathcal{G}'$, using only its nodes and edges, and achieves a secure key rate of $d-\ell_o-\ell_e-2\ell_{oe}$. This observation extends the applicability of Theorem~\ref{T1} beyond networks that are globally $d$-vertex connected, to any network that merely admits such a $d$-vertex-connected subgraph.
\end{remark}}
}

\begin{remark}
It is worth noting that under a passive adversary model, the execution of Algorithm~\ref{alg:keycast_active} for a single terminal set is equivalent to the network secret-sharing protocol presented in \cite{shah2015distributed}. Crucially, the additional pre-distribution and polynomial hashing mechanisms introduced in our algorithm seamlessly extend to this secret sharing framework. Therefore, all results derived herein, including the active security guarantees, maximum decoding error probabilities, and achievable rates established in Theorem~\ref{T1}, hold for network secret sharing against active adversaries over $d$-vertex connected networks.
\end{remark}

\begin{corollary}[Multiple-source setting]\label{C2}
    Consider an instance $\mathcal{I}=\bigl(\mathcal{V},\mathcal{E},\mathcal{B},\mathcal{S}=\{s_1,\dots,s_{|\mathcal{S}|}\},\mathcal{T}\bigr)$ of the secure multiple key-cast problem.
    For any integer $x < |\mathcal{S}|$, let the adversary structure be defined as
    $$
    {
    \begin{aligned}
    \mathcal{B} = \Big\{ &(\beta_o, \beta_e, \beta_{oe}) \mid \text{disjoint}\ \beta_o, \beta_e, \beta_{oe} \subset \mathcal{V},  |\beta_o \setminus \mathcal{S}| \le \ell_o,\\
    &  |\beta_e| \le \ell_e, \, |\beta_{oe}| \le \ell_{oe}, \, 
    |\beta_o \cap \mathcal{S}| \le x, \, \beta_{\mathrm{jam}} \cap \mathcal{S} = \emptyset
    \Big\}.
    \end{aligned}
    }
    $$
    where $x$ denotes the maximum number of source nodes that can be compromised by the adversary for eavesdropping.
    If every node in the network is $d$-vertex connected from each source $s \in \mathcal{S}$, then the network key-capacity under the additive jamming model is given by
    $$
    C_{\mathrm{add}} =
    \begin{cases}
    \frac{|\mathcal{S}|-x}{|\mathcal{S}|} \bigl(d-\ell_o-\ell_e-2\ell_{oe}\bigr), & \text{if } d>\ell_o+\ell_e+2\ell_{oe},\\[1.5ex]
    0, & \text{otherwise,}
    \end{cases}
    $$
    and the network key-capacity under the overwrite jamming model is given by
    $$
    C_{\mathrm{ow}} =
    \begin{cases}
    \frac{|\mathcal{S}|-x}{|\mathcal{S}|} \bigl(d-\ell_o-\ell_e-2\ell_{oe}\bigr), & \text{if } d>\ell_o+2\ell_e+2\ell_{oe},\\[1.5ex]
    0, & \text{otherwise.}
    \end{cases}
    $$
\end{corollary}
\begin{proof}[Proof sketch]
The proof follows \cite[Corollary 1]{eavsPaper} and involves multiple executions of the algorithm of Theorem~\ref{T1}. 
\end{proof}

\section{Multiple Key-Cast in Networks with  Partially-Connected Nodes}\label{Pcon}
In this section, we relax the $d$-vertex connectivity assumption considered previously and study multiple key-cast in networks where only terminal nodes are guaranteed to be $d$-vertex connected from the source. Intermediate nodes that do not satisfy this connectivity requirement are referred to as partially connected nodes; that is, nodes that are less than $d$-vertex connected from the source. These nodes have limited connectivity and thus restrict the flow of independent information from the source.
We show that, despite the presence of partially connected nodes, secure multiple key-cast remains feasible under a structural condition on the network. Let $z$ denote the maximum number of partially connected nodes that provide input to any $d$-vertex connected node. We establish that if $z \le d-\ell_o-\ell_e-2\ell_{oe}$, then the network sustains a positive secure key-rate. Building on Lemma~\ref{lem:parallel_edge_active_secure} and ideas in \cite[Theorem~2]{eavsPaper}, we construct a single-source coding scheme that guarantees reliability and perfect secrecy against active adversaries.
Finally, we extend these results to the multiple-source setting.

A key structural insight used in our analysis is provided by Lemma~\ref{lem:subnetwork-indegree}, which shows that any DAG contains a subgraph in which each node's in-degree equals its vertex connectivity from the source; this allows us to apply our coding scheme on a reduced subgraph without loss of connectivity.
We introduce the notations used in the lemma below.
For a directed acyclic network $\mathcal{G}=(\mathcal{V},\mathcal{E})$ with source $s\in\mathcal{V}$
and a node $v\in\mathcal{V}\setminus\{s\}$:
 
\begin{itemize}
  \item $|\mathcal{P}_{\mathcal{G}}(v)|$ denotes the \emph{in-degree} of $v$ in
        $\mathcal{G}$, i.e., the number of directed edges entering $v$,
        where $\mathcal{P}_{\mathcal{G}}(v)$ is the set of parent nodes of $v$
        in $\mathcal{G}$.
 
  \item $\kappa_{\mathcal{G}}(s,v)$ denotes the \emph{vertex connectivity} from
        $s$ to $v$, defined as the maximum number of
        vertex-disjoint paths from $s$ to $v$ in $\mathcal{G}$.
\end{itemize}

\begin{lemma}[Modification of {\cite[Theorem 2]{lovasz1973connectivity}}] \label{lem:subnetwork-indegree}
For any directed acyclic graph $\mathcal{G}=(\mathcal{V},\mathcal{E})$ with source
$s\in\mathcal{V}$, there exists a subgraph
$\mathcal{G}'=(\mathcal{V},\mathcal{E}')$ with $\mathcal{E}'\subseteq\mathcal{E}$
such that, for every node $v\in\mathcal{V}\setminus\{s\}$,
\[
  |\mathcal{P}_{\mathcal{G}'}(v)|
  \;=\;
  \kappa_{\mathcal{G}'}(s,v)
  \;=\;
  \kappa_{\mathcal{G}}(s,v).
\]
\end{lemma}
\begin{proof}
    The proof closely follows \cite[Theorem 2]{lovasz1973connectivity}, adapted to vertex-based connectivity via a node-splitting reduction.
    The full proof is provided in Appendix~\ref{L1Proof}.
\end{proof}

\begin{theorem}[Partially-connected, single-source setting]\label{T3-active}
Consider an instance $\mathcal{I}=\bigl(\mathcal{V},\mathcal{E},\mathcal{B},\mathcal{S}=\{s\},\mathcal{T}\bigr)$ of the Secure Multiple Key-Cast problem with 
$$
{
\begin{aligned}
\mathcal{B} = \Big\{ &(\beta_o, \beta_e, \beta_{oe}) \mid \text{disjoint}\ \beta_o, \beta_e, \beta_{oe} \subset \mathcal{V} \setminus \{s\}, \\
& |\beta_o| \le \ell_o, \, |\beta_e| \le \ell_e, \, |\beta_{oe}| \le \ell_{oe}
\Big\},
\end{aligned}
}
$$
in which all terminal nodes are $d$-vertex connected from the source node $s$.
Let $L$ denote the maximum number of partially-connected nodes forming a contiguous directed path within the network. 
Suppose no $d$-vertex connected node receives input from more than $z$ partially-connected nodes.
Then, under the condition that $z \le d-\ell_o-\ell_e-2\ell_{oe}$, the network key-capacity under the additive jamming model satisfies
$$
C_{\mathrm{add}} \ge
\begin{cases}
\frac{d-\ell_o-\ell_e-2\ell_{oe}-z+1}{d'(d-1)^{L-1}+1}, & \text{if } d>\ell_o+\ell_e+2\ell_{oe},\\[1.5ex]
0, & \text{otherwise,}
\end{cases}
$$
where $d'=d - \ell_e - \ell_{oe}$, and the network key-capacity under the overwrite jamming model satisfies
$$
C_{\mathrm{ow}} \ge
\begin{cases}
\frac{d-\ell_o-\ell_e-2\ell_{oe}-z+1}{d'(d-1)^{L-1}+1}, & \text{if } d>\ell_o+2\ell_e+2\ell_{oe},\\[1.5ex]
0, & \text{otherwise.}
\end{cases}
$$
\end{theorem}
\begin{proof}[Proof sketch]
{We provide only a proof sketch here and defer the full proof to Appendix~\ref{T2Proof}.}
The proposed scheme builds on the approach of \cite{eavsPaper} for handling partially-connected nodes, combined with the active-adversary protection mechanism of Theorem~\ref{T1}. Specifically, two parallel instances of Algorithm~\ref{alg:keycast_active} are executed, while verification hashes and auxiliary randomness are securely pre-distributed to the collection $\mathcal{V}^{\mathrm{Full}} \subset \mathcal{V}$ of $d$-vertex connected nodes. Since Lemma~\ref{lem:parallel_edge_active_secure} requires $d$-vertex connectivity from the source, only these nodes can perform local error detection and MDS-based error correction.
Similar to \cite{eavsPaper}, the first parallel instance of Algorithm~\ref{alg:keycast_active} operates only on $\mathcal{V}^{\mathrm{Full}}$, treating inputs from partially-connected nodes as erasures, while the second instance utilizes the entire network. The final key is obtained by securely combining the outputs of the two instances. The principal difference from the passive setting of \cite{eavsPaper} is that partially-connected nodes cannot reliably decode and reconstruct local shares, as they do not possess the verification information required to detect adversarial modifications. Consequently, such nodes must forward all linearly independent symbols they receive.
For a tree of depth $L$ consisting solely of partially-connected nodes, this forwarding requirement causes the communication overhead to grow as $d'(d-1)^{L-1}$, yielding the rate expression stated in the theorem. Reliability follows from the same hashing-based verification arguments used in Theorem~\ref{T1}, where the probability of an undetected forgery or hash collision can be made arbitrarily small by choosing sufficiently large field sizes, while secrecy follows by combining the security analysis of \cite{eavsPaper} together with the secure pre-distribution phase. 
\end{proof}

{
\begin{remark}[Choice of the design parameter $d$ and of the subgraph $\mathcal{G}'$]\label{R1}
Note that the parameter $d$ is a design choice, whereas the connectivity parameters below are determined by the network. Moreover, by Remark~\ref{R0}, the schemes of Theorem~\ref{T1} and Theorem~\ref{T3-active} need not be run on the full network $\mathcal{G}=(\mathcal{V},\mathcal{E})$. For any subgraph $\mathcal{G}'=(\mathcal{V}',\mathcal{E}')\subseteq\mathcal{G}$ that contains the source $s$, all terminal nodes, and any chosen subset of the remaining (intermediate) nodes: Theorem~\ref{T1} applies directly to $\mathcal{G}'$ whenever every node of $\mathcal{G}'$ is $d$-vertex connected from $s$ within $\mathcal{G}'$, and Theorem~\ref{T3-active} applies whenever only the terminals of $\mathcal{G}'$ are guaranteed to be $d$-vertex connected within $\mathcal{G}'$, with the structural condition of Theorem~\ref{T3-active} holding for the remaining, possibly partially-connected, nodes. We therefore phrase the discussion below for a general such $\mathcal{G}'$, with the understanding that $\mathcal{G}'=\mathcal{G}$ recovers the setting considered so far.

For a fixed choice of $\mathcal{G}'$, define the following quantities, all evaluated within $\mathcal{G}'$:
\begin{itemize}
    \item $\hat d_{\mathcal{G}'}$, the minimum vertex-connectivity from $s$ taken over all nodes of $\mathcal{G}'$;
    \item $\Omega_{\mathcal{G}'}(d)$, the set of partially-connected nodes of $\mathcal{G}'$:
    \[
    \{v\in\mathcal{V}':\ v\ \text{is less than}\ d\text{-vertex connected from}\ s\ \text{in}\ \mathcal{G}'\};
    \]
    \item $z_{\mathcal{G}'}(d)$, the largest number of partially-connected parents (i.e., parents in $\Omega_{\mathcal{G}'}(d)$) of any $d$-vertex connected node of $\mathcal{G}'$;
    \item $L_{\mathcal{G}'}(d)$, the maximum number of partially-connected nodes forming a contiguous directed path within $\mathcal{G}'$;
    \item $d_{\mathcal{T}}(\mathcal{G}')$, the minimum vertex-connectivity from $s$ within $\mathcal{G}'$ taken over the terminal nodes.
\end{itemize}
Increasing $d$ enlarges the set $\Omega_{\mathcal{G}'}(d)$, and thus, as $d$ grows, a $d$-vertex connected node of $\mathcal{G}'$ may acquire additional partially-connected parents, i.e., $z_{\mathcal{G}'}(d)$ may grow with $d$. Similarly, $L_{\mathcal{G}'}(d)$ may grow with $d$.

Since all terminals are required to be $d$-vertex connected from $s$ within $\mathcal{G}'$, the admissible values of $d$ satisfy
\begin{equation}\label{Ceq1}
 \ell_o+\ell_e+2\ell_{oe}
+\max\{z_{\mathcal{G}'}(d),1\}
\le d\le d_{\mathcal{T}}(\mathcal{G}')
\end{equation}
for the additive error case, and
\begin{equation}\label{Ceq2}
\max\bigl\{
\ell_o+\ell_e+2\ell_{oe}+z_{\mathcal{G}'}(d),
\ell_o+2\ell_e+2\ell_{oe}+1
\bigr\}
\le d\le d_{\mathcal{T}}(\mathcal{G}')    
\end{equation}
for the overwrite error case. These lower bounds combine the structural condition
$z_{\mathcal{G}'}(d)\le d-\ell_o-\ell_e-2\ell_{oe}$ of Theorem~\ref{T3-active} (applied to $\mathcal{G}'$) with the corresponding connectivity requirement of the theorem, and guarantee a positive numerator in the rate expression below.

For fixed $\mathcal{G}'$, the rate of Theorem~\ref{T3-active} is therefore to be read as a function of $d$, namely
\[
R_{\mathcal{G}'}(d)\triangleq
\frac{d-\ell_o-\ell_e-2\ell_{oe}-z_{\mathcal{G}'}(d)+1}
{\bigl(d-\ell_e-\ell_{oe}\bigr)(d-1)^{L_{\mathcal{G}'}(d)-1}+1},
\]
and, by Remark~\ref{R0}, when partially-connected nodes are present, the best rate guaranteed by Theorem~\ref{T3-active} is obtained by jointly optimizing $R_{\mathcal{G}'}(d)$ over the admissible range of $d$ \emph{and} over the choice of subgraph $\mathcal{G}'$ itself, i.e., by evaluating
\[
R^{\star}\triangleq
\max_{\mathcal{G}'\in\mathcal{A}}\ \max_d\ R_{\mathcal{G}'}(d),
\]
where $\mathcal{A}$ denotes the family of \emph{admissible} subgraphs, i.e., those $\mathcal{G}'\subseteq\mathcal{G}$ that contain the source $s$ (or all sources, in the multi-source case) and all terminal nodes, and for which some choice of $d$ makes Theorem~\ref{T1} or Theorem~\ref{T3-active} applicable to $\mathcal{G}'$, as described above. For each $\mathcal{G}'\in\mathcal{A}$, the inner maximization over $d$ is subject to \eqref{Ceq1} in the additive error case, and \eqref{Ceq2} in the overwrite error case.

Two regimes arise, according to whether $\mathcal{A}$ contains a subgraph whose minimum vertex-connectivity exceeds the corresponding connectivity threshold of Theorem~\ref{T1}.

\emph{(i) There exists $\mathcal{G}'\in\mathcal{A}$ with
$\hat d_{\mathcal{G}'}>\ell_o+\ell_e+2\ell_{oe}$ in the additive error case, or
$\hat d_{\mathcal{G}'}>\ell_o+2\ell_e+2\ell_{oe}$ in the overwrite error case.}
We first observe that partially-connected nodes come at a cost in rate. For any $\mathcal{G}'\in\mathcal{A}$ and any $d>\hat d_{\mathcal{G}'}$, if $z_{\mathcal{G}'}(d)\ge1$, then $L_{\mathcal{G}'}(d)\ge1$ and Theorem~\ref{T3-active} gives
\begin{align*}
R_{\mathcal{G}'}(d)
&=
\frac{d-\ell_o-\ell_e-2\ell_{oe}-z_{\mathcal{G}'}(d)+1}
{\bigl(d-\ell_e-\ell_{oe}\bigr)
(d-1)^{L_{\mathcal{G}'}(d)-1}+1}\\
&\le
 \frac{d-\ell_o-\ell_e-2\ell_{oe}-z_{\mathcal G'}(d)+1} {d-\ell_e-\ell_{oe}+1}\\ 
 &\le 1,
\end{align*}
where the inequality uses $z_{\mathcal{G}'}(d)\ge1$, $\ell_o\ge0$, $\ell_{oe}\ge0$ and $L_{\mathcal{G}'}(d)\ge1$. Thus, whenever $z_{\mathcal{G}'}(d)\ge1$, the resulting rate is less than or equal to $1$.

It is therefore preferable to avoid partially-connected nodes altogether, which for a given $\mathcal{G}'$ means restricting to $d\le\hat d_{\mathcal{G}'}$, so that Theorem~\ref{T1} applies and yields the rate
$d-\ell_o-\ell_e-2\ell_{oe}$. Since this rate increases with $d$, the best choice for a given $\mathcal{G}'$ is the largest $d$ for which Theorem~\ref{T1} still applies, namely $d=\hat d_{\mathcal{G}'}$, giving rate
$\hat d_{\mathcal{G}'}-\ell_o-\ell_e-2\ell_{oe}$. This value is at least $1$ whenever $\hat d_{\mathcal{G}'}$ satisfies the corresponding connectivity condition of Theorem~\ref{T1}, and hence is no smaller than the rate obtained for any $d>\hat d_{\mathcal{G}'}$ and $z_{\mathcal{G}'}(d)\ge1$.

Optimizing over $\mathcal{G}'$ as well, the best rate is thus attained by the subgraph with the largest such $\hat d_{\mathcal{G}'}$: setting
\[
\mathcal{G}^{\star}\in
\begin{cases}
\displaystyle
\arg\max_{\mathcal{G}'\in\mathcal{A}:\,
\hat d_{\mathcal{G}'}>\ell_o+\ell_e+2\ell_{oe}}
\hat d_{\mathcal{G}'},
& \text{additive error case},\\[3ex]
\displaystyle
\arg\max_{\mathcal{G}'\in\mathcal{A}:\,
\hat d_{\mathcal{G}'}>\ell_o+2\ell_e+2\ell_{oe}}
\hat d_{\mathcal{G}'},
& \text{overwrite error case},
\end{cases}
\]
\[
d^{\star}=\hat d_{\mathcal{G}^{\star}},
\]
achieves
\[
R^{\star}
=
\hat d_{\mathcal{G}^{\star}}
-\ell_o-\ell_e-2\ell_{oe}.
\]
By the converse underlying Theorem~\ref{T1} (see Remark~\ref{rem:asymptotic_optimality}), this rate cannot be improved in general: there exist $d$-vertex connected instances whose secure key-cast capacity is at most $d-\ell_o-\ell_e-2\ell_{oe}$, so the bound
$\hat d_{\mathcal{G}^{\star}}-\ell_o-\ell_e-2\ell_{oe}$ established above cannot be improved in general based solely on this connectivity condition.

\emph{(ii) Every $\mathcal{G}'\in\mathcal{A}$ satisfies
$\hat d_{\mathcal{G}'}\le\ell_o+\ell_e+2\ell_{oe}$ in the additive error case, or
$\hat d_{\mathcal{G}'}\le\ell_o+2\ell_e+2\ell_{oe}$ in the overwrite error case.}
Here, for every $\mathcal{G}'\in\mathcal{A}$, every admissible $d$ is strictly larger than the corresponding connectivity threshold and hence satisfies $d>\hat d_{\mathcal{G}'}$. Therefore, at least one node attaining the minimum connectivity $\hat d_{\mathcal{G}'}$ is necessarily partially-connected, and Theorem~\ref{T1} is not applicable on any admissible subgraph for any choice of $d$. The best rate guaranteed by Theorem~\ref{T3-active} is then the value of the joint discrete optimization problem
\[
\begin{aligned}
R^{\star}&=\max_{\mathcal{G}'\in\mathcal{A}}\max_d\quad 
\frac{d-\ell_o-\ell_e-2\ell_{oe}-z_{\mathcal{G}'}(d)+1}
{\bigl(d-\ell_e-\ell_{oe}\bigr)
(d-1)^{L_{\mathcal{G}'}(d)-1}+1}\\[0.2em]
\text{s.t.}\quad &
\begin{cases}
\eqref{Ceq1},
& \text{additive},\\[1.5ex]
\eqref{Ceq2},
& \text{overwrite}.
\end{cases}
\end{aligned}
\]
The rate is not monotone in $d$, since increasing $d$ changes the denominator both directly and through $L_{\mathcal{G}'}(d)$, while the numerator is governed by $d-z_{\mathcal{G}'}(d)$, and neither the smallest nor the largest admissible $d$ need be optimal. For instance, in the additive error case, consider a subgraph $\mathcal{G}'$ with $\hat d_{\mathcal{G}'}=1$, $\ell_o=2$, $\ell_e=\ell_{oe}=0$, $d_{\mathcal{T}}(\mathcal{G}')\ge5$, $L_{\mathcal{G}'}(d)=1$ for $d\in\{3,4,5\}$, and
\[
z_{\mathcal{G}'}(3)=z_{\mathcal{G}'}(4)=1,
\qquad
z_{\mathcal{G}'}(5)=2.
\]
The corresponding rates are $1/4$, $2/5$, and $1/3$, respectively, so that neither the smallest admissible choice $d=3$ nor the largest choice $d=5$ is optimal, and $d=4$ is preferable. Since the choice of $\mathcal{G}'$ likewise affects $\hat d_{\mathcal{G}'}$, $z_{\mathcal{G}'}(d)$, $L_{\mathcal{G}'}(d)$, and $d_{\mathcal{T}}(\mathcal{G}')$, the two maximizations must be carried out jointly. We leave a full characterization of this joint optimization, including efficient algorithms for identifying an optimal subgraph $\mathcal{G}'$, to future work.
\end{remark}
}

\begin{corollary}[Partially-connected, multiple-source setting]\label{C3-active}
    Consider an instance $\mathcal{I}=\bigl(\mathcal{V},\mathcal{E},\mathcal{B},\mathcal{S}=\{s_1,\dots,s_{|\mathcal{S}|}\},\mathcal{T}\bigr)$ of the secure multiple key-cast problem.
    For any integer $x < |\mathcal{S}|$, let the adversary structure be defined as
    $$
    {
    \begin{aligned}
    \mathcal{B} = \Big\{ &(\beta_o, \beta_e, \beta_{oe}) \mid \text{disjoint}\ \beta_o, \beta_e, \beta_{oe} \subset \mathcal{V},\ 
    |\beta_o \setminus \mathcal{S}| \le \ell_o,\\
    & |\beta_e| \le \ell_e, \, |\beta_{oe}| \le \ell_{oe}, \, 
    |\beta_o \cap \mathcal{S}| \le x, \, \beta_{\mathrm{jam}} \cap \mathcal{S} = \emptyset
    \Big\},
    \end{aligned}
    }
    $$
    where $x$ denotes the maximum number of source nodes that can be compromised by the adversary for eavesdropping.
    Suppose all terminal nodes are $d$-vertex connected from each source node $s \in \mathcal{S}$. Let $L$ denote the maximum number of partially-connected nodes forming a contiguous directed path, and suppose no $d$-vertex connected node receives input from more than $z$ partially-connected nodes. If $z \le d-\ell_o-\ell_e-2\ell_{oe}$, then the network key-capacity under the additive jamming model satisfies
    $$
    C_{\mathrm{add}} \ge
    \begin{cases}
    \frac{|\mathcal{S}|-x}{|\mathcal{S}|} \left( \frac{d-\ell_o-\ell_e-2\ell_{oe}-z+1}{d'(d-1)^{L-1}+1} \right), & \text{if } d>\ell_o+\ell_e+2\ell_{oe},\\[1.5ex]
    0, & \text{otherwise,}
    \end{cases}
    $$
    where $d'=d - \ell_e - \ell_{oe}$, and the network key-capacity under the overwrite jamming model satisfies
    $$
    C_{\mathrm{ow}} \ge
    \begin{cases}
    \frac{|\mathcal{S}|-x}{|\mathcal{S}|} \left( \frac{d-\ell_o-\ell_e-2\ell_{oe}-z+1}{d'(d-1)^{L-1}+1} \right), & \text{if } d>\ell_o+2\ell_e+2\ell_{oe},\\[1.5ex]
    0, & \text{otherwise.}
    \end{cases}
    $$
\end{corollary}
\begin{proof}[Proof sketch]
The proof follows \cite[Corollary 1]{eavsPaper} and involves multiple executions of the algorithm of Theorem~\ref{T3-active}. 
\end{proof}

\section{Conclusion}\label{Con}
In this paper, we studied secure multiple key-cast over noiseless networks under an active, node-based adversarial model. For the single-source setting, we showed that in networks where every node is $d$-vertex connected from the source, the network key-capacity is $d-\ell_o-\ell_e-2\ell_{oe}$. We further extended these results to networks containing partially connected nodes and to the multiple-source setting, including scenarios in which the adversary may observe source nodes. The proposed constructions also naturally extend to network secret sharing and secure multicast network coding. The active, node-based adversarial model differs fundamentally from both edge-based error models and passive eavesdropping models, introducing unique challenges due to the adversary's ability to both observe and corrupt transmissions through compromised nodes (which potentially control several outgoing edges).
Several research directions remain open. As established in Theorem~\ref{T3-active}, long chains of partially connected nodes cause the transmission overhead to grow exponentially, significantly reducing the achievable key-capacity; overcoming this limitation remains an important open problem. Another challenging direction is to remove the structural assumptions imposed on partially connected nodes (represented by the parameter $z$) in Theorem~\ref{T3-active} and derive matching converse bounds for general network topologies. We note that, while the capacity lower bound in Theorem~\ref{T3-active} decrease as $z$ and $L$ increase, the achievable rate remains sufficient for many cryptographic applications that require only constant-sized secret keys.
{Finally, establishing results of similar nature in the presence of node-based adversaries that design their corruptions using information available at any nodes in 
the set $\beta_{\mathrm{obs}}$ irrelevant of their topological order 
is subject to future work.}

\section{Acknowledgments}
This work is supported in part by NSF grant CCF-2245204.
\appendices

\section{Proof of Lemma \ref{lem:parallel_edge_active_secure}}\label{LemmaProof}
\begin{proof}
We now provide a complete proof of Lemma \ref{lem:parallel_edge_active_secure}.
We first observe that a node-based adversary operating on $d$ vertex-disjoint paths possesses the exact same capabilities as an edge-based adversary attacking $d$ parallel links. This equivalence allows us to directly apply the parallel-edge framework studied in \cite{zhang2023optimal}. Our proof follows the same overall method as the proofs of \cite[Theorems 1 and 2]{zhang2023optimal}, utilizing standard MDS codes for the underlying data transmission. However, we introduce a modified polynomial hashing mechanism. This specific refinement to the hashing function is what upgrades the asymptotic weak secrecy guarantees of their model to perfect secrecy, ensuring the eavesdropper gains zero information.
\subsection*{Encoding:}
Consider a single-source, single-terminal network in which the terminal node $t$ is $d$-vertex connected from the source $s$. Specifically, we specify $d$ vertex-disjoint paths from $s$ to $t$. Because the intermediate nodes along these paths simply forward the information they receive, an adversary located at any intermediate node within a given path does not gain any additional eavesdropping or error-injection capabilities. Consequently, it suffices to consider the worst-case scenario in which the adversary takes action directly upon the final parent node of each of the $d$ paths.
Let $d' = d - \ell_e - \ell_{oe}$, and let $\ell' = \ell_o + \ell_{oe}$ be the maximum number of eavesdropped nodes. Suppose that the source wishes to securely transmit a message matrix $\boldsymbol{X}\in\mathbb{F}_q^{n\times(d'-\ell')}$ to $t$ over $d$ vertex-disjoint paths. To ensure secrecy, the source generates a random matrix $\boldsymbol{R}\in\mathbb{F}_q^{n\times\ell'}$ whose entries are independent and uniformly distributed over $\mathbb{F}_q$. Let $\boldsymbol{V} \in \mathbb{F}_q^{d' \times d}$ be a Vandermonde matrix constructed from distinct Vandermonde indices, so that any $d'$ columns of $\boldsymbol{V}$ are linearly independent.
The source node generates $\boldsymbol{U} \in \mathbb{F}_q^{n \times d}$ as
\[
\boldsymbol{U} = \begin{bmatrix} \boldsymbol{X} & \boldsymbol{R} \end{bmatrix} \boldsymbol{V}.
\]
For each vertex-disjoint path $i \in \{1,\ldots,d\}$, the source transmits the vector
\[
\boldsymbol{t}_i=
\begin{bmatrix}
\boldsymbol{u}_i^T &
\boldsymbol{\alpha}_{i}^T &
\boldsymbol{\gamma}_{i}^T &
\boldsymbol{h}_i^T
\end{bmatrix}^{T},
\]
where $\boldsymbol{u}_i$ is the $i$-th column of $\boldsymbol{U}$, $\boldsymbol{\alpha}_{i}, \boldsymbol{\gamma}_{i} \in \mathbb{F}_q^{d}$ are random vectors with independent uniformly distributed entries over $\mathbb{F}_q$, whose $j$-th components are denoted by $\alpha_{i,j}$ and $\gamma_{i,j}$, respectively, and $\boldsymbol{h}_i=[h_{i,1},\ldots,h_{i,d}]^{T}\in\mathbb{F}_q^{d}$ is a hash vector whose $j$-th component is defined by
\begin{equation}\label{eq-verify-M}
\begin{aligned}
h_{ij}
&= f\left(\boldsymbol{u}_j, \alpha_{i,j}+\alpha_{j,i}, \gamma_{i,j}+\gamma_{j,i}\right) \\
&\triangleq
\sum_{\tau=1}^{n}
[\boldsymbol{u}_{j}]_{\tau}
\bigl(\alpha_{i,j}+\alpha_{j,i}\bigr)^{\tau}
+\bigl(\alpha_{i,j}+\alpha_{j,i}\bigr)^{n+2} \\
&\quad
+\gamma_{i,j}+\gamma_{j,i} .
\end{aligned}
\end{equation}
\subsection*{Decoding:}
At the destination, the terminal receives $\hat{\boldsymbol{t}}_i$ from the $i$-th parent node and performs two layers of consistency checks to isolate uncorrupted paths:
\subsubsection{Self-Consistency Check}
In this step, the terminal identifies and isolates transmissions from parent nodes in $\beta_e$ that inject errors without being observed by the adversary.

Consider $i \notin \beta_{\mathrm{obs}}$.
For an \underline{additive} jammer, the vector received from the $i$-th parent node can be expressed as:
$$\hat{\boldsymbol{t}}_i = \begin{bmatrix}
\boldsymbol{u}_i^T + {\boldsymbol{u}_i'}^T &
\boldsymbol{\alpha}_{i}^T + {\boldsymbol{\alpha}_{i}'}^T &
\boldsymbol{\gamma}_{i}^T + {\boldsymbol{\gamma}_{i}'}^T &
\boldsymbol{h}_i^T + {\boldsymbol{h}_i'}^T
\end{bmatrix}^{T},$$
where $\boldsymbol{u}'_i$, $\boldsymbol{\alpha}_{i}'$, $\boldsymbol{\gamma}_{i}'$, and $\boldsymbol{h}_i'$ for $i \in [d]$ denote the possible additive errors injected at the $i$-th parent node into $\boldsymbol{u}_i$, $\boldsymbol{\alpha}_{i}$, $\boldsymbol{\gamma}_{i}$, and $\boldsymbol{h}_i$, respectively.
The terminal verifies that the received components satisfy the hash relation. Failure of
this verification identifies the corresponding parent node as belonging to
$\beta_{\mathrm{jam}}$, and its transmitted symbols are discarded.

Let the $j$-th components of $\boldsymbol{h}_i'$, $\boldsymbol{\alpha}_i'$, and $\boldsymbol{\gamma}_i'$ be denoted by $h'_{i,j}$, $\alpha'_{i,j}$, and $\gamma'_{i,j}$, respectively.
For every $i \in [d]$ (where $i=j$), the terminal checks:
$$\begin{aligned}
h_{i,i} + h'_{i,i} &= \sum_{\tau=1}^{n} [\boldsymbol{u}_i + \boldsymbol{u}'_i]_{\tau} \Bigl(2\alpha_{i,i} + 2\alpha_{i,i}'\Bigr)^{\tau} \\
& \quad + \Bigl(2\alpha_{i,i} + 2\alpha_{i,i}'\Bigr)^{n+2} + 2\gamma_{i,i} + 2\gamma_{i,i}'.
\end{aligned}$$
If this equality does not hold, then the transmission from parent node $i$ is identified as erroneous and treated as an erasure.
Therefore, the adversary succeeds in bypassing this check if and only if:
\begin{equation}\label{eq:hprime}
\begin{aligned}
h'_{ii} &= \Bigl((2\alpha_{i,i} + 2\alpha_{i,i}')^{n+2} - (2\alpha_{i,i})^{n+2}\Bigr) \\
&\quad + \sum_{\tau=1}^{n} [\boldsymbol{u}_i]_\tau \Bigl((2\alpha_{i,i} + 2\alpha_{i,i}')^\tau - (2\alpha_{i,i})^\tau\Bigr) \\
&\quad + \sum_{\tau=1}^{n} [\boldsymbol{u}'_i]_\tau (2\alpha_{i,i} + 2\alpha_{i,i}')^\tau + 2\gamma_{i,i}'.
\end{aligned}
\end{equation}
We only consider the case where $\boldsymbol{u}'_i \neq \boldsymbol{0}$, since if $\boldsymbol{u}'_i = \boldsymbol{0}$, the adversary does not alter the data payload and hence induces no harmful corruption.
Let the adversary's observation set be denoted by $\{\boldsymbol{t}_j\}_{j \in \beta_{\mathrm{obs}}}$. By assumption, $i \notin \beta_{\mathrm{obs}}$, since the adversary can only inject errors at this parent node and not observe it.
Therefore, $\alpha_{i,i}$ and $\gamma_{i,i}$ are uniformly distributed and independent of the adversary's observations.
Consequently, the injected error vectors $\boldsymbol{u}'_i$, $\boldsymbol{\alpha}'_{i}$,
$\boldsymbol{\gamma}'_{i}$, and $\boldsymbol{h}'_{i}$ must be chosen without any knowledge
of $\alpha_{i,i}$ and $\gamma_{i,i}$.

First, suppose that $\alpha'_{i,i} = 0$.
Then \eqref{eq:hprime} collapses to:
\[
h_{i,i}' = \sum_{\tau=1}^{n} [\boldsymbol{u}'_i]_\tau (2\alpha_{i,i})^\tau + 2\gamma'_{i,i}.
\]
Since $\boldsymbol{u}'_i \neq \boldsymbol{0}$, the right-hand side is a nonzero polynomial in $2\alpha_{i,i}$ of degree at most $n$ (where $2\gamma'_{i,i}$ simply acts as a constant offset). Therefore, for any fixed $h_{i,i}'$ and $\gamma'_{i,i}$, this equality can hold for at most $n$ values of $\alpha_{i,i}$. Hence,
\[
\Pr(\text{Self-Consistency Check passes} \mid i \in \beta_{e},\; \alpha'_{i,i} = 0) \le \frac{n}{q}.
\]
Now suppose that $\alpha'_{i,i} \neq 0$.
Look at the right-hand side of \eqref{eq:hprime} as a polynomial in $2\alpha_{i,i}$. The term of highest degree $n+1$ comes exclusively from the expansion of $((2\alpha_{i,i} + 2\alpha'_{i,i})^{n+2} - (2\alpha_{i,i})^{n+2})$. The leading term of degree $n+1$ resulting from this expression is:
\[
(n+2) 2\alpha'_{i,i} (2\alpha_{i,i})^{n+1}.
\]
Since $\alpha'_{i,i} \neq 0$, the coefficient $(n+2) 2\alpha'_{i,i}$ is strictly non-zero.
The remaining sums in \eqref{eq:hprime} have a maximum degree of $n$, and $2\gamma'_{i,i}$ remains a degree-zero constant. Therefore, the right-hand side of \eqref{eq:hprime} is guaranteed to be a strictly non-zero polynomial in $2\alpha_{i,i}$ of degree exactly $n+1$. For any fixed $h_{i,i}'$, this equality can hold for at most $n+1$ values of $\alpha_{i,i}$. Hence,
\[\small
\Pr(\text{Self-Consistency Check passes} \mid i \in \beta_{e},\;\alpha'_{i,i} \neq 0) \le \frac{n+1}{q}.
\]
Combining both cases, the probability that a harmful corruption from parent node $i \in \beta_e$ passes the hash check is bounded by:
\[
\Pr(\text{Self-Consistency Check passes} \mid i \in \beta_{e}) \le \frac{n+1}{q}.
\]
In the case of an \underline{overwrite} jammer, the adversary can arbitrarily choose any value for the received vector. Consequently, it can always ensure that the self-consistency check fails to identify the corruption.
\subsubsection{Cross-Consistency Check}
For any two self-consistent transmissions from parent nodes $i$ and $j$, the terminal
checks if they are pairwise-consistent. They are pairwise-consistent if and only if the
hash constraints hold across the distinct received transmissions. In practice, because the
terminal does not know \textit{a priori} which parent nodes belong to
$\beta_{\mathrm{jam}}$, it must execute these verification checks exhaustively for every
pair of parent nodes in both directions.
However, for the purpose of analytically bounding the probability of a successful attack, it is sufficient to analyze the scenario involving an uncorrupted transmission from parent $i \notin (\beta_{\mathrm{jam}}\cup \beta_{\mathrm{obs}})$ and a parent node $j \in \beta_{\mathrm{jam}}$.
For the pair to be deemed consistent, both directional checks must pass.

In the additive-error model, all parent nodes in $\beta_e$ are already treated as erasures by the self-consistency check, and therefore, any parent node that remains active after the self-consistency check must satisfy $j \in \beta_{oe}$. In contrast, under the overwrite-error model, any parent node in $\beta_{\mathrm{jam}}$ may be fully controlled by the adversary, so we assume $j \in \beta_{\mathrm{jam}}$.

Let
$$\hat{\boldsymbol{t}}_i = \begin{bmatrix}
 \hat{\boldsymbol{u}}_i^T & \hat{\boldsymbol{\alpha}}_{i}^T & \hat{\boldsymbol{\gamma}}_{i}^T  & \hat{\boldsymbol{h}}_i^T
\end{bmatrix}$$
denote the received vector by the terminal from any parent node $i$. We define the adversarial error injected into these components as their deviation from the true transmitted values:
\[
\boldsymbol{u}'_k \triangleq \hat{\boldsymbol{u}}_k - \boldsymbol{u}_k, \; \boldsymbol{\alpha}'_{i} \triangleq \hat{\boldsymbol{\alpha}}_{i} - \boldsymbol{\alpha}_{i}, \; \boldsymbol{\gamma}'_{i} \triangleq \hat{\boldsymbol{\gamma}}_{i} - \boldsymbol{\gamma}_{i}, \; \boldsymbol{h}'_{i} \triangleq \hat{\boldsymbol{h}}_{i} - \boldsymbol{h}_{i}.
\]
Under an additive jammer, the adversary directly chooses these error terms ($\boldsymbol{u}'_i, \boldsymbol{\alpha}'_i, \boldsymbol{\gamma}'_i,\boldsymbol{h}'_{i}$), whereas under an overwrite jammer, the adversary chooses the exact received values ($\hat{\boldsymbol{u}}_i, \hat{\boldsymbol{\alpha}}_i, \hat{\boldsymbol{\gamma}}_i, \hat{\boldsymbol{h}}_{i}$), which implicitly defines the error terms. The terminal verifies the cross-consistency between the transmission from parent $i$ and $j$ by checking whether:
\[\small
\hat{h}_{ij} = \sum_{\tau=1}^{n} [\hat{\boldsymbol{u}}_j]_{\tau} \Bigl(\hat{\alpha}_{i,j} + \hat{\alpha}_{j,i}\Bigr)^{\tau} + \Bigl(\hat{\alpha}_{i,j} + \hat{\alpha}_{j,i}\Bigr)^{n+2} + \hat{\gamma}_{i,j} + \hat{\gamma}_{j,i},
\]
where $\hat{h}_{i,j}$, $\hat{\alpha}_{i,j}$, and $\hat{\gamma}_{i,j}$ denote the $j$-th components of $\hat{\boldsymbol{h}}_i$, $\hat{\boldsymbol{\alpha}}_i$, and $\hat{\boldsymbol{\gamma}}_i$, respectively.
If the transmissions received from parents $i$ and $j$ are both uncorrupted, all error terms are zero ($\hat{h}_{ij} = h_{ij}$, $\hat{\boldsymbol{u}}_j = \boldsymbol{u}_j$, etc.), and the check passes trivially. The terminal constructs an undirected graph connecting pairwise-consistent parent nodes and isolates the uncorrupted transmissions by finding the largest clique.
To bound the success probability of a parent node $j \in \beta_{\mathrm{jam}}$ faking consistency with an uncorrupted transmission from parent $i \notin (\beta_{\mathrm{jam}}\cup \beta_{\mathrm{obs}})$, we evaluate the verification check processed at parent node $i$. Because the transmission from parent $i$ is uncorrupted, its local received values match the true values ($\hat{\alpha}_{i,j} = \alpha_{i,j}$, $\hat{\gamma}_{i,j} = \gamma_{i,j}$, and $\hat{h}_{ij} = h_{ij}$). Thus, the check simplifies to:
\[\small
h_{ij} = \sum_{\tau=1}^{n} [\hat{\boldsymbol{u}}_j]_{\tau} \Bigl(\alpha_{i,j} + \hat{\alpha}_{j,i}\Bigr)^{\tau} + \Bigl(\alpha_{i,j} + \hat{\alpha}_{j,i}\Bigr)^{n+2} + \gamma_{i,j} + \hat{\gamma}_{j,i}.
\]
Expanding $h_{ij}$ using its definition and gathering all terms to one side, the random variable $\gamma_{i,j}$ cancels out. The adversary succeeds if and only if:
\begin{equation*}
\begin{aligned}
0 &= \sum_{\tau=1}^{n} [\hat{\boldsymbol{u}}_j]_\tau \Bigl(\alpha_{i,j} + \hat{\alpha}_{j,i}\Bigr)^\tau - \sum_{\tau=1}^{n} [\boldsymbol{u}_j]_\tau \Bigl(\alpha_{i,j} + \alpha_{j,i}\Bigr)^\tau \\
&\quad + \Bigl(\alpha_{i,j} + \hat{\alpha}_{j,i}\Bigr)^{n+2} - \Bigl(\alpha_{i,j} + \alpha_{j,i}\Bigr)^{n+2} + \gamma'_{j,i}.
\end{aligned}
\end{equation*}
To formalize this via polynomial analysis, we define our evaluation point as $x \triangleq \alpha_{i,j} + \hat{\alpha}_{j,i}$. For this verification check to be secure, $x$ must be uniformly distributed and entirely unpredictable from the adversary's perspective. Notice that, outside of the uncorrupted node $i \notin (\beta_{\mathrm{obs}}\cup\beta_{\mathrm{jam}})$, the random variables $\alpha_{i,j}$ and $\gamma_{i,j}$ only ever appear in the hash value $h_{j,i}$. According to \eqref{eq-verify-M}, $h_{j,i}$ is additively masked by $\gamma_{i,j}$. Because node $i$ is uncorrupted, conditioned on the adversary's view, $\gamma_{i,j}$ remains a uniformly distributed random variable. It effectively acts as a one-time pad, ensuring that the hash $h_{j,i}$ reveals no information about $\alpha_{i,j}$. Consequently, $\alpha_{i,j}$, and by extension, the evaluation point $x$, remains uniformly distributed over $\mathbb{F}_q$ from the adversary's point of view.
Expressing the true evaluation point as $\alpha_{i,j} + \alpha_{j,i} = x - \alpha'_{j,i}$, we rewrite the adversary's success condition as a polynomial in $x$:
\begin{equation}\label{eq:cross_poly_unified}
\begin{aligned}
0 &= \Bigl(x^{n+2} - (x - \alpha'_{j,i})^{n+2}\Bigr) \\
&\quad + \sum_{\tau=1}^{n} [\hat{\boldsymbol{u}}_j]_\tau x^\tau - \sum_{\tau=1}^{n} [\boldsymbol{u}_j]_\tau (x - \alpha'_{j,i})^\tau \\
&\quad + \gamma'_{j,i}.
\end{aligned}
\end{equation}
We focus exclusively on harmful corruptions where the data payload is altered, meaning $\hat{\boldsymbol{u}}_j \neq \boldsymbol{u}_j$ (or equivalently, $\boldsymbol{u}'_j \neq \boldsymbol{0}$).
First, suppose that $\alpha'_{j,i} = 0$ (i.e., the adversary preserves or perfectly guesses the true evaluation parameter, $\hat{\alpha}_{j,i} = \alpha_{j,i}$). Equation \eqref{eq:cross_poly_unified} collapses to:
\[
0 = \sum_{\tau=1}^{n} ([\hat{\boldsymbol{u}}_j]_\tau - [\boldsymbol{u}_j]_\tau) x^\tau + \gamma'_{j,i}.
\]
Since $\hat{\boldsymbol{u}}_j \neq \boldsymbol{u}_j$, the right-hand side yields a non-zero polynomial in $x$ of degree at most $n$. For any fixed adversarial choices, this equality can hold for at most $n$ values of $x$. This yields the following bound
\[\small
\begin{aligned}
\Pr\Big(\text{Cross-Consistency Check passes} \mid & \, i \notin (\beta_{\mathrm{jam}} \cup \beta_{\mathrm{obs}}), \\
& \, j \in \beta_{\mathrm{jam}}, \alpha'_{j,i} = 0\Big) \le \frac{n}{q}.
\end{aligned}
\]
Now suppose that $\alpha'_{j,i} \neq 0$. The highest degree term in \eqref{eq:cross_poly_unified} arises uniquely from the binomial expansion of $x^{n+2} - (x - \alpha'_{j,i})^{n+2}$, which produces the leading term:
\[
(n+2)(\alpha'_{j,i})x^{n+1}.
\]
Because $\alpha'_{j,i} \neq 0$, this leading coefficient is non-zero. The remaining summations contribute terms of degree at most $n$. Therefore, the expression is guaranteed to be a non-zero polynomial in $x$ of degree exactly $n+1$, which possesses at most $n+1$ distinct roots. Note that the field size is assumed to be much larger than $n$ (e.g, $q \gg n$). Hence:
\[\small
\begin{aligned}
\Pr\Big(\text{Cross-Consistency Check passes}& \mid  \, i \notin (\beta_{\mathrm{jam}} \cup \beta_{\mathrm{obs}}), \\
& \, j \in \beta_{\mathrm{jam}}, \alpha'_{j,i} \neq 0\Big) \le \frac{n+1}{q}.
\end{aligned}
\]
Combining the two cases via the law of total probability, the probability that a parent node $j \in \beta_{\mathrm{jam}}$ successfully passes the pairwise consistency check with respect to an uncorrupted transmission from parent $i \notin (\beta_{\mathrm{obs}}\cup\beta_{\mathrm{jam}})$ under either an additive or an overwrite attack is bounded by
\[\small
\begin{aligned}
\Pr\Big(\text{Cross-Consistency Check passes}& \mid  \, i \notin (\beta_{\mathrm{jam}} \cup \beta_{\mathrm{obs}}), \\
& \, j \in \beta_{\mathrm{jam}}\Big) \le \frac{n+1}{q}.
\end{aligned}
\]
\paragraph*{Clique Size Analysis}
The terminal constructs an undirected graph $\mathcal{H}$ whose vertices correspond
to the transmissions received from the $d$ parent nodes. A link $(i,j)$ is included
if and only if the transmissions from vertices $i$ and $j$ satisfy the pairwise
cross-consistency check. Under the additive error model, vertices whose transmissions
fail the self-consistency check are first identified and removed. The terminal then
recovers the set of uncorrupted transmissions by finding the largest clique in
$\mathcal{H}$.
In order to uniquely identify the uncorrupted transmissions, the true clique $\mathcal{C}_{\text{true}}$, formed by all mutually consistent uncorrupted parent nodes $j \notin \beta_{\mathrm{jam}}$, must be larger than any clique the adversary can fabricate. The true clique size is
\[
|\mathcal{C}_{\text{true}}| = d - |\beta_{\mathrm{jam}}| = d - |\beta_e| - |\beta_{oe}|.
\]
By the probability bounds established in the cross-consistency analysis, the adversary  cannot, with high probability, forge a consistent transmission from a parent node $j \in \beta_{\mathrm{jam}}$ that passes the check against a transmission from a parent $i \notin (\beta_{\mathrm{obs}} \cup \beta_{\mathrm{jam}})$. Consequently, any
adversarially fabricated clique $\mathcal{C}_{\text{fake}}$ can only span nodes within $\beta_{\mathrm{obs}} \cup \beta_{\mathrm{jam}}$, and its maximum size depends on the error model:

\textbf{Case 1: Additive Jammer.} Under additive jamming, the adversary does not observe the transmissions of jam-only parent nodes in $\beta_e$. As established by the self-consistency check, any error injected into these unobserved parent nodes will cause their transmissions to fail verification with high probability, leading to their immediate removal from $\mathcal{H}$. Therefore, to construct a fake clique, the adversary can only utilize the parent nodes it simultaneously observes and corrupts ($\beta_{oe}$), and attempt to make their transmissions pairwise-consistent with the parent nodes it merely observes ($\beta_o$). The maximum size of any such fake clique is bounded by $|\mathcal{C}_{\mathrm{fake}}| \le |\beta_{oe}| + |\beta_o|$. For the terminal to successfully isolate the true clique, we require $|\mathcal{C}_{\mathrm{true}}| > |\mathcal{C}_{\mathrm{fake}}|$, which gives
\[
  d - |\beta_e| - |\beta_{oe}| > |\beta_{oe}| + |\beta_o|.
\]
Rearranging yields the necessary condition for correct decoding under additive jamming:
\[
  d > |\beta_o| + |\beta_e| + 2|\beta_{oe}|,
\]
which, since $|\beta_o| \le \ell_o$, $|\beta_e| \le \ell_e$, and $|\beta_{oe}| \le \ell_{oe}$ for every $\beta \in \mathcal{B}$, reduces to $d > \ell_o + \ell_e + 2\ell_{oe}$.

\textbf{Case 2: Overwrite Jammer.} An overwrite jammer replaces the transmissions of all jammed parent nodes in $\beta_{\mathrm{jam}} = \beta_e \cup \beta_{oe}$. Consequently, the adversary can forge arbitrary values on the jam-only parent nodes in $\beta_e$ to force them to pass the self-consistency check perfectly. Because these parent nodes are no longer discarded, the adversary can construct a larger fake clique by combining all manipulable parent nodes ($\beta_e \cup \beta_{oe}$) and forging consistency with the observed parent nodes ($\beta_o$). The maximum size of this fake clique is bounded by $|\mathcal{C}_{\mathrm{fake}}| \le |\beta_e| + |\beta_{oe}| + |\beta_o|$. For the terminal to successfully isolate the true clique, we require $|\mathcal{C}_{\mathrm{true}}| > |\mathcal{C}_{\mathrm{fake}}|$, which gives
\[
  d - |\beta_e| - |\beta_{oe}| > |\beta_e| + |\beta_{oe}| + |\beta_o|.
\]
Rearranging yields the necessary condition for correct decoding under overwrite jamming:
\[
  d > |\beta_o| + 2|\beta_e| + 2|\beta_{oe}|,
\]
which, since $|\beta_o| \le \ell_o$, $|\beta_e| \le \ell_e$, and $|\beta_{oe}| \le \ell_{oe}$ for every $\beta \in \mathcal{B}$, reduces to $d > \ell_o + 2\ell_e + 2\ell_{oe}$.

\subsection*{Decoding Error Probability}
% CHANGE 5: replaced "paths the adversary observes or controls" with beta notation
We now bound the overall probability of a decoding error, $P_{\mathrm{err}}$. As
established in the clique size analysis, the topological conditions guarantee that the
true clique $\mathcal{C}_{\text{true}}$ is larger than any fake clique whose
nodes are confined to $\beta_{\mathrm{obs}} \cup \beta_{\mathrm{jam}}$. Therefore, a
decoding error occurs if and only if the adversary successfully bypasses the verification
checks to artificially enlarge a fake clique.

To artificially incorporate a transmission from parent $i \notin (\beta_{\mathrm{jam}} \cup \beta_{\mathrm{obs}})$ into a fake clique, it is a necessary condition that the adversary successfully forges pairwise consistency between at least one parent node in $\beta_{\mathrm{jam}}$ and $i$. While forming a valid clique requires passing multiple simultaneous consistency checks, bounding the probability of at least one successful pairwise forgery serves as an upper bound on the overall error probability.
For an additive jammer, an error occurs if either a transmission from parent $j \in \beta_e$ successfully passes the Self-Consistency Check (defined as event $\mathcal{E}_{\text{self}}$) to expand the fake clique, or a transmission from parent $j \in \beta_{oe}$ successfully fakes pairwise consistency with an uncorrupted transmission from parent $i \notin (\beta_{\mathrm{jam}} \cup \beta_{\mathrm{obs}})$, defined as event $\mathcal{E}_{\text{cross}}$. Applying the union bound over these necessary conditions yields:
\[
\begin{aligned}
P_{\mathrm{err}}^{\text{add}} &\le \Pr(\mathcal{E}_{\text{self}}) + \Pr(\mathcal{E}_{\text{cross}}) \\
&\le \sum_{j \in \beta_e} \frac{n+1}{q} + \sum_{j \in \beta_{oe}} \sum_{i \notin (\beta_{\mathrm{jam}} \cup \beta_{\mathrm{obs}})} \frac{n+1}{q}.
\end{aligned}
\]
Since $|\beta_e| \le \ell_e$, $|\beta_{oe}| \le \ell_{oe}$, and the total number of secure paths is trivially upper-bounded by $d$, we have:
\[
P_{\mathrm{err}}^{\text{add}} \le \left(\ell_e + d \cdot \ell_{oe}\right) \frac{n+1}{q}.
\]
For an overwrite jammer, an error occurs if and only if any parent node $j \in \beta_{\mathrm{jam}}$ (where $\beta_{\mathrm{jam}} = \beta_e \cup \beta_{oe}$) successfully fakes pairwise consistency with an uncorrupted transmission from parent $i \notin (\beta_{\mathrm{jam}} \cup \beta_{\mathrm{obs}})$. Applying the union bound over all combinations:
\[
\begin{aligned}
P_{\mathrm{err}}^{\text{ow}} &\le \Pr(\mathcal{E}_{\text{cross}}) \\
&\le \sum_{j \in \beta_{\mathrm{jam}}} \sum_{i \notin (\beta_{\mathrm{jam}} \cup \beta_{\mathrm{obs}})} \frac{n+1}{q}.
\end{aligned}
\]
Since $|\beta_{\mathrm{jam}}| \le \ell_e + \ell_{oe}$, we obtain:
\[
P_{\mathrm{err}}^{\text{ow}} \le d (\ell_e + \ell_{oe}) \frac{n+1}{q}.
\]
In both cases, the total decoding error probability is strictly bounded by the worst-case scenario (noting that $\ell_e + d\cdot\ell_{oe} \le d(\ell_e + \ell_{oe})$ for any operational $d \ge 1$):
\[
P_{\mathrm{err}} \le P_{\mathrm{err}}^{\text{ow}} \le d (\ell_e + \ell_{oe}) \frac{n+1}{q}.
\]
\subsection*{Rate Analysis}
The secret message $\boldsymbol{X}$ contains $n(d' - \ell')$ symbols over $\mathbb{F}_q$. Substituting the operational path metrics for $d'$ and $\ell'$, the total size of the securely transmitted message is:
\[
|\boldsymbol{X}| = n(d - \ell_e - \ell_o - 2\ell_{oe}).
\]
To transmit this message, the source sends a vector $\boldsymbol{t}_i$ over each of the $d$ vertex disjoint paths. Each vector $\boldsymbol{t}_i$ consists of a payload $\boldsymbol{u}_i$ of $n$ symbols, two random vectors $\boldsymbol{\alpha}_{i}$ and $\boldsymbol{\gamma}_{i}$ each of $d$ symbols, and a hash vector of size $d$. Thus, the total transmission size per path is $n + 2d + d = n + 3d$ symbols.
Defining the secure rate $R$ as the ratio of the secret message size to the transmission size per path, we obtain:
\[
R = \frac{|\boldsymbol{X}|}{n + 3d} = \frac{n}{n + 3d}(d - \ell_e - \ell_o - 2\ell_{oe}).
\]
\subsection*{Key-Capacity}
We now formalize the key-capacity of our proposed scheme based on the established definition.
Recall that an instance has capacity $C$ if for any $\epsilon > 0$ and $\rho > 0$, there exists a corresponding key-code achieving a key-rate $R \ge C - \rho$ with a decoding error probability $P_{\mathrm{err}} \le \epsilon$.
From our rate analysis, the secure rate achieved by the scheme is:
\[
R = \frac{n}{n + 3d}(d - \ell_e - \ell_o - 2\ell_{oe}).
\]
For a fixed network topology where $d$ is constant, as the data payload blocklength $n$ becomes arbitrarily large ($n \to \infty$), the ratio $\frac{n}{n + 3d}$ converges to $1$. Consequently, the overhead penalty diminishes to zero, and the rate asymptotically approaches:
\[
\lim_{n \to \infty} R = d - \ell_e - \ell_o - 2\ell_{oe}.
\]
Thus, for any target gap $\rho > 0$, we can always select a sufficiently large payload parameter $n$ such that the achieved rate satisfies $R \ge (d - \ell_e - \ell_o - 2\ell_{oe}) - \rho$.
Simultaneously, our decoding error probability is bounded by the worst-case scenario:
\[
P_{\mathrm{err}} \le d (\ell_e + \ell_{oe}) \frac{n+1}{q}.
\]
Once a sufficiently large $n$ is chosen to satisfy the rate requirement above, the topological parameters $d, \ell_e, \ell_{oe}$, and $n$ are all fixed finite values. By constructing our code over a sufficiently large finite field $\mathbb{F}_q$ such that $q \gg n$, $\frac{n+1}{q}$ converges to zero. Therefore, for any strict error tolerance $\epsilon > 0$, there exists a large enough field size $q$ to guarantee that $P_{\mathrm{err}} \le \epsilon$.
\subsection*{Security Proof:}
Let $\ell' = \ell_o + \ell_{oe} = |\beta_{\mathrm{obs}}|$ be the maximum number of nodes the adversary can observe. We aim to show that the adversary gains no information about the secret message $\boldsymbol{X}$ from their observations, i.e.,
\begin{equation}
\begin{aligned}
I(\boldsymbol{X};\,& \{\boldsymbol{t}_i : i \in {\beta_{\mathrm{obs}}}\}) \\
&= I\bigl(\boldsymbol{X};\, \{\boldsymbol{u}_i\}_{i \in {\beta_{\mathrm{obs}}}},\, \{\boldsymbol{\alpha}_{i}, \boldsymbol{\gamma}_{i}\}_{i \in {\beta_{\mathrm{obs}}}},\, \{\boldsymbol{h}_{i}\}_{i \in {\beta_{\mathrm{obs}}}}\bigr) \\
&\overset{(a)}{=} I\bigl(\boldsymbol{X};\, \{\boldsymbol{u}_i\}_{i \in {\beta_{\mathrm{obs}}}},\, \{\boldsymbol{\alpha}_{i}, \boldsymbol{\gamma}_{i}\}_{i \in {\beta_{\mathrm{obs}}}},\, \{h_{ij}\}_{i \in {\beta_{\mathrm{obs}}}, j \notin {\beta_{\mathrm{obs}}}}\bigr) \\
&\overset{(b)}{=} I\bigl(\boldsymbol{X};\, \{\boldsymbol{u}_i\}_{i \in {\beta_{\mathrm{obs}}}},\, \{\boldsymbol{\alpha}_{i}, \boldsymbol{\gamma}_{i}\}_{i \in {\beta_{\mathrm{obs}}}}\bigr) \\
&\overset{(c)}{=} I\bigl(\boldsymbol{X};\, \{\boldsymbol{u}_i\}_{i \in {\beta_{\mathrm{obs}}}}\bigr) \\
&\overset{(d)}{=} 0,
\end{aligned}
\end{equation}
where step (a) follows because for all eavesdropped transmissions from parents $j \in {\beta_{\mathrm{obs}}}$, the set of hashes $\{h_{ij}\}_{i \in {\beta_{\mathrm{obs}}}, j \in {\beta_{\mathrm{obs}}}}$ defined in \eqref{eq-verify-M} can be deterministically computed from the payload shares $\{\boldsymbol{u}_i\}_{i \in {\beta_{\mathrm{obs}}}}$ and random vectors $\{\boldsymbol{\alpha}_{i}, \boldsymbol{\gamma}_{i}\}_{i \in {\beta_{\mathrm{obs}}}}$ already present in the adversary's observation set.
Step (b) follows directly from the hash construction. For an unobserved target transmission from parent $j \notin {\beta_{\mathrm{obs}}}$ evaluated on an eavesdropped transmission from parent $i \in {\beta_{\mathrm{obs}}}$, the hash is given by:
\[
h_{ij} = \sum_{\tau=1}^{n} [\boldsymbol{u}_{j}]_{\tau} (\alpha_{i,j} + \alpha_{j,i})^{\tau} + (\alpha_{i,j} + \alpha_{j,i})^{n+2} + \gamma_{i,j} + \gamma_{j,i}.
\]
Crucially, for any $j \notin {\beta_{\mathrm{obs}}}$, the symbol $\gamma_{j,i}$ is drawn uniformly at random from $\mathbb{F}_q$ and is jointly independent of the secret message $\boldsymbol{X}$, $\{\boldsymbol{u}_i\}_{i \in {\beta_{\mathrm{obs}}}}$, and $\{\boldsymbol{\alpha}_{i}, \boldsymbol{\gamma}_{i}\}_{i \in {\beta_{\mathrm{obs}}}}$. By acting as a linear additive term over a finite field, $\gamma_{j,i}$ serves as a perfect one-time pad. This guarantees that the hashes $\{h_{ij}\}_{i \in {\beta_{\mathrm{obs}}}, j \notin {\beta_{\mathrm{obs}}}}$ are themselves uniformly distributed and jointly independent of the adversary's available observations. Since they provide zero mutual information about the payload, they can be safely dropped from the conditioning set.
Step (c) follows since the random vectors in $\{\boldsymbol{\alpha}_{i}, \boldsymbol{\gamma}_{i}\}_{i \in {\beta_{\mathrm{obs}}}}$ are generated independently of both the message $\boldsymbol{X}$ and the payload shares $\{\boldsymbol{u}_i\}_{i \in {\beta_{\mathrm{obs}}}}$.
To establish step (d) and prove perfect secrecy, we must show that the adversary gains zero information about the secret message $\boldsymbol{X}$ from the exposed payloads. Let $\boldsymbol{U}_{\beta_{\mathrm{obs}}} \in \mathbb{F}_q^{n \times |{\beta_{\mathrm{obs}}}|}$ denote the matrix consisting of the columns of $\boldsymbol{U}$ indexed by ${\beta_{\mathrm{obs}}}$. The adversary's observation of these payloads can be expressed algebraically as an affine transformation:
\[
\boldsymbol{U}_{\beta_{\mathrm{obs}}} = \boldsymbol{X}\boldsymbol{V}_{X,{\beta_{\mathrm{obs}}}} + \boldsymbol{R}\boldsymbol{V}_{R,{\beta_{\mathrm{obs}}}}
\]
where $\boldsymbol{V}_{X,{\beta_{\mathrm{obs}}}} \in \mathbb{F}_q^{(d'-\ell') \times |{\beta_{\mathrm{obs}}}|}$ and $\boldsymbol{V}_{R,{\beta_{\mathrm{obs}}}} \in \mathbb{F}_q^{\ell' \times |{\beta_{\mathrm{obs}}}|}$ are the corresponding submatrices of the partitioned Vandermonde matrix $\boldsymbol{V}$.
Because $\boldsymbol{R}$ is generated with elements chosen independently and uniformly at random from $\mathbb{F}_q$, and $\boldsymbol{V}_{R,{\beta_{\mathrm{obs}}}}$ has full column rank (since any square submatrix of a Vandermonde matrix is invertible and $|{\beta_{\mathrm{obs}}}| \le \ell'$), the product $\boldsymbol{R}\boldsymbol{V}_{R,{\beta_{\mathrm{obs}}}}$ is uniformly distributed over $\mathbb{F}_q^{n \times |{\beta_{\mathrm{obs}}}|}$.
Consequently, $\boldsymbol{R}\boldsymbol{V}_{R,{\beta_{\mathrm{obs}}}}$ acts as a one-time pad. The addition of this uniform random matrix masks the deterministic term $\boldsymbol{X}\boldsymbol{V}_{X,{\beta_{\mathrm{obs}}}}$. As a result, the distribution of the adversary's observation $\boldsymbol{U}_{\beta_{\mathrm{obs}}}$ is strictly uniform and statistically independent of $\boldsymbol{X}$. This implies that the following mutual information is exactly zero:
\[
I(\boldsymbol{X}; \{\boldsymbol{u}_i\}_{i \in {\beta_{\mathrm{obs}}}}) = 0.
\]
This concludes the security proof, ensuring the eavesdropper gains no information regardless of their observations.
\end{proof}

\section{Proof of Theorem \ref{T1}}\label{T1Proof}
\begin{proof}
Consider a network in which all nodes are $d$-vertex connected from the source. Let $d' = d - (\ell_e + \ell_{oe})$. Each non-terminal node $v_i \in \mathcal{V}$ is assigned a Vandermonde vector indexed by $i$, denoted by $\boldsymbol{v}_i \in \mathbb{F}_{q}^{d'}$, and for each terminal set $\mathcal{T}_i \subseteq \mathcal{T}$, all terminal nodes $v_t \in \mathcal{T}_i$ share the same Vandermonde vector, denoted by $\boldsymbol{v}_{\mathcal{T}_i} \in \mathbb{F}_{q}^{d'}$. Define the key for terminal set $\mathcal{T}_i$ as
\begin{equation}\label{key-defT1-new}\small
\boldsymbol{K}_{\mathcal{T}_i} = \bigl[\boldsymbol{M}^{(1)} \boldsymbol{v}_{\mathcal{T}_i}, \dots, \boldsymbol{M}^{(n)} \boldsymbol{v}_{\mathcal{T}_i}\bigr]_{1:d'-(\ell_o+\ell_{oe})} \in \mathbb{F}_{q}^{(d'-(\ell_o + \ell_{oe})) \times n},
\end{equation}
where for each $i \in [n]$, $\boldsymbol{M}^{(i)} \in \mathbb{F}_{q}^{d' \times d'}$ is an independent symmetric random matrix whose upper-triangular entries are chosen independently and uniformly from $\mathbb{F}_q$.
\subsection*{Encoding:}
The encoding is described in Algorithm~\ref{alg:keycast_active} and proceeds as follows.
\paragraph*{Step 1}
For each node $j \in \mathcal{V}\setminus\{s\}$, let $b_j \le d$ be the number of edges from $s$ to $j$. The source generates an independent Vandermonde matrix $\boldsymbol{V}_{s\to j} \in \mathbb{F}_q^{d' \times b_j}$ whose columns $\boldsymbol{v}_{1}^{s\to j}, \dots, \boldsymbol{v}_{b_j}^{s\to j} \in \mathbb{F}_q^{d'}$ are distinct Vandermonde vectors unused elsewhere in the network. Over the $i$-th edge ($1 \le i \le b_j$) from $s$ to $j$, the source computes and transmits the data vector
$$\boldsymbol{s}_{s \to j}^{(i)} = [\boldsymbol{v}_{j}^T \boldsymbol{M}^{(1)} \boldsymbol{v}_{i}^{s \to j},\; \dots,\; \boldsymbol{v}_{j}^T \boldsymbol{M}^{(n)} \boldsymbol{v}_{i}^{s \to j}]^{T} \in \mathbb{F}_{q}^{n}.$$
Since the source knows all randomness in the network, it can pre-compute, for every node $j$ and every parent node $j_p \in \mathcal{P}(j)$, the data vector
\begin{equation}\label{eq:data_vector}
\boldsymbol{s}_{j_p \to j} = [\boldsymbol{v}_{j}^T \boldsymbol{M}^{(1)} \boldsymbol{v}_{j_p},\; \dots,\; \boldsymbol{v}_{j}^T \boldsymbol{M}^{(n)} \boldsymbol{v}_{j_p}]^{T} \in \mathbb{F}_{q}^{n}
\end{equation}
that will be transmitted over the edge $e=(j_p,j)\in\mathcal{E}$. Then, for each node $j$, the source node generates an independent random symbol $\alpha_j$, chosen uniformly from $\mathbb{F}_q$, and for each $j_p \in \mathcal{P}(j)$, computes a verification hash value as
\begin{equation}\label{eq:hash_generation}
h_{j_p \to j} = \sum_{\tau=1}^{n} [\boldsymbol{s}_{j_p \to j}]_\tau (\alpha_j)^\tau \in \mathbb{F}_q.
\end{equation}
Since node $j$ is $d$-vertex connected from the source, by Lemma~\ref{lem:parallel_edge_active_secure}, the source generates $\alpha_j$ and $\{h_{j_p \to j}\}_{j_p \in \mathcal{P}(j)}$ and securely transmits the tuple
\[
\left( \alpha_j,\, \bigl( h_{j_p \to j} \bigr)_{j_p \in \mathcal{P}(j)} \right)
\]
to node $j$ through any $d$ vertex-disjoint paths from $s$ to $j$, provided that $d > \ell_o + \ell_e + 2\ell_{oe}$ under an additive adversary and $d > \ell_o + 2\ell_e + 2\ell_{oe}$ under an overwrite adversary. Each intermediate node on these paths serves as a relay, forwarding the received symbols toward $j$ as the topological order advances, so that node $j$ has received the tuple $(\alpha_j, \{h_{j_p \to j}\}_{j_p \in \mathcal{P}(j)})$ by the time its turn arrives in the topological order. Note that the transmissions using Lemma~\ref{lem:parallel_edge_active_secure} operate over a smaller finite field of size $\hat{q}$, whereas the primary transmissions of this theorem operate over a much larger field of size $q$, such that $\hat{q} \le q$. This dual-field approach is practically justifiable because physical transmissions over the network ultimately occur in bits; we can naturally model each field symbol as a structured collection of bits (e.g., allocating $\log_2(q)$ bits to represent a symbol in $\mathbb{F}_q$).
\paragraph*{Step 2}
By induction, let $j \in \mathcal{V} \setminus \{s\}$ denote the next node in the topological order such that every preceding node, including any parent node $j_p \in \mathcal{P}(j)$, has already obtained its share
$$\boldsymbol{S}_{j_p} = [\boldsymbol{M}^{(1)} \boldsymbol{v}_{j_p}, \dots, \boldsymbol{M}^{(n)} \boldsymbol{v}_{j_p}] \in \mathbb{F}_q^{d' \times n},$$
with the base case being the source node $s$. Each non-source parent $j_p \in \mathcal{P}(j)$ computes $\boldsymbol{s}_{j_p \to j}$ as defined in \eqref{eq:data_vector} from its recovered share $\boldsymbol{S}_{j_p}$, and transmits it to $j$. Simultaneously, as described in Step 1, node $j$ receives $\alpha_j$ and the hash values $\{h_{j_p \to j}\}_{j_p \in \mathcal{P}(j)}$ relayed from the source via Lemma~\ref{lem:parallel_edge_active_secure} over the $d$ vertex-disjoint paths. For each $j_p \in \mathcal{P}(j)$, node $j$ verifies the received vector $\hat{\boldsymbol{s}}_{j_p \to j}$ by checking:
\begin{equation}\label{eq-new}
h_{j_p \to j} = \sum_{i=1}^{n} [\hat{\boldsymbol{s}}_{j_p \to j}]_{i} (\alpha_j)^{i}.
\end{equation}
If this equality does not hold, node $j$ declares the transmission corrupted and treats it as an erasure. In the Decoding Error Probability section, we prove that the adversary cannot forge the transmission and bypass this verification except with negligibly small probability. Because the hash check reduces any active corruption to a known erasure, and since at most $\ell_e+\ell_{oe}$ nodes in the parent set of $j$ can belong to $\beta_{\mathrm{jam}}$, node $j$ receives at most $\ell_e+\ell_{oe}$ erasures and zero undetected errors. Let $\mathcal{P}(j)=\{j_1,\dots,j_{d-b_j}\}$ be the non-source parent set of node $j$, and define the Vandermonde matrix
$$\boldsymbol{V}_j = \bigl[\boldsymbol{V}_{s \to j}, \boldsymbol{v}_{j_1}, \dots, \boldsymbol{v}_{j_{d-b_j}}\bigr] \in \mathbb{F}_{q}^{d' \times d}.$$
For each $i\in[n]$, node $j$ receives $b_j$ symbols from the source and $d-b_j$ symbols from its non-source parents, yielding a row vector of the form:
$$\boldsymbol{v}_{j}^{T}\boldsymbol{M}^{(i)}\boldsymbol{V}_j \in \mathbb{F}_{q}^{1 \times d}.$$
Since $\boldsymbol{V}_j$ is a Vandermonde matrix with distinct evaluation points, it generates a $(d,d')$ MDS code over $\mathbb{F}_{q}$ with minimum distance $d_{\min} = d-d'+1 = \ell_e+\ell_{oe}+1$. An MDS code with minimum distance $d_{\min}$ can correct up to $d_{\min}-1 = \ell_e+\ell_{oe}$ erasures, which covers all induced erasures \cite{macwilliams1977theory}. Thus, for every $i\in[n]$, node $j$ recovers $\boldsymbol{M}^{(i)}\boldsymbol{v}_{j}$ and reconstructs its share as
$$\boldsymbol{S}_{j} = [\boldsymbol{M}^{(1)} \boldsymbol{v}_{j}, \boldsymbol{M}^{(2)} \boldsymbol{v}_{j}, \dots, \boldsymbol{M}^{(n)} \boldsymbol{v}_{j}].$$
By induction over the topological ordering, every node reconstructs its share. In particular, every terminal node $t \in \mathcal{T}_i$ reconstructs
$$\boldsymbol{S}_{t} = [\boldsymbol{M}^{(1)} \boldsymbol{v}_{\mathcal{T}_i}, \dots, \boldsymbol{M}^{(n)} \boldsymbol{v}_{\mathcal{T}_i}].$$
\subsection*{Decoding:}
\paragraph*{Step 3}
For each terminal node $t \in \mathcal{T}_i$, the recovered key is the first $d'-(\ell_o+\ell_{oe})$ rows of $\boldsymbol{S}_{t}$, namely
$$\small \boldsymbol{K}_{\mathcal{T}_i} = \bigl[\boldsymbol{M}^{(1)} \boldsymbol{v}_{\mathcal{T}_i}, \dots, \boldsymbol{M}^{(n)} \boldsymbol{v}_{\mathcal{T}_i}\bigr]_{1:d'-(\ell_o+\ell_{oe})} \in \mathbb{F}_{q}^{(d'-(\ell_o + \ell_{oe})) \times n}.$$
\subsection*{Decoding Error Probability}
The overall decoding error probability of the proposed scheme arises from two sources. The first is the probability that the local hash verification fails to detect an adversarially corrupted transmission. The second arises from the application of Lemma~\ref{lem:parallel_edge_active_secure} used to securely deliver the verification parameters $\left(\alpha_j, \{h_{j_p \to j}\}_{j_p \in \mathcal{P}(j)}\right)$ to each node $j$. Let $P_{pd}$ denote the error probability of each such application. In what follows, we first bound the probability of a hash verification failure, then combine this with $P_{pd}$ via a union bound to obtain the overall decoding error probability.

Consider node $j$ and a corrupted parent node $j_p \in \beta_{\mathrm{jam}}$. The adversary attempts to forge a transmission $\hat{\boldsymbol{s}}_{j_p \to j} \neq \boldsymbol{s}_{j_p \to j}$ while satisfying \eqref{eq-new}. By the adversary model, when injecting an error at $j_p$, the adversary's valid observation set is $\mathcal{O}(j_p) = \{u \in \beta_{\mathrm{obs}} \mid j_p \not\prec_{\mathcal{G}} u\}$, i.e., only those observed nodes that are not descendants of $j_p$. Since $(j_p, j) \in \mathcal{E}$, we have $j_p \prec_{\mathcal{G}} j$, so $j \notin \mathcal{O}(j_p)$. Therefore, $\alpha_j$ is not accessible to the adversary when determining its corruption at $j_p$, and hence remains uniformly distributed and independent of the adversary's corruption strategy. Substituting the true hash definition \eqref{eq:hash_generation} into the verification check \eqref{eq-new} yields
\[
\sum_{i=1}^{n} \left( [\hat{\boldsymbol{s}}_{j_p \to j}]_i - [\boldsymbol{s}_{j_p \to j}]_i \right) (\alpha_j)^i = 0.
\]
Let $\boldsymbol{s}'_{j_p \to j} \triangleq \hat{\boldsymbol{s}}_{j_p \to j} - \boldsymbol{s}_{j_p \to j}$ denote the adversarial error vector. The verification condition reduces to
\begin{equation}\label{eq:error_poly_local}
\sum_{i=1}^{n} [\boldsymbol{s}'_{j_p \to j}]_i (\alpha_j)^i = 0.
\end{equation}
If $\boldsymbol{s}'_{j_p \to j}=\boldsymbol{0}$, no corruption has occurred and the check passes trivially. Otherwise, \eqref{eq:error_poly_local} is a non-zero polynomial in $\alpha_j$ of degree at most $n$. Since $\alpha_j$ is independent of $\boldsymbol{s}'_{j_p \to j}$, the polynomial has at most $n$ roots in $\mathbb{F}_q$, giving
\[
\Pr\!\left( \text{Local Check passes} \mid j_p \in \beta_{\mathrm{jam}} \right) \le \frac{n}{q}.
\]
To establish the overall decoding error probability, we apply a union bound over all potential failure events across the network. The application of Lemma~\ref{lem:parallel_edge_active_secure} is executed $|\mathcal{V}|-1$ times, once for each non-source node, contributing at most $(|\mathcal{V}|-1) P_{pd}$ to the total error. During the data transmission phase, each of the $|\mathcal{V}|-1$ non-source nodes evaluates its incoming transmissions, and the adversary can corrupt at most $\ell_e + \ell_{oe}$ nodes in any parent set. Summing the local failure probability over all corrupted nodes and all verifying nodes yields an aggregate forgery probability bounded by $(|\mathcal{V}|-1)(\ell_e + \ell_{oe})\frac{n}{q}$. Consequently, the overall decoding error probability $P_{\mathrm{err}}$ is bounded by
$$P_{\mathrm{err}} \le (|\mathcal{V}|-1) (P_{pd} + (\ell_e + \ell_{oe})\frac{n}{q}).$$

\subsection*{Rate Analysis}
The secret key $\boldsymbol{K}_{\mathcal{T}_i}$ contains $n(d-\ell_o-\ell_e-2\ell_{oe})$ symbols over $\mathbb{F}_q$. The total network blocklength consists of a data transmission phase and the overhead from the secure delivery of verification parameters. The data transmission phase requires $n$ transmissions. Let $R_{pd}$ denote the transmission rate achieved by Lemma~\ref{lem:parallel_edge_active_secure}. Since the source securely transmits $d+1$ verification symbols to each of the $|\mathcal{V}|-1$ non-source nodes, the resulting blocklength overhead is $(|\mathcal{V}|-1)\frac{d+1}{R_{pd}}$. Therefore, the achieved secure key rate is
\[
R = \frac{n(d-\ell_o-\ell_e-2\ell_{oe})}{n+(|\mathcal{V}|-1)\frac{d+1}{R_{pd}}}.
\]
\subsection*{Key-Capacity}
We now formalize the key-capacity of our proposed scheme based on the established definition. Recall that an instance has key-capacity $C$ if for any $\epsilon > 0$ and $\rho > 0$, there exists a corresponding key-code achieving a key-rate $R \ge C - \rho$ with a decoding error probability $P_{\mathrm{err}} \le \epsilon$. From our rate analysis, the secure rate achieved by the scheme is
\[
R = \frac{n(d-\ell_o-\ell_e-2\ell_{oe})}{n+(|\mathcal{V}|-1)\frac{d+1}{R_{pd}}}.
\]
For a fixed network topology and $R_{pd}>0$, as the data payload blocklength $n$ becomes arbitrarily large ($n \to \infty$), the ratio $\frac{n}{n + (|\mathcal{V}|-1)\frac{d+1}{R_{pd}}}$ converges to $1$. Consequently, the overhead penalty diminishes to zero, and the rate asymptotically approaches:
\[
\lim_{n \to \infty} R = d - \ell_e - \ell_o - 2\ell_{oe}.
\]
Thus, for any target gap $\rho > 0$, we can always select a sufficiently large payload parameter $n$ such that $R \ge (d - \ell_e - \ell_o - 2\ell_{oe}) - \rho$. Simultaneously, the decoding error probability satisfies
\[
P_{\mathrm{err}} \le (|\mathcal{V}|-1) (P_{pd} + (\ell_e + \ell_{oe})\frac{n}{q}).
\]
Once $n$ is fixed, by constructing the code over a sufficiently large field $\mathbb{F}_q$ with $q \gg n$, the term $(\ell_e + \ell_{oe})\frac{n}{q}$ converges to zero. Furthermore, from the proof of Lemma~\ref{lem:parallel_edge_active_secure}, $P_{pd} \to 0$ as $\hat{q} \to \infty$. Therefore, for any $\epsilon > 0$, there exist sufficiently large field sizes $q$ and $\hat{q}$ such that $P_{\mathrm{err}} \le \epsilon$.
\subsection*{Security Proof:}
Let $\ell' = \ell_o + \ell_{oe}$ be the maximum number of nodes the adversary can observe. We aim to show that for any terminal set $\mathcal{T}_i \subseteq \mathcal{T}$ such that no terminal in $\mathcal{T}_i$ is observed, the adversary gains no information about the key $\boldsymbol{K}_{\mathcal{T}_i}$, i.e.,
\begin{equation}\small
\begin{aligned}
I\big(\boldsymbol{K}_{\mathcal{T}_i} &;\, \{X_v^{\text{main}}\}_{v\in\beta_{\mathrm{obs}}}, \{X_v^{\text{pd}}\}_{v\in\beta_{\mathrm{obs}}}\big) \\
&= I\big(\boldsymbol{K}_{\mathcal{T}_i};\, \{X_v^{\text{main}}\}_{v\in\beta_{\mathrm{obs}}}\big) \\
&\quad + I\big(\boldsymbol{K}_{\mathcal{T}_i};\, \{X_v^{\text{pd}}\}_{v\in\beta_{\mathrm{obs}}} \mid \{X_v^{\text{main}}\}_{v\in\beta_{\mathrm{obs}}}\big) \\
&= 0,
\end{aligned}
\end{equation}
where $X_v^{\text{main}}$ denotes the data held at node $v$ during the main transmission phase (node shares, hashes, and $\alpha_v$), and $X_v^{\text{pd}}$ denotes the data intercepted at node $v$ during the applications of Lemma~\ref{lem:parallel_edge_active_secure} (the physical transmissions routed through $v$ to deliver verification parameters to other nodes).

To evaluate the conditional term 
$$I\big(\boldsymbol{K}_{\mathcal{T}_i};\, \{X_v^{\text{pd}}\}_{v\in\beta_{\mathrm{obs}}} \mid \{X_v^{\text{main}}\}_{v\in\beta_{\mathrm{obs}}}\big),$$
we refer to the encoding in the proof of Lemma~\ref{lem:parallel_edge_active_secure}. To ensure a strict upper bound on information leakage, we assume the worst case in which $\beta_{\mathrm{obs}}$ intersects every transmission path utilized by each application of the lemma. For each such application, the adversary's observation $\{X_v^{\text{pd}}\}_{v\in\beta_{\mathrm{obs}}}$ takes the form
$$\{X_v^{\text{pd}}\}_{v\in\beta_{\mathrm{obs}}} = \boldsymbol{X}_{pd}\boldsymbol{V}_{X,\beta_{\mathrm{obs}}}^{(pd)} + \boldsymbol{R}\boldsymbol{V}_{R,\beta_{\mathrm{obs}}}^{(pd)},$$
where $\boldsymbol{X}_{pd}$ contains the verification payloads ($\alpha_j$ and hashes $\{h_{j_p\to j}\}_{j_p \in \mathcal{P}(j)}$) and $\boldsymbol{R}$ is a fresh independent random masking matrix generated for each application. Since $\boldsymbol{V}_{R,\beta_{\mathrm{obs}}}^{(pd)}$ has full column rank and $\boldsymbol{R}$ is uniform and independent of $\boldsymbol{K}_{\mathcal{T}_i}$ and $\{X_v^{\text{main}}\}_{v\in\beta_{\mathrm{obs}}}$, the product $\boldsymbol{R}\boldsymbol{V}_{R,\beta_{\mathrm{obs}}}^{(pd)}$ is uniformly distributed over the field and acts as a perfect one-time pad, masking $\boldsymbol{X}_{pd}\boldsymbol{V}_{X,\beta_{\mathrm{obs}}}^{(pd)}$ entirely. Consequently,
\begin{equation}\label{CondiT11}
I\big(\boldsymbol{K}_{\mathcal{T}_i};\, \{X_v^{\text{pd}}\}_{v\in\beta_{\mathrm{obs}}} \mid \{X_v^{\text{main}}\}_{v\in\beta_{\mathrm{obs}}}\big) = 0.
\end{equation}

Next, for the main transmission term, we obtain
\begin{equation}\label{CondiT12}\small
    \begin{aligned}
    I(&\boldsymbol{K}_{\mathcal{T}_i};\, \{X_v^{\text{main}}:v\in\beta_{\mathrm{obs}}\})\\
    &= I\!\big([\boldsymbol{M}^{(1)} \boldsymbol{v}_{\mathcal{T}_i}, \dots, \boldsymbol{M}^{(n)} \boldsymbol{v}_{\mathcal{T}_i}]_{1:d'-(\ell_o + \ell_{oe})};\\
    &\qquad \boldsymbol{M}^{(1)}\boldsymbol{V}_{\epsilon},
      \ldots,
      \boldsymbol{M}^{(n)}\boldsymbol{V}_{\epsilon},\\
      &\qquad\{\alpha_j\}_{j\in\beta_{\mathrm{obs}}},
      \bigl\{\{h_{j_p\to j}\}_{j_p\in\mathcal{P}(j)}\bigr\}_{j\in\beta_{\mathrm{obs}}}\big) \\
    &\overset{(a)}{=} I\!\big([\boldsymbol{M}^{(1)} \boldsymbol{v}_{\mathcal{T}_i}, \dots, \boldsymbol{M}^{(n)} \boldsymbol{v}_{\mathcal{T}_i}]_{1:d'-(\ell_o + \ell_{oe})};\\
    &\qquad \boldsymbol{M}^{(1)}\boldsymbol{V}_{\epsilon},
      \ldots,
      \boldsymbol{M}^{(n)}\boldsymbol{V}_{\epsilon},
      \{\alpha_j\}_{j\in\beta_{\mathrm{obs}}},
    \big) \\
    &\overset{(b)}{=} I\!\big([\boldsymbol{M}^{(1)} \boldsymbol{v}_{\mathcal{T}_i}, \dots, \boldsymbol{M}^{(n)} \boldsymbol{v}_{\mathcal{T}_i}]_{1:d'-(\ell_o + \ell_{oe})};\\
    &\qquad \boldsymbol{M}^{(1)}\boldsymbol{V}_{\epsilon},
      \ldots,
      \boldsymbol{M}^{(n)}\boldsymbol{V}_{\epsilon}, \big) \\
    &\overset{(c)}{=}
    \sum_{\tau=1}^{n}
    I\left(
    [\boldsymbol{M}^{(\tau)} \boldsymbol{v}_{\mathcal{T}_i}]_{1:d'-(\ell_o + \ell_{oe})}
    ;
    \boldsymbol{M}^{(\tau)}\boldsymbol{V}_\epsilon
    \right) \\
    &= \sum_{\tau=1}^{n} \Big( H([\boldsymbol{M}^{(\tau)}\boldsymbol{v}_{\mathcal{T}_i}]_{1:d'-\ell'}) \\
    & \qquad - H\!\left([\boldsymbol{M}^{(\tau)}\boldsymbol{v}_{\mathcal{T}_i}]_{1:d'-\ell'}\,\middle|\,\boldsymbol{M}^{(\tau)}\boldsymbol{V}_\epsilon\right) \Big) \\
    &\overset{(d)}{=} \sum_{\tau=1}^{n} \left( H([\boldsymbol{M}^{(\tau)}\boldsymbol{v}_{\mathcal{T}_i}]_{1:d'-\ell'}) - H\!\left([\boldsymbol{M}^{(\tau)}\boldsymbol{v}_{\mathcal{T}_i}]_{1:d'-\ell'}\right) \right) \\
    &= 0,
    \end{aligned}
\end{equation}
where $\boldsymbol{V}_\epsilon \in \mathbb{F}_q^{d' \times \ell'}$ has columns $\{\boldsymbol{v}_j\}_{j \in \beta_{\mathrm{obs}}}$; (a) follows because, for all eavesdropped nodes $j \in \beta_{\mathrm{obs}}$, the hashes $\{h_{j_p\to j}\}_{j_p\in\mathcal{P}(j)}$ defined in \eqref{eq:hash_generation} can be deterministically computed from $\{\alpha_j\}_{j\in\beta_{\mathrm{obs}}}$ and $\{\boldsymbol{M}^{(\tau)}\boldsymbol{v}_j\}_{j\in\beta_{\mathrm{obs}}, \tau\in[n]}$ already present in $\{\boldsymbol{M}^{(\tau)}\boldsymbol{V}_{\epsilon}\}_{\tau \in [n]}$; (b) follows because $\{\alpha_j\}_{j \in \beta_{\mathrm{obs}}}$ are generated independently and uniformly, hence independent of $\{\boldsymbol{M}^{(\tau)}\}_{\tau\in[n]}$ and any function thereof; (c) follows since $\boldsymbol{M}^{(1)},\dots,\boldsymbol{M}^{(n)}$ are independent, making the pairs $\bigl([\boldsymbol{M}^{(\tau)}\boldsymbol{v}_{\mathcal{T}_i}]_{1:d'-(\ell_o+\ell_{oe})}, \boldsymbol{M}^{(\tau)}\boldsymbol{V}_\epsilon\bigr)$ independent across $\tau$; and (d) follows from \cite[Lemma~3.2]{eavsPaper} and the fact that $\boldsymbol{v}_{\mathcal{T}_i}$ is linearly independent of the columns of $\boldsymbol{V}_\epsilon$, since no terminal in $\mathcal{T}_i$ is observed. 

By combining \eqref{CondiT11} and \eqref{CondiT12}, the total information leakage is 
$$I\big(\boldsymbol{K}_{\mathcal{T}_i};\, \{X_v^{\text{main}}\}_{v\in\beta_{\mathrm{obs}}}, \{X_v^{\text{pd}}\}_{v\in\beta_{\mathrm{obs}}}\big) = 0.$$
Hence, the eavesdropper obtains no information about the key $\boldsymbol{K}_{\mathcal{T}_i}$, ensuring perfect secrecy.
\end{proof}

\section{Proof of Lemma \ref{lem:subnetwork-indegree}}\label{L1Proof}
\begin{proof}
The proof closely follows \cite{lovasz1973connectivity}, adapted to
vertex-based connectivity via a standard node-splitting reduction.

\paragraph*{Node splitting}
We reduce vertex connectivity to edge connectivity using the following
construction.  Build a split digraph
$\hat{\mathcal{G}}=(\hat{\mathcal{V}},\hat{\mathcal{E}})$ by replacing every
non-source node $v\in\mathcal{V}\setminus\{s\}$ with a pair of nodes
$v^{\rm in}$ and $v^{\rm out}$ joined by a single directed
\emph{edge} $v^{\rm in}\!\to\! v^{\rm out}$:
\begin{align*}
  \hat{\mathcal{V}}
    &= \{s\}
       \;\cup\;
       \bigl\{v^{\rm in},\,v^{\rm out}
             \;:\; v\in\mathcal{V}\setminus\{s\}\bigr\},\\[4pt]
  \hat{\mathcal{E}}
    &= \bigl\{(s,\,v^{\rm in})
              \;:\; (s,v)\in\mathcal{E}\bigr\}\\
     &\qquad\cup\;
       \bigl\{(u^{\rm out},\,v^{\rm in})
              \;:\; (u,v)\in\mathcal{E},\;u\neq s\bigr\}\\
     &\qquad\cup\;
       \bigl\{(v^{\rm in},\,v^{\rm out})
              \;:\; v\in\mathcal{V}\setminus\{s\}\bigr\}.
\end{align*}
Let $c_{\hat{\mathcal{G}}}(s,v)$ denote the maximum number of
\emph{edge}-disjoint directed paths from $s$ to $v$ in $\hat{\mathcal{G}}$.
The splitting establishes
\begin{equation}
  c_{\hat{\mathcal{G}}}(s,\,v^{\rm in})
  \;=\;
  \kappa_{\mathcal{G}}(s,v)
  \qquad
  \forall\,v\in\mathcal{V}\setminus\{s\},
  \label{eq:split-correspondence}
\end{equation}
and, since every path to $v^{\rm out}$ must cross the single bottleneck
edge $v^{\rm in}\to v^{\rm out}$,
\begin{equation}
  c_{\hat{\mathcal{G}}}(v^{\rm in},\,v^{\rm out})
  \;=\;
  1
  \qquad
  \forall\,v\in\mathcal{V}\setminus\{s\}.
  \label{eq:split-correspondence-out}
\end{equation}
Indeed, every internally vertex-disjoint path
$s\!=\!v_0\!\to\!v_1\!\to\!\cdots\!\to\!v_k\!=\!v$ in $\mathcal{G}$
maps to an edge-disjoint path
$s\to v_1^{\rm in}\!\to\!v_1^{\rm out}\!\to\!\cdots\!\to\!v^{\rm in}$
in $\hat{\mathcal{G}}$, where vertex-disjointness at each internal node
$v_i$ corresponds precisely to using the bottleneck edge
$v_i^{\rm in}\!\to\!v_i^{\rm out}$ at most once.  The converse direction
is analogous.  The rest of the proof is identical to that provided in
\cite{lovasz1973connectivity} on $\hat{\mathcal{G}}$, but for clarity we
provide it here.

\paragraph*{Minimal subgraph}
Let $\hat{\mathcal{G}}'$ be a \emph{minimal} subgraph of $\hat{\mathcal{G}}$
(with the same vertex set $\hat{\mathcal{V}}$) satisfying
\begin{equation}
  c_{\hat{\mathcal{G}}'}(s,\,v^{\rm in})
  \;=\;
  \kappa_{\mathcal{G}}(s,v)
  \qquad
  \forall\,v\in\mathcal{V}\setminus\{s\},
  \label{eq:minimality-condition}
\end{equation}
and
\begin{equation}
  c_{\hat{\mathcal{G}}'}(v^{\rm in},\,v^{\rm out})
  \;=\;
  1
  \qquad
  \forall\,v\in\mathcal{V}\setminus\{s\}.
  \label{eq:minimality-condition-out}
\end{equation}
Such a subgraph exists, take $\hat{\mathcal{G}}$ itself, and a minimal one
is obtained by removing edges from $\hat{\mathcal{E}}$ one at a time as
long as \eqref{eq:minimality-condition} and
\eqref{eq:minimality-condition-out} continue to hold.  For a set
$X\subseteq\hat{\mathcal{V}}\setminus\{s\}$, write $\phi(X)$ for the
number of edges entering $X$ from $\hat{\mathcal{V}}\setminus X$ in
$\hat{\mathcal{G}}'$.

\paragraph*{Regular sets}
Call a set $X\subseteq\hat{\mathcal{V}}\setminus\{s\}$ \emph{regular}
(relative to $s$) if
\[
  \phi(X) \;=\; c_{\hat{\mathcal{G}}}(s,w)
  \quad
  \text{for some } w\in X,
\]
and refer to such $w$ as a \emph{core} of $X$.  In words, $X$ is regular
when the edges entering $X$ form a minimum cut for some node inside $X$.

The function $\phi$ is \emph{submodular}: for any
$X,Y\subseteq\hat{\mathcal{V}}\setminus\{s\}$,
\begin{equation}
  \phi(X\cup Y)+\phi(X\cap Y)\;\leq\;\phi(X)+\phi(Y),
  \label{eq:submodular}
\end{equation}
since each edge entering $X\cup Y$ or $X\cap Y$ is counted at least as
many times on the right-hand side.  We now establish two properties of
regular sets that will be used in the main argument.

\begin{claim}[Closure under $\cap$ and $\cup$]\label{l1c1}
If $U$ and $V$ are regular sets with cores $u$ and $v$
respectively, and $u\in V$, then $U\cap V$ is regular with core $u$ and
$U\cup V$ is regular with core $v$.
\end{claim}
\begin{proof}
By \eqref{eq:submodular}:
\[
  \phi(U\cup V)+\phi(U\cap V)
  \;\leq\;
  \phi(U)+\phi(V)
  \;=\;
  c_{\hat{\mathcal{G}}}(s,u)+c_{\hat{\mathcal{G}}}(s,v).
\]
Since $u\in U\cap V$ (as $u\in U$ and $u\in V$ by assumption), Menger's
theorem applied to $\hat{\mathcal{G}}'$ gives
$\phi(U\cap V)\geq c_{\hat{\mathcal{G}}}(s,u)$.  Since $v\in U\cup V$,
similarly $\phi(U\cup V)\geq c_{\hat{\mathcal{G}}}(s,v)$.  Adding these
two lower bounds matches the upper bound, so equality holds throughout:
$\phi(U\cap V)=c_{\hat{\mathcal{G}}}(s,u)$ and
$\phi(U\cup V)=c_{\hat{\mathcal{G}}}(s,v)$.
\end{proof}

\begin{claim}[Minimal regular set]\label{l1c2}
For every $w\in\hat{\mathcal{V}}\setminus\{s\}$ there exists a
unique smallest regular set $T_w$ with core $w$: it is contained in every
regular set that contains $w$.
\end{claim}
\begin{proof}
By the edge version of Menger's theorem applied to $\hat{\mathcal{G}}'$,
any minimum $(s\leadsto w)$-cut corresponds to a set
$U\subseteq\hat{\mathcal{V}}\setminus\{s\}$ with $w\in U$ and
$\phi(U)=c_{\hat{\mathcal{G}}'}(s,w)=c_{\hat{\mathcal{G}}}(s,w)$
(the last equality holds by \eqref{eq:minimality-condition} for nodes
$w=v^{\rm in}$, and such nodes are the only ones that appear as
cores in the argument below).  Hence $w$ is always the core of some
regular set.  Define $T_w$ as the intersection of all regular sets with
core $w$.  By repeated application of Claim~\ref{l1c1}, $T_w$ is itself
regular with core $w$.  For any regular set $U$ with $w\in U$:
Claim~\ref{l1c1} gives that $T_w\cap U$ is regular with core $w$, so by
the definition of $T_w$ we have $T_w\subseteq T_w\cap U\subseteq U$.
\end{proof}

\paragraph*{Incoming edges of $v^{\rm in}$ originate outside
$T_{v^{\rm in}}$}
Fix any $v\in\mathcal{V}\setminus\{s\}$.  We show that every directed
edge in $\hat{\mathcal{G}}'$ entering $v^{\rm in}$ has its tail
\emph{outside} $T_{v^{\rm in}}$.  Let $E=(y\to v^{\rm in})$ be any such
edge.  Removing $E$ from $\hat{\mathcal{G}}'$ yields a subgraph
$\hat{\mathcal{G}}''$; by minimality of $\hat{\mathcal{G}}'$ there exists
some $z\in\mathcal{V}\setminus\{s\}$ such that
\[
  c_{\hat{\mathcal{G}}''}(s,\,z^{\rm in})
  \;<\;
  c_{\hat{\mathcal{G}}'}(s,\,z^{\rm in})
  \;=\;
  c_{\hat{\mathcal{G}}}(s,\,z^{\rm in}).
\]
By Menger's theorem applied to $\hat{\mathcal{G}}''$, there is a set
$Z\subseteq\hat{\mathcal{V}}\setminus\{s\}$ with $z^{\rm in}\in Z$ and
\[
  \phi''(Z)
  \;=\;
  c_{\hat{\mathcal{G}}''}(s,\,z^{\rm in})
  \;<\;
  c_{\hat{\mathcal{G}}}(s,\,z^{\rm in})
  \;\leq\;
  \phi(Z),
\]
where $\phi''$ denotes the cut function in $\hat{\mathcal{G}}''$ and the
last inequality is Menger applied to $\hat{\mathcal{G}}'$.  Since
$\hat{\mathcal{G}}''$ differs from $\hat{\mathcal{G}}'$ only by the
removal of the single edge $E$, we have $\phi''(Z)\geq\phi(Z)-1$.
Combined with $\phi''(Z)<\phi(Z)$ and integrality, this forces
\[
  \phi(Z) \;=\; c_{\hat{\mathcal{G}}}(s,\,z^{\rm in}),
  \qquad
  \phi''(Z) \;=\; \phi(Z)-1.
\]
The first equality shows $Z$ is a \emph{regular} set with core $z^{\rm in}$
in $\hat{\mathcal{G}}'$.  The second equality shows that edge $E$ crosses
into $Z$, i.e., $v^{\rm in}\in Z$ and $y\notin Z$.  Since $v^{\rm in}\in Z$
and $Z$ is a regular set, Claim~\ref{l1c2} gives $T_{v^{\rm in}}\subseteq
Z$.  Because $y\notin Z$, it follows that $y\notin T_{v^{\rm in}}$.
Since $E$ was an arbitrary edge entering $v^{\rm in}$, we conclude: every
edge of $\hat{\mathcal{G}}'$ entering $v^{\rm in}$ originates outside
$T_{v^{\rm in}}$.

\paragraph*{Conclusion and translation back to $\mathcal{G}$}
Because every edge entering $v^{\rm in}$ in $\hat{\mathcal{G}}'$ comes
from outside $T_{v^{\rm in}}$, and $v^{\rm in}\in T_{v^{\rm in}}$, each
such edge also crosses into the set $T_{v^{\rm in}}$.  Hence
\[
  |\mathcal{P}_{\hat{\mathcal{G}}'}(v^{\rm in})|
  \;\leq\;
  \phi(T_{v^{\rm in}})
  \;=\;
  c_{\hat{\mathcal{G}}}(s,\,v^{\rm in})
  \;=\;
  \kappa_{\mathcal{G}}(s,v),
\]
where the first equality uses the fact that $T_{v^{\rm in}}$ is regular
with core $v^{\rm in}$.  On the other hand, $\{v^{\rm in}\}$ is itself a
valid $(s\leadsto v^{\rm in})$-cut in $\hat{\mathcal{G}}'$, so Menger
gives
\[
  \kappa_{\mathcal{G}}(s,v)
  \;=\;
  c_{\hat{\mathcal{G}}'}(s,\,v^{\rm in})
  \;\leq\;
  \phi\!\bigl(\{v^{\rm in}\}\bigr)
  \;=\;
  |\mathcal{P}_{\hat{\mathcal{G}}'}(v^{\rm in})|.
\]
Combining the two inequalities:
\[
  |\mathcal{P}_{\hat{\mathcal{G}}'}(v^{\rm in})|
  \;=\;
  \kappa_{\mathcal{G}}(s,v).
\]
Finally, define $\mathcal{G}'=(\mathcal{V},\mathcal{E}')$ by collapsing
the split nodes:
\begin{align*}
      \mathcal{E}'
  \;&=\;
  \bigl\{(u,v)\in\mathcal{E}
         \;:\; (u^{\rm out},v^{\rm in})\in\hat{\mathcal{E}}'\bigr\} \\
  &\qquad\cup\;
  \bigl\{(s,v)\in\mathcal{E}
         \;:\; (s,v^{\rm in})\in\hat{\mathcal{E}}'\bigr\}.
\end{align*}
Before using this construction, observe that for any
$v\in\mathcal{V}\setminus\{s\}$ with $\kappa_{\mathcal{G}}(s,v)\neq 0$,
the split edge $v^{\rm in}\to v^{\rm out}$ is never removed during
minimization: this edge is the unique edge entering $v^{\rm out}$ in
$\hat{\mathcal{G}}'$, so removing it would reduce
$c_{\hat{\mathcal{G}}'}(s,\,v^{\rm out})$ to zero, directly violating
\eqref{eq:minimality-condition-out}.  Hence $\hat{\mathcal{G}}'$ retains
every such split edge, a fact we use below.  Every edge entering $v$ in
$\mathcal{G}$ corresponds to a unique edge entering $v^{\rm in}$ in
$\hat{\mathcal{G}}$, so
\[
  |\mathcal{P}_{\mathcal{G}'}(v)|
  \;=\;
  |\mathcal{P}_{\hat{\mathcal{G}}'}(v^{\rm in})|
  \;=\;
  \kappa_{\mathcal{G}}(s,v).
\]
Moreover, the same splitting correspondence applied to $\mathcal{G}'$,
which is valid precisely because every split edge of $\hat{\mathcal{G}}$
survives in $\hat{\mathcal{G}}'$, gives
\[
  \kappa_{\mathcal{G}'}(s,v)
  \;=\;
  c_{\hat{\mathcal{G}}'}(s,\,v^{\rm in})
  \;=\;
  c_{\hat{\mathcal{G}}}(s,\,v^{\rm in})
  \;=\;
  \kappa_{\mathcal{G}}(s,v),
\]
completing the proof.
\end{proof}

\section{Proof of Theorem \ref{T3-active}}\label{T2Proof}
\begin{proof}
By Lemma~\ref{lem:subnetwork-indegree}, for any directed acyclic graph $\mathcal{G}=(\mathcal{V},\mathcal{E})$ with source $s$, there exists a subgraph $\hat{\mathcal{G}}=(\mathcal{V},\hat{\mathcal{E}})$ with $\hat{\mathcal{E}}\subseteq\mathcal{E}$ satisfying
\[
  \kappa_{\hat{\mathcal{G}}}(s,v)=\kappa_{\mathcal{G}}(s,v)
\]
and
\[
  \bigl|\mathcal{P}_{\hat{\mathcal{G}}}(v)\bigr|=\kappa_{\mathcal{G}}(s,v),
\]
for all $v\in\mathcal{V}\setminus\{s\}$. That is, $\hat{\mathcal{G}}$ preserves every node's vertex connectivity from $s$ while reducing every node's in-degree to exactly that vertex connectivity. In particular, every partially-connected node $j'$ satisfies
\[
  \bigl|\mathcal{P}_{\hat{\mathcal{G}}}(j')\bigr| = \kappa_{\hat{\mathcal{G}}}(s,j') \le d-1,
\]
since by definition a partially-connected node has fewer than $d$ vertex-disjoint paths from $s$. Because any coding scheme that is correct and secure over $\hat{\mathcal{G}}$ is equally correct and secure over $\mathcal{G}$ (which contains all edges of $\hat{\mathcal{G}}$ as a subgraph), it suffices to establish the achievability result on $\hat{\mathcal{G}}$. Henceforth, all transmissions, path selections, and edge references are with respect to $\hat{\mathcal{G}}$.

Consider a network in which all terminal nodes are $d$-vertex connected from the source. Let $d' = d - (\ell_e + \ell_{oe})$ and $\ell' = \ell_o + \ell_{oe}$. Each non-terminal node $v_i \in \mathcal{V}$ is assigned two Vandermonde vectors $\boldsymbol{v}_{i}^{(M)} \in \mathbb{F}_q^{d'-z}$ and $\boldsymbol{v}_i^{(R)}\in \mathbb{F}_q^{d'}$ indexed by $i$. Furthermore, for each terminal set $\mathcal{T}_i \subseteq \mathcal{T}$, all terminal nodes $v_t \in \mathcal{T}_i$ share the same Vandermonde vectors, denoted by $\boldsymbol{v}_{\mathcal{T}_i}^{(M)} \in \mathbb{F}_q^{d'-z}$ and $\boldsymbol{v}_{\mathcal{T}_i}^{(R)}\in \mathbb{F}_q^{d'}$. All Vandermonde vectors described above are distinct. Define the key for terminal set $\mathcal{T}_i$ as the matrix $\boldsymbol{K}_{\mathcal{T}_i} \in \mathbb{F}_{q}^{(d'-\ell'-z+1) \times n}$, where the $i$-th column (for $i \in [n]$) is constructed as:
\[
  \boldsymbol{k}_{\mathcal{T}_i}^{(i)} = [\boldsymbol{M}^{(i)} \boldsymbol{v}_{\mathcal{T}_i}^{(M)} ]_{1:d'-\ell'-z+1} +[\boldsymbol{R}^{(i)} \boldsymbol{v}_{\mathcal{T}_i}^{(R)} ]_{1:d'-\ell'-z+1}.
\]
Here, for each $i \in [n]$, $\boldsymbol{M}^{(i)} \in \mathbb{F}_{q}^{(d'-z) \times (d'-z)}$ and $\boldsymbol{R}^{(i)} \in \mathbb{F}_{q}^{d' \times d'}$ are independent symmetric random matrices whose upper-triangular entries are chosen independently and uniformly from $\mathbb{F}_q$.

\subsection*{Encoding:}
\paragraph*{Step 1}
For each node $j \in \mathcal{V}\setminus\{s\}$, let $b_j \le d$ be the number of direct edges from $s$ to $j$. The source generates independent Vandermonde matrices $\boldsymbol{V}_{s\to j}^{(M)} \in \mathbb{F}_q^{(d'-z) \times b_j}$ and $\boldsymbol{V}_{s\to j}^{(R)} \in \mathbb{F}_q^{d' \times b_j}$ whose columns are distinct Vandermonde vectors unused elsewhere in the network. If $j$ is $d$-vertex connected, over the $m$-th edge ($1 \le m \le b_j$) from $s$ to $j$, the source computes and transmits the data vectors of size $n$:
$$\small \boldsymbol{s}_{s \to j}^{(M, m)} = \bigl[ {\boldsymbol{v}_{j}^{(M)}}^T \boldsymbol{M}^{(1)} \boldsymbol{v}_{m}^{(M), s \to j},\; \dots, \;{\boldsymbol{v}_{j}^{(M)}}^T \boldsymbol{M}^{(n)} \boldsymbol{v}_{m}^{(M), s \to j} \bigr]^{T}, $$
$$\small \boldsymbol{s}_{s \to j}^{(R, m)} = \bigl[ {\boldsymbol{v}_{j}^{(R)}}^T \boldsymbol{R}^{(1)} \boldsymbol{v}_{m}^{(R), s \to j},\; \dots, \;{\boldsymbol{v}_{j}^{(R)}}^T \boldsymbol{R}^{(n)} \boldsymbol{v}_{m}^{(R), s \to j} \bigr]^{T}. $$
If $j$ is less than $d$-vertex connected (partially connected), then for each $i \in [n]$, the source directly sends the vector $\boldsymbol{R}^{(i)}\boldsymbol{v}_j^{(R)}$ to node $j$.

Furthermore, since the source knows all randomness in the network, it can pre-compute, for every $d$-vertex connected node $j$ and every parent node $j_p \in \mathcal{P}(j)$, the vectors $\boldsymbol{s}_{j_p \to j}^{(M)}$ and $\boldsymbol{s}_{j_p \to j}^{(R)}$ that will be transmitted over the edge $(j_p,j)$ according to the node-type rules defined later in Step 2; if a particular component is absent (e.g., $\boldsymbol{s}_{j_p \to j}^{(M)}$ when $j_p$ is partially connected), it is treated as the zero vector $\boldsymbol{0} \in \mathbb{F}_q^n$. The source then generates an independent random symbol $\alpha_j$ uniformly from $\mathbb{F}_q$, and for each $j_p \in \mathcal{P}(j)$, computes verification hash values as
\begin{align}\label{hash}
  h_{j_p \to j}^{(M)} &= \sum_{i=1}^{n} [\boldsymbol{s}_{j_p \to j}^{(M)}]_i (\alpha_j)^i \in \mathbb{F}_q, \\
  h_{j_p \to j}^{(R)} &= \sum_{i=1}^{n} [\boldsymbol{s}_{j_p \to j}^{(R)}]_i (\alpha_j)^i \in \mathbb{F}_q.
\end{align}
Since node $j$ is $d$-vertex connected from the source, by Lemma~\ref{lem:parallel_edge_active_secure}, the source generates the tuple
\[
  \Bigl(\alpha_j,\; \bigl(h_{j_p \to j}^{(M)}\bigr)_{j_p \in \mathcal{P}(j)},\; \bigl(h_{j_p \to j}^{(R)}\bigr)_{j_p \in \mathcal{P}(j)}\Bigr)
\]
and securely transmits it to node $j$ through any $d$ vertex-disjoint paths from $s$ to $j$, provided that $d > \ell_o + \ell_e + 2\ell_{oe}$ under an additive adversary and $d > \ell_o + 2\ell_e + 2\ell_{oe}$ under an overwrite adversary. Each intermediate node on these paths serves as a relay, forwarding the received symbols toward $j$ in topological order. The transmissions using Lemma~\ref{lem:parallel_edge_active_secure} operate over a smaller finite field of size $\hat{q}$, whereas the primary transmissions operate over $\mathbb{F}_q$ with $\hat{q} \le q$.

\paragraph*{Step 2}
By induction, let $j \in \mathcal{V}$ be the next node in the topological order such that every preceding node, including any parent node $j_p \in \mathcal{P}(j)$, has either received its shares or is the source node. For each $d$-vertex connected node ${j_p}$, these shares are
\[
  \boldsymbol{S}_{{j_p}}^{(M)} = [\boldsymbol{M}^{(1)}\boldsymbol{v}_{{j_p}}^{(M)}, \dots, \boldsymbol{M}^{(n)}\boldsymbol{v}_{{j_p}}^{(M)}],
\]
and
\[
  \boldsymbol{S}_{{j_p}}^{(R)} = [\boldsymbol{R}^{(1)}\boldsymbol{v}_{{j_p}}^{(R)}, \dots, \boldsymbol{R}^{(n)}\boldsymbol{v}_{{j_p}}^{(R)}].
\]
For each node ${j_p} \in \mathcal{N}(s)$ that is less than $d$-vertex connected, its share is
\[
  \boldsymbol{S}_{{j_p}}^{(R)} = [\boldsymbol{R}^{(1)}\boldsymbol{v}_{{j_p}}^{(R)}, \dots, \boldsymbol{R}^{(n)}\boldsymbol{v}_{{j_p}}^{(R)}].
\]
For each partially-connected node ${j_p} \notin \mathcal{N}(s)$, the share consists of the accumulated collection $\mathcal{s}_{{j_p}}^{(R)}$, defined as all path-tagged matrix copies that ${j_p}$ has received from nodes in $\mathcal{D}({j_p})$, defined below
\begin{equation}\small
  \mathcal{D}({j_p}) \triangleq \left\{ j' \in \mathcal{V}: \begin{array}{l} j' \text{ is either } d\text{-vertex connected from } s \\ \text{or is a source-neighbor node, and} \\ \text{there exists a directed path from } j' \\ \text{to } {j_p} \text{ whose internal nodes, if any,} \\ \text{are all partially-connected nodes} \end{array} \right\}.
\end{equation}
Let $\Pi(j'\!\to\!{j_p})$ denote the set of directed paths from $j'$ to ${j_p}$ in $\hat{\mathcal{G}}$ whose internal nodes are all partially-connected.  Since transmissions proceed in topological order,
each node ${j_p}$ receives the shares from $\mathcal{D}({j_p})$ in a fixed deterministic order induced by the topology.  We define
\begin{align*}
  \mathcal{s}_{{j_p}}^{(R)}
  \;\triangleq\;
  \Bigl([\boldsymbol{R}^{(1)}\boldsymbol{v}_{j'}^{(R)},
          \dots,
          \boldsymbol{R}^{(n)}\boldsymbol{v}_{j'}^{(R)}]
  \;:\; j'\in\mathcal{D}({j_p}),\\
        \pi\in\Pi(j'\!\to\!{j_p})\Bigr),
\end{align*}
as an \emph{ordered} tuple, where the ordering is the topological arrival order. Note that the same underlying matrix $[\boldsymbol{R}^{(1)}\boldsymbol{v}_{j'}^{(R)},\dots, \boldsymbol{R}^{(n)}\boldsymbol{v}_{j'}^{(R)}]$ may appear multiple times in this tuple, once per distinct path $\pi\in\Pi(j'\!\to\!{j_p})$, since the same share may reach ${j_p}$ via several different routes through the partially-connected subgraph of $\hat{\mathcal{G}}$. Since all nodes know the network topology, the position of each entry in the tuple unambiguously identifies the path it traversed, with no additional transmission overhead. 
% The precise construction of these tuples and how they are processed will be detailed later.

We specify the transmissions over each edge $j_p \to j$ for each non-source parent node $j_p \in \mathcal{P}(j)$ depending on its type:
\begin{itemize}
  \item \textbf{If $j_p$ is a $d$-vertex connected node:} Node $j_p$ utilizes its recovered shares
    \[
      \boldsymbol{S}_{{j_p}}^{(M)} = [\boldsymbol{M}^{(1)}\boldsymbol{v}_{{j_p}}^{(M)}, \dots, \boldsymbol{M}^{(n)}\boldsymbol{v}_{{j_p}}^{(M)}],
    \]
    and
    \[
      \boldsymbol{S}_{{j_p}}^{(R)} = [\boldsymbol{R}^{(1)}\boldsymbol{v}_{{j_p}}^{(R)}, \dots, \boldsymbol{R}^{(n)}\boldsymbol{v}_{{j_p}}^{(R)}].
    \]
    The transmission depends on the classification of the receiving node $j$.
    \begin{itemize}
      \item If $j$ is $d$-vertex connected, $j_p$ computes and transmits:
        \begin{align*}
          \boldsymbol{s}_{j_p \to j}^{(M)} &= \bigl[{\boldsymbol{v}_{j}^{(M)}}^T \boldsymbol{M}^{(1)} \boldsymbol{v}_{j_p}^{(M)}, \dots, {\boldsymbol{v}_{j}^{(M)}}^T \boldsymbol{M}^{(n)} \boldsymbol{v}_{j_p}^{(M)}\bigr]^T, \\
          \boldsymbol{s}_{j_p \to j}^{(R)} &= \bigl[{\boldsymbol{v}_{j}^{(R)}}^T \boldsymbol{R}^{(1)} \boldsymbol{v}_{j_p}^{(R)}, \dots, {\boldsymbol{v}_{j}^{(R)}}^T \boldsymbol{R}^{(n)} \boldsymbol{v}_{j_p}^{(R)}\bigr]^T.
        \end{align*}
      \item If $j$ is partially connected, $j_p$ directly forwards its full matrix share $\boldsymbol{S}_{j_p}^{(R)}$ so that $j$ can incorporate it into $\mathcal{s}_j^{(R)}$.
    \end{itemize}

  \item \textbf{If $j_p \in \mathcal{N}(s)$ is partially connected:} Node $j_p$ holds $\boldsymbol{S}_{j_p}^{(R)} = [\boldsymbol{R}^{(1)}\boldsymbol{v}_{j_p}^{(R)}, \dots, \boldsymbol{R}^{(n)}\boldsymbol{v}_{j_p}^{(R)}]$, received directly from the source.
    \begin{itemize}
      \item If $j$ is $d$-vertex connected, $j_p$ transmits:
        \[
          \boldsymbol{s}_{j_p \to j}^{(R)} = \bigl[{\boldsymbol{v}_{j}^{(R)}}^T \boldsymbol{R}^{(1)} \boldsymbol{v}_{j_p}^{(R)}, \dots, {\boldsymbol{v}_{j}^{(R)}}^T \boldsymbol{R}^{(n)} \boldsymbol{v}_{j_p}^{(R)}\bigr]^T.
        \]
      \item If $j$ is partially connected, $j_p$ directly forwards its share $\boldsymbol{S}_{j_p}^{(R)}$.
    \end{itemize}

  \item \textbf{If $j_p \notin \mathcal{N}(s)$ is partially connected:} Node $j_p$ holds the accumulated collection $\mathcal{s}_{{j_p}}^{(R)}$ as defined above.
    \begin{itemize}
    \item If $j$ is also partially connected, $j_p$ forwards its \emph{entire} accumulated ordered tuple $\mathcal{s}_{j_p}^{(R)}$, including every copy from every distinct path, to node $j$. Node $j$ then forms its own ordered tuple $\mathcal{s}_{j}^{(R)}$ by concatenating all tuples and data vectors received from its parents in topological arrival order. The same underlying matrix may appear multiple times in $\mathcal{s}_{j}^{(R)}$, once per distinct path through which it arrived, and all copies must be retained, since each corresponds to a distinct path that may be required by different downstream nodes. Since all nodes share knowledge of the network topology, the position of each entry in $\mathcal{s}_{j}^{(R)}$ unambiguously identifies the path it traversed, with no additional transmission overhead.
    \item If $j$ is $d$-vertex connected, since in $\hat{\mathcal{G}}$ every node's in-degree equals its vertex connectivity, there exist exactly $\kappa(s,j)=d$ vertex-disjoint paths from $s$ to $j$; denote them $\pi_1,\dots,\pi_d$. Since these paths are vertex-disjoint, no two share a parent of $j$, so each parent $j_p\in\mathcal{P}(j)$ belongs to exactly one distinct path among $\pi_1,\dots,\pi_d$. In particular, every partially-connected parent is essential for maintaining $\kappa(s,j)$, and \cite[Claim~5.1]{eavsPaper} applies to all of them. By that claim, for each partially-connected parent $j_p$, there exists an index $\eta(j_p)\notin\mathcal{P}(j)$ such that $[\boldsymbol{R}^{(1)}\boldsymbol{v}_{\eta(j_p)}^{(R)},\dots,\boldsymbol{R}^{(n)}\boldsymbol{v}_{\eta(j_p)}^{(R)}]\in\mathcal{s}_{j_p}^{(R)}$, and these indices are mutually distinct. Since $\eta(j_p)\notin\mathcal{P}(j)$, node $\eta(j_p)$ can only reach $j$ through $j_p$, so it lies on the same path $\pi_i$ as $j_p$. Since $\mathcal{s}_{j_p}^{(R)}$ is ordered by topological arrival order and the topology is known to all nodes, $j_p$ unambiguously identifies the copy of $[\boldsymbol{R}^{(1)}\boldsymbol{v}_{\eta(j_p)}^{(R)},\dots,\boldsymbol{R}^{(n)}\boldsymbol{v}_{\eta(j_p)}^{(R)}]$ corresponding to path $\pi_i$ by its position in the tuple, and transmits:
        \[
          \boldsymbol{s}_{j_p \to j}^{(R)} = \bigl[{\boldsymbol{v}_{j}^{(R)}}^T \boldsymbol{R}^{(1)} \boldsymbol{v}_{\eta(j_p)}^{(R)}, \dots, {\boldsymbol{v}_{j}^{(R)}}^T \boldsymbol{R}^{(n)} \boldsymbol{v}_{\eta(j_p)}^{(R)}\bigr]^T.
        \]
        \end{itemize}
\end{itemize}

Simultaneously, every $d$-vertex connected node $j$ receives $\alpha_j$ and the hash values $\{h_{j_p \to j}^{(M)}\}_{j_p \in \mathcal{P}(j)}$ and $\{h_{j_p \to j}^{(R)}\}_{j_p \in \mathcal{P}(j)}$ from the source via Lemma~\ref{lem:parallel_edge_active_secure} over the $d$ vertex-disjoint paths. For each $j_p \in \mathcal{P}(j)$, node $j$ verifies consistency by checking:
\begin{equation}\label{eq2-newM}
  h_{j_p \to j}^{(M)} = \sum_{i=1}^{n} [\hat{\boldsymbol{s}}_{j_p \to j}^{(M)}]_{i} (\alpha_j)^{i},
\end{equation}
and
\begin{equation}\label{eq2-newR}
  h_{j_p \to j}^{(R)} = \sum_{i=1}^{n} [\hat{\boldsymbol{s}}_{j_p \to j}^{(R)}]_{i} (\alpha_j)^{i},
\end{equation}
where $\hat{\boldsymbol{s}}_{j_p \to j}^{(M)}$ and $\hat{\boldsymbol{s}}_{j_p \to j}^{(R)}$ denote the received vectors from $j_p$. If either equality fails, node $j$ declares the transmission corrupted and treats it as an erasure. Because the hash check reduces any active corruption to the setting of erasure, and there can be at most $\ell_e+\ell_{oe}$ corrupted nodes in the parent set of $j$, we assume from this point on that node $j$ receives at most $\ell_e+\ell_{oe}$ erasures and zero undetected errors.

For any $d$-vertex connected node $j$, let $\mathcal{P}_{par}(j)$ be the set of partially-connected nodes in the parent set of $j$. Then, the Vandermonde matrices $\boldsymbol{V}_j^{(M)} \in \mathbb{F}_q^{(d'-z) \times (d-z)}$ and $\boldsymbol{V}_j^{(R)} \in \mathbb{F}_q^{d' \times d}$ consist of
$$\boldsymbol{V}_j^{(M)} = \bigl[\, \boldsymbol{V}_{s\to j}^{(M)} \;,\; \{ \boldsymbol{v}_{j_p}^{(M)} \}_{j_p \in \mathcal{P}(j) \setminus \mathcal{P}_{par}(j)} \,\bigr],$$
and
$$\small \boldsymbol{V}_{j}^{(R)} = \bigl[\, \boldsymbol{V}_{s\to j}^{(R)} \;,\; \{ \boldsymbol{v}_{j_p}^{(R)} \}_{j_p \in \mathcal{P}(j) \setminus \mathcal{P}_{par}(j)} \;,\; \{ \boldsymbol{v}_{\eta(j_{p})}^{(R)} \}_{j_p \in \mathcal{P}_{par}(j)} \,\bigr],$$
respectively. Note that $\boldsymbol{V}_j^{(M)}$ may have more than $(d-z)$ columns when $|\mathcal{P}_{\mathrm{par}}(j)|<z$. It is sufficient to consider the worst case $|\mathcal{P}_{\mathrm{par}}(j)|=z$, since the result for smaller $|\mathcal{P}_{\mathrm{par}}(j)|$ follows by restricting to any $(d-z)$ columns of $\boldsymbol{V}_j^{(M)}$. For each $i\in[n]$, node $j$ obtains
\begin{align*}
  {\boldsymbol{v}_{j}^{(M)}}^{T}\boldsymbol{M}^{(i)}\boldsymbol{V}_j^{(M)} &\in \mathbb{F}_{q}^{1\times(d-z)}, \\
  {\boldsymbol{v}_{j}^{(R)}}^{T}\boldsymbol{R}^{(i)}\boldsymbol{V}_j^{(R)} &\in \mathbb{F}_{q}^{1\times d}.
\end{align*}
Since $\boldsymbol{V}_j^{(M)}$ and $\boldsymbol{V}_j^{(R)}$ are Vandermonde matrices, they generate $(d-z,\,d'-z)$ and $(d,\,d')$ MDS codes over $\mathbb{F}_{q}$, respectively, with minimum distances $d_{\min}^{(M)} = \ell_e+\ell_{oe}+1$ and $d_{\min}^{(R)} = \ell_e+\ell_{oe}+1$. An MDS code with minimum distance $\ell_e+\ell_{oe}+1$ corrects up to $\ell_e+\ell_{oe}$ erasures \cite{macwilliams1977theory}. Thus, for every $i\in[n]$, node $j$ recovers $\boldsymbol{M}^{(i)}\boldsymbol{v}_{j}^{(M)}$ and $\boldsymbol{R}^{(i)}\boldsymbol{v}_{j}^{(R)}$ and reconstructs:
\begin{align*}
  \boldsymbol{{S}}_{j}^{(M)} &= [\boldsymbol{M}^{(1)} \boldsymbol{v}_{j}^{(M)}, \dots, \boldsymbol{M}^{(n)} \boldsymbol{v}_{j}^{(M)}], \\
  \boldsymbol{{S}}_{j}^{(R)} &= [\boldsymbol{R}^{(1)} \boldsymbol{v}_{j}^{(R)}, \dots, \boldsymbol{R}^{(n)} \boldsymbol{v}_{j}^{(R)}].
\end{align*}
By induction over the topological ordering, every node reconstructs its share. In particular, every terminal $t \in \mathcal{T}_i$ reconstructs $\boldsymbol{{S}}_{t}^{(M)}$ and $\boldsymbol{{S}}_{t}^{(R)}$.

\subsection*{Decoding:}
\paragraph*{Step 3}
Using the recovered matrices $\boldsymbol{S}_t^{(M)}$ and $\boldsymbol{S}_t^{(R)}$, the terminal node $t \in \mathcal{T}_i$ recovers the key $\boldsymbol{K}_{\mathcal{T}_i}$ with its $i$-th column (for $i \in [n]$):
\[
  \boldsymbol{k}_{\mathcal{T}_i}^{(i)} = [\boldsymbol{M}^{(i)} \boldsymbol{v}_{\mathcal{T}_i}^{(M)}]_{1:d'-\ell'-z+1} + [\boldsymbol{R}^{(i)} \boldsymbol{v}_{\mathcal{T}_i}^{(R)}]_{1:d'-\ell'-z+1}.
\]

\subsection*{Decoding Error Probability}
The overall decoding error probability arises exclusively from the $d$-vertex connected nodes. It is attributed to two sources: the probability that the local hash verification fails to detect an adversarially corrupted transmission, and the error probability of each application of Lemma~\ref{lem:parallel_edge_active_secure}. Let $P_{pd}$ denote the error probability of each such application.

Consider a $d$-vertex connected node $j$ and a corrupted parent $j_p \in \beta_{\mathrm{jam}}$. The adversary attempts to forge $\hat{\boldsymbol{s}}_{j_p \to j}^{(M)} \neq \boldsymbol{s}_{j_p \to j}^{(M)}$ or $\hat{\boldsymbol{s}}_{j_p \to j}^{(R)} \neq \boldsymbol{s}_{j_p \to j}^{(R)}$ while satisfying the verification checks \eqref{eq2-newM}--\eqref{eq2-newR}. By the adversary model, when injecting an error at $j_p$, the adversary's valid observation set is $\mathcal{O}(j_p) = \{u \in \beta_{\mathrm{obs}} \mid j_p \not\prec_{\mathcal{G}} u\}$, i.e., only those observed nodes that are not descendants of $j_p$. Since $(j_p, j) \in \mathcal{E}$, we have $j_p \prec_{\mathcal{G}} j$, so $j \notin \mathcal{O}(j_p)$. Therefore, $\alpha_j$ is not accessible to the adversary when determining its corruption at $j_p$, and hence remains uniformly distributed and independent of the adversary's corruption strategy. Substituting the true hash definition into the verification condition yields, for either transmission component,
$$\sum_{i=1}^{n} \left( [\hat{\boldsymbol{s}}_{j_p \to j}]_i - [\boldsymbol{s}_{j_p \to j}]_i \right) (\alpha_j)^i = 0.$$
Here $\boldsymbol{s}_{j_p \to j}$ represents either $\boldsymbol{s}_{j_p \to j}^{(M)}$ or $\boldsymbol{s}_{j_p \to j}^{(R)}$. Let $\boldsymbol{s}'_{j_p \to j} \triangleq \hat{\boldsymbol{s}}_{j_p \to j} - \boldsymbol{s}_{j_p \to j}$ denote the injected error vector. The verification condition reduces to
$$\sum_{i=1}^{n} [\boldsymbol{s}'_{j_p \to j}]_i (\alpha_j)^i = 0.$$
If $\boldsymbol{s}'_{j_p \to j} = \boldsymbol{0}$, no corruption occurs and the check holds trivially. Otherwise, the above expression defines a non-zero polynomial in $\alpha_j$ of degree at most $n$. Since $\alpha_j$ is independent of $\boldsymbol{s}'_{j_p \to j}$, the polynomial has at most $n$ roots in $\mathbb{F}_q$, so for a single transmission,
$$\Pr\!\left( \text{Local Check passes} \mid j_p \in \beta_{\mathrm{jam}} \right) \le \frac{n}{q}.$$
Both $\boldsymbol{s}_{j_p \to j}^{(M)}$ and $\boldsymbol{s}_{j_p \to j}^{(R)}$ use the same hash structure, so the same bound applies to each, yielding a factor of $2$. Applying a union bound over all $d$-vertex connected nodes $\mathcal{V}_d$ and at most $\ell_e+\ell_{oe}$ corrupted parents per node, the overall decoding error probability $P_{\mathrm{err}}$ satisfies
$$P_{\mathrm{err}} \le |\mathcal{V}_d|\left(P_{pd} + 2(\ell_e + \ell_{oe})\frac{n}{q}\right).$$

\subsection*{Key Rate:}
The secret key $\boldsymbol{K}_{\mathcal{T}_i}$ contains $n(d' - \ell' - z + 1) = n(d - \ell_o - \ell_e - 2\ell_{oe} - z + 1)$ symbols over $\mathbb{F}_q$. The total network transmission blocklength consists of the data transmission phase and the overhead from the secure delivery of verification parameters. We analyze each component in turn.

\paragraph*{M-related transmission cost} Each edge carries $n$ symbols of $\boldsymbol{M}$-related data, giving a cost of $n$ symbols per edge.

\paragraph*{R-related transmission cost} Let $L$ denote the length of the longest contiguous directed chain of partially-connected nodes in $\hat{\mathcal{G}}$. Every edge terminating at a $d$-vertex connected node carries an $\boldsymbol{R}$-related data vector of only $n$ symbols, so the dominant cost arises on edges terminating at partially-connected nodes. Recall from Step~2 that when the parent $j_p$ of a partially-connected node $j$ is either $d$-vertex connected or a source-neighbor partially-connected node $j_p \in \mathcal{N}(s)$, it performs no combining and transmits a single matrix share $[\boldsymbol{R}^{(1)}\boldsymbol{v}_{j_p}^{(R)}, \dots, \boldsymbol{R}^{(n)}\boldsymbol{v}_{j_p}^{(R)}] \in \mathbb{F}_q^{d' \times n}$, that is, $nd'$ symbols, whereas when $j_p$ is itself partially connected it performs no processing at all and simply forwards its entire accumulated collection $\mathcal{s}_{j_p}^{(R)}$ unmodified over the edge to $j$. The size of this edge transmission is therefore exactly the size of $\mathcal{s}_{j_p}^{(R)}$, so bounding the per-edge cost amounts to bounding, for each partially-connected node, the number of $d'\times n$ matrices in its accumulated collection. We index this bound by $\lambda$, the length of the longest contiguous chain of partially-connected nodes terminating at the node, and prove by induction on $\lambda$ that a partially-connected node with index $\lambda$ holds at most $(d-1)^{\lambda}$ matrices. For $\lambda = 1$, node $j$ has no partially-connected parent, since such a parent together with $j$ would form a chain of length two; every parent of $j$ is then either $d$-vertex connected or a source-neighbor partially-connected node, and each forwards exactly one $d'\times n$ matrix. Since $|\mathcal{P}_{\hat{\mathcal{G}}}(j)| = \kappa_{\hat{\mathcal{G}}}(s,j) \le d-1$, node $j$ receives at most $d-1$ such matrices, one from each parent, so $\mathcal{s}_{j}^{(R)}$ contains at most $d-1 = (d-1)^1$ matrices. For $\lambda = 2$, node $j$ has at least one partially-connected parent $j_p$, and since a chain of length two terminating at $j$ cannot extend further back through $j_p$ without exceeding $\lambda$, every partially-connected parent of $j$ has index $1$ and, by the base case, holds and forwards at most $d-1$ matrices over its edge to $j$; every non-partially-connected parent of $j$ still forwards a single matrix. Because $j$ has at most $d-1$ parents in total and $\mathcal{s}_j^{(R)}$ is the union of what all its parents forward, the count is maximized when every parent of $j$ is itself partially connected with index $1$, giving at most $(d-1)(d-1) = (d-1)^2$ matrices in $\mathcal{s}_j^{(R)}$. The same reasoning extends to general $\lambda > 1$: assume every partially-connected node of index at most $\lambda - 1$ holds at most $(d-1)^{\lambda - 1}$ matrices, and let $j$ have index $\lambda$. Appending $j$ to a chain terminating at a partially-connected parent $j_p \in \mathcal{P}(j)$ produces a chain terminating at $j$ that is one node longer, so if the chain at $j_p$ had length $\lambda$ or more, the chain at $j$ would have length exceeding $\lambda$, contradicting the index of $j$; hence every partially-connected parent of $j$ has index at most $\lambda - 1$, and by the induction hypothesis holds and forwards at most $(d-1)^{\lambda-1}$ matrices, while every non-partially-connected parent still forwards exactly one. Since $j$ has at most $d-1$ parents, the size of $\mathcal{s}_j^{(R)}$ is maximized when all $d-1$ parents are partially connected with index $\lambda - 1$, giving at most $(d-1)\cdot(d-1)^{\lambda-1} = (d-1)^{\lambda}$ matrices and completing the induction. Applying this bound to the edge between two consecutive partially-connected nodes in the longest chain, the transmitting parent node has index at most $L-1$ and forwards at most $(d-1)^{L-1}$ matrices of $nd'$ symbols each, so this edge carries at most $nd'(d-1)^{L-1}$ symbols; for $L=1$ no such edge exists, and the bound reduces to the single $nd'$-symbol matrix forwarded by a non-partially-connected parent. Every other edge, including the source's direct transmissions to partially-connected nodes in Step~1, carries a single matrix of $nd'$ symbols, so $nd'(d-1)^{L-1}$ bounds the $\boldsymbol{R}$-related per-edge cost across the entire network.

\textit{Role of the in-degree constraint.} The restriction to $\hat{\mathcal{G}}$ is what makes the bound above depend only on $d$ and $L$ rather than on the network size and topology. Without it, a partially-connected node's in-degree is an uncontrolled quantity determined by the topology of $\mathcal{G}$, even though its vertex connectivity from $s$ is at most $d-1$, and repeating the induction with this topology-dependent in-degree in place of $d-1$ would yield a per-edge cost, and hence a blocklength, that grows with the network size and topology rather than with the connectivity parameters alone. Working in $\hat{\mathcal{G}}$, where every node's in-degree equals its vertex connectivity, eliminates these redundant transmissions and caps the number of copies accumulated at depth $\lambda$ at $(d-1)^{\lambda}$, so the achieved rate depends on the network only through $d$ and $L$.

Since the $\boldsymbol{M}$-related and $\boldsymbol{R}$-related transmissions are sequential on the same edge, the total data transmission blocklength is at most
$$n + nd'(d-1)^{L-1} = n\bigl(d'(d-1)^{L-1} + 1\bigr)$$
symbols per edge.

\paragraph*{Overhead from verification parameter delivery} Let $R_{pd}$ denote the transmission rate achieved by Lemma~\ref{lem:parallel_edge_active_secure}. For each $d$-vertex connected node $j \in \mathcal{V}_d$, the source must securely deliver the tuple 
$$(\alpha_j, \{h_{j_p \to j}^{(M)}\}_{j_p \in \mathcal{P}(j)}, \{h_{j_p \to j}^{(R)}\}_{j_p \in \mathcal{P}(j)}),$$
comprising $2d+1$ symbols. Doing so requires a blocklength overhead of $\frac{2d+1}{R_{pd}}$ per node. Summing over all $|\mathcal{V}_d|$ such nodes gives a total overhead of $|\mathcal{V}_d|\frac{2d+1}{R_{pd}}$ transmissions.

\textit{Achieved key rate.} Combining the data blocklength and the verification overhead, the achieved secure key rate is
\[
  R = \frac{n(d-\ell_o-\ell_e-2\ell_{oe}-z+1)}{n\bigl(d'(d-1)^{L-1}+1\bigr) + |\mathcal{V}_d|\dfrac{2d+1}{R_{pd}}}.
\]

\subsection*{Key-Capacity}
We now formalize the key-capacity based on the established definition. Recall that an instance has key-capacity $C$ if for any $\epsilon > 0$ and $\rho > 0$, there exists a key-code achieving rate $R \ge C - \rho$ with decoding error probability $P_{\mathrm{err}} \le \epsilon$.

For fixed network topology and $R_{pd}>0$, as $n \to \infty$ the rate approaches:
\[
  \lim_{n \to \infty} R = \frac{d-\ell_o-\ell_e-2\ell_{oe}-z+1}{d'(d-1)^{L-1}+1}.
\]
Thus, for any target gap $\rho > 0$, there exists a sufficiently large $n$ such that $R \ge \frac{d-\ell_o-\ell_e-2\ell_{oe}-z+1}{d'(d-1)^{L-1}+1} - \rho$. The decoding error probability satisfies
\[
  P_{\mathrm{err}} \le |\mathcal{V}_d| \!\left(P_{pd} + 2(\ell_e + \ell_{oe})\frac{n}{q}\right).
\]
Once $n$ is fixed, over a sufficiently large field $\mathbb{F}_q$ with $q \gg n$, the term $2|\mathcal{V}_d|(\ell_e+\ell_{oe})\frac{n}{q} \to 0$. Furthermore, $P_{pd} \to 0$ as $\hat{q} \to \infty$ by the proof of Lemma~\ref{lem:parallel_edge_active_secure}. Therefore, for any $\epsilon > 0$, sufficiently large $q$ and $\hat{q}$ ensure $P_{\mathrm{err}} \le \epsilon$. Combining this achievable asymptotic rate with the connectivity bounds of Lemma~\ref{lem:parallel_edge_active_secure}, the network key-capacity under the additive jamming model satisfies:
\[
  C_{\mathrm{add}} \ge \begin{cases} \dfrac{d-\ell_o-\ell_e-2\ell_{oe}-z+1}{d'(d-1)^{L-1}+1}, & \text{if } d>\ell_o+\ell_e+2\ell_{oe},\\[1.5ex] 0, & \text{otherwise,} \end{cases}
\]
and the network key-capacity under the overwrite jamming model satisfies:
\[
  C_{\mathrm{ow}} \ge \begin{cases} \dfrac{d-\ell_o-\ell_e-2\ell_{oe}-z+1}{d'(d-1)^{L-1}+1}, & \text{if } d>\ell_o+2\ell_e+2\ell_{oe},\\[1.5ex] 0, & \text{otherwise.} \end{cases}
\]

\subsection*{Security Proof:}
Let $\ell' = \ell_o + \ell_{oe}$ be the maximum number of nodes the adversary can observe. We aim to show that for any terminal set $\mathcal{T}_i \subseteq \mathcal{T}$ and any adversarial observation set $\beta_{\mathrm{obs}} \in \mathcal{B}$ where $\beta_{\mathrm{obs}} \cap \mathcal{T}_i = \emptyset$, the adversary gains no information about the key, i.e.,
\begin{equation}\small
\begin{aligned}
I\big(\boldsymbol{K}_{\mathcal{T}_i} &;\, \{X_v^{\text{main}}\}_{v\in\beta_{\mathrm{obs}}}, \{X_v^{\text{pd}}\}_{v\in\beta_{\mathrm{obs}}}\big) \\ &= I\big(\boldsymbol{K}_{\mathcal{T}_i};\, \{X_v^{\text{main}}\}_{v\in\beta_{\mathrm{obs}}}\big) \\ & \quad + I\big(\boldsymbol{K}_{\mathcal{T}_i};\, \{X_v^{\text{pd}}\}_{v\in\beta_{\mathrm{obs}}} \mid \{X_v^{\text{main}}\}_{v\in\beta_{\mathrm{obs}}}\big) \\ &= 0,
\end{aligned}
\end{equation}
where $X_v^{\text{main}}$ denotes the data stored at node $v$ during the main transmission phase (node shares, hashes, and $\alpha_v$), and $X_v^{\text{pd}}$ represents the data intercepted at node $v$ during the applications of Lemma~\ref{lem:parallel_edge_active_secure} (the physical transmissions routed through $v$ to deliver verification parameters to other nodes).

To evaluate the conditional term $I\big(\boldsymbol{K}_{\mathcal{T}_i};\, \{X_v^{\text{pd}}\}_{v\in\beta_{\mathrm{obs}}} \mid \{X_v^{\text{main}}\}_{v\in\beta_{\mathrm{obs}}}\big)$, we apply the exact reasoning established in the security proof of Theorem~\ref{T1} in Appendix~\ref{T1Proof}. For each application of Lemma~\ref{lem:parallel_edge_active_secure}, the transmitted payloads are masked by a fresh, independent random matrix that maintains full column rank, acting as a perfect one-time pad. Because this randomness is generated independently for each transmission and is completely independent of the secret key and the main transmission, the intercepted data $\{X_v^{\text{pd}}\}_{v\in\beta_{\mathrm{obs}}}$ reveals zero additional information. Consequently:
$$I\big(\boldsymbol{K}_{\mathcal{T}_i};\, \{X_v^{\text{pd}}\}_{v\in\beta_{\mathrm{obs}}} \mid \{X_v^{\text{main}}\}_{v\in\beta_{\mathrm{obs}}}\big) = 0.$$
This leaves the total information leakage bounded solely by the main transmission phase, $I\big(\boldsymbol{K}_{\mathcal{T}_i};\, \{X_v^{\text{main}}\}_{v\in\beta_{\mathrm{obs}}}\big)$. We now evaluate this term under two distinct topological scenarios based on the nature of the nodes observed by the adversary.

\paragraph*{Case 1: Eavesdropper Observes at Least One Partially-Connected Node}
Because partially-connected nodes may possess significant information regarding $\{\boldsymbol{R}^{(\tau)}\}_{\tau \in [n]}$, we bound the secrecy under the worst-case assumption where the eavesdropper has full knowledge of all matrices $\{\boldsymbol{R}^{(\tau)}\}_{\tau=1}^n$.
\begin{equation}
\small
\begin{aligned}
I(\boldsymbol{K}_{\mathcal{T}_i}; & \{X_v^{\text{main}}:v\in\beta_{\mathrm{obs}}\}) \\ &\le I\big(\boldsymbol{K}_{\mathcal{T}_i}; \{X_v^{\text{main}}:v\in\beta_{\mathrm{obs}}\}, \{\boldsymbol{R}^{(\tau)}\}_{\tau=1}^n \big) \\ &= I\!\big(\boldsymbol{K}_{\mathcal{T}_i}; \{\boldsymbol{M}^{(\tau)}\boldsymbol{V}_{\epsilon}^{(M)}\}_{\tau=1}^n, \{\boldsymbol{R}^{(\tau)}\}_{\tau=1}^n, \\ & \qquad \{\alpha_j\}_{j\in\beta_{\mathrm{obs}}}, \bigl\{\{h_{j_p \to j}^{(M)}\}_{j_p \in \mathcal{P}(j)}\bigr\}_{j \in \beta_{\mathrm{obs}}}, \\ & \qquad \bigl\{\{h_{j_p \to j}^{(R)}\}_{j_p \in \mathcal{P}(j)}\bigr\}_{j \in \beta_{\mathrm{obs}}} \big) \\ &\overset{(a)}{=} I\!\big(\boldsymbol{K}_{\mathcal{T}_i}; \{\boldsymbol{M}^{(\tau)}\boldsymbol{V}_{\epsilon}^{(M)}\}_{\tau=1}^n, \{\boldsymbol{R}^{(\tau)}\}_{\tau=1}^n, \{\alpha_j\}_{j\in\beta_{\mathrm{obs}}} \big) \\ &\overset{(b)}{=} I\!\big(\boldsymbol{K}_{\mathcal{T}_i}; \{\boldsymbol{M}^{(\tau)}\boldsymbol{V}_{\epsilon}^{(M)}\}_{\tau=1}^n, \{\boldsymbol{R}^{(\tau)}\}_{\tau=1}^n \big) \\ &\overset{(c)}{=} I\!\Big( [\boldsymbol{M}^{(1)} \boldsymbol{v}_{\mathcal{T}_i}^{(M)}, \dots, \boldsymbol{M}^{(n)} \boldsymbol{v}_{\mathcal{T}_i}^{(M)}]_{1:d'-\ell'-z+1} \\ & \qquad + [\boldsymbol{R}^{(1)} \boldsymbol{v}_{\mathcal{T}_i}^{(R)}, \dots, \boldsymbol{R}^{(n)} \boldsymbol{v}_{\mathcal{T}_i}^{(R)}]_{1:d'-\ell'-z+1} ; \\ & \qquad \{\boldsymbol{M}^{(\tau)}\boldsymbol{V}_{\epsilon}^{(M)}\}_{\tau=1}^n \;\Big|\; \{\boldsymbol{R}^{(\tau)}\}_{\tau=1}^n \Big) \\ &\overset{(d)}{=} I\!\Big( [\boldsymbol{M}^{(1)} \boldsymbol{v}_{\mathcal{T}_i}^{(M)}, \dots, \boldsymbol{M}^{(n)} \boldsymbol{v}_{\mathcal{T}_i}^{(M)}]_{1:d'-\ell'-z+1} ; \\ & \qquad \{\boldsymbol{M}^{(\tau)}\boldsymbol{V}_{\epsilon}^{(M)}\}_{\tau=1}^n \;\Big|\; \{\boldsymbol{R}^{(\tau)}\}_{\tau=1}^n \Big) \\ &\overset{(e)}{=} I\!\Big( [\boldsymbol{M}^{(1)} \boldsymbol{v}_{\mathcal{T}_i}^{(M)}, \dots, \boldsymbol{M}^{(n)} \boldsymbol{v}_{\mathcal{T}_i}^{(M)}]_{1:d'-\ell'-z+1} ; \\ & \qquad \{\boldsymbol{M}^{(\tau)}\boldsymbol{V}_{\epsilon}^{(M)}\}_{\tau=1}^n \Big) \\ &\overset{(f)}{=} \sum_{\tau=1}^{n} I\!\Big([\boldsymbol{M}^{(\tau)}\boldsymbol{v}_{\mathcal{T}_i}^{(M)}]_{1:d'-\ell'-z+1}; \boldsymbol{M}^{(\tau)}\boldsymbol{V}_{\epsilon}^{(M)} \Big) \\ &\overset{(g)}{=} 0,
\end{aligned}
\end{equation}
where $\boldsymbol{V}_\epsilon^{(M)} \in \mathbb{F}_q^{(d'-z) \times \ell'}$ is the Vandermonde matrix formed by the columns $\{\boldsymbol{v}_j^{(M)}\}_{j \in \beta_{\mathrm{obs}}}$. Step (a) holds because, for all eavesdropped nodes $j \in \beta_{\mathrm{obs}}$, the hashes $\bigl\{\{h_{j_p\to j}^{(M)}\}_{j_p\in\mathcal{P}(j)}\bigr\}_{j\in\beta_{\mathrm{obs}}}$ and $\bigl\{\{h_{j_p\to j}^{(R)}\}_{j_p\in\mathcal{P}(j)}\bigr\}_{j\in\beta_{\mathrm{obs}}}$ can be deterministically computed from $\{\alpha_j\}_{j\in\beta_{\mathrm{obs}}}$, $\{\boldsymbol{M}^{(\tau)}\boldsymbol{V}_{\epsilon}^{(M)}\}_{\tau \in [n]}$, and $\{\boldsymbol{R}^{(\tau)}\}_{\tau \in [n]}$. Step (b) follows because $\{\alpha_j\}_{j\in\beta_{\mathrm{obs}}}$ are independently generated uniform random symbols, independent of all matrices. Step (c) follows by substituting the definition of $\boldsymbol{K}_{\mathcal{T}_i}$ and applying the chain rule, using the fact that $I\!\bigl([\boldsymbol{M}^{(\tau)}\boldsymbol{v}_{\mathcal{T}_i}^{(M)}]_{1:d'-\ell'-z+1} + [\boldsymbol{R}^{(\tau)}\boldsymbol{v}_{\mathcal{T}_i}^{(R)}]_{1:d'-\ell'-z+1}; \{\boldsymbol{M}^{(\tau)}\}_{\tau=1}^n\bigr) = 0$ since $\{\boldsymbol{M}^{(\tau)}\}$ and $\{\boldsymbol{R}^{(\tau)}\}$ are independent. Step (d) follows because conditioning on $\{\boldsymbol{R}^{(\tau)}\}_{\tau=1}^n$ makes the $\boldsymbol{R}$-related key vectors deterministic, allowing them to be dropped. Step (e) follows since $\{\boldsymbol{M}^{(\tau)}\}$ is independent of $\{\boldsymbol{R}^{(\tau)}\}$, so the conditioning can be dropped. Step (f) follows since $\boldsymbol{M}^{(1)},\dots,\boldsymbol{M}^{(n)}$ are independent, making the pairs $\bigl([\boldsymbol{M}^{(\tau)}\boldsymbol{v}_{\mathcal{T}_i}^{(M)}]_{1:d'-\ell'-z+1}, \boldsymbol{M}^{(\tau)}\boldsymbol{V}_\epsilon^{(M)}\bigr)$ independent across $\tau$. Step (g) follows from \cite[Lemma~3.2]{eavsPaper} and the fact that $\boldsymbol{v}_{\mathcal{T}_i}^{(M)}$ is linearly independent of the columns of $\boldsymbol{V}_\epsilon^{(M)}$.

\paragraph*{Case 2: Eavesdropper Observes Only Fully-Connected Nodes}
Here, the eavesdropper observes no partially-connected nodes. In the worst-case scenario, the eavesdropper gains full knowledge of all matrices $\{\boldsymbol{M}^{(\tau)}\}_{\tau=1}^n$.
\begin{equation}
\small
\begin{aligned}
I(\boldsymbol{K}_{\mathcal{T}_i}; & \{X_v^{\text{main}}:v\in\beta_{\mathrm{obs}}\}) \\ &\le I\big(\boldsymbol{K}_{\mathcal{T}_i}; \{X_v^{\text{main}}:v\in\beta_{\mathrm{obs}}\}, \{\boldsymbol{M}^{(\tau)}\}_{\tau=1}^n \big) \\ &= I\!\big(\boldsymbol{K}_{\mathcal{T}_i}; \{\boldsymbol{R}^{(\tau)}\boldsymbol{V}_{\epsilon}^{(R)}\}_{\tau=1}^n, \{\boldsymbol{M}^{(\tau)}\}_{\tau=1}^n, \\ & \qquad \{\alpha_j\}_{j\in\beta_{\mathrm{obs}}}, \bigl\{\{h_{j_p \to j}^{(M)}\}_{j_p \in \mathcal{P}(j)}\bigr\}_{j \in \beta_{\mathrm{obs}}}, \\ & \qquad \bigl\{\{h_{j_p \to j}^{(R)}\}_{j_p \in \mathcal{P}(j)}\bigr\}_{j \in \beta_{\mathrm{obs}}} \big) \\ &\overset{(a)}{=} I\!\big(\boldsymbol{K}_{\mathcal{T}_i}; \{\boldsymbol{R}^{(\tau)}\boldsymbol{V}_{\epsilon}^{(R)}\}_{\tau=1}^n, \{\boldsymbol{M}^{(\tau)}\}_{\tau=1}^n, \{\alpha_j\}_{j\in\beta_{\mathrm{obs}}} \big) \\ &\overset{(b)}{=} I\!\big(\boldsymbol{K}_{\mathcal{T}_i}; \{\boldsymbol{R}^{(\tau)}\boldsymbol{V}_{\epsilon}^{(R)}\}_{\tau=1}^n, \{\boldsymbol{M}^{(\tau)}\}_{\tau=1}^n \big) \\ &\overset{(c)}{=} I\!\Big( [\boldsymbol{M}^{(1)} \boldsymbol{v}_{\mathcal{T}_i}^{(M)}, \dots, \boldsymbol{M}^{(n)} \boldsymbol{v}_{\mathcal{T}_i}^{(M)}]_{1:d'-\ell'-z+1} \\ & \qquad + [\boldsymbol{R}^{(1)} \boldsymbol{v}_{\mathcal{T}_i}^{(R)}, \dots, \boldsymbol{R}^{(n)} \boldsymbol{v}_{\mathcal{T}_i}^{(R)}]_{1:d'-\ell'-z+1} ; \\ & \qquad \{\boldsymbol{R}^{(\tau)}\boldsymbol{V}_{\epsilon}^{(R)}\}_{\tau=1}^n \;\Big|\; \{\boldsymbol{M}^{(\tau)}\}_{\tau=1}^n \Big) \\ &\overset{(d)}{=} I\!\Big( [\boldsymbol{R}^{(1)} \boldsymbol{v}_{\mathcal{T}_i}^{(R)}, \dots, \boldsymbol{R}^{(n)} \boldsymbol{v}_{\mathcal{T}_i}^{(R)}]_{1:d'-\ell'-z+1} ; \\ & \qquad \{\boldsymbol{R}^{(\tau)}\boldsymbol{V}_{\epsilon}^{(R)}\}_{\tau=1}^n \;\Big|\; \{\boldsymbol{M}^{(\tau)}\}_{\tau=1}^n \Big) \\ &\overset{(e)}{=} I\!\Big( [\boldsymbol{R}^{(1)} \boldsymbol{v}_{\mathcal{T}_i}^{(R)}, \dots, \boldsymbol{R}^{(n)} \boldsymbol{v}_{\mathcal{T}_i}^{(R)}]_{1:d'-\ell'-z+1} ; \\ & \qquad \{\boldsymbol{R}^{(\tau)}\boldsymbol{V}_{\epsilon}^{(R)}\}_{\tau=1}^n \Big) \\ &\overset{(f)}{=} \sum_{\tau=1}^{n} I\!\Big([\boldsymbol{R}^{(\tau)}\boldsymbol{v}_{\mathcal{T}_i}^{(R)}]_{1:d'-\ell'-z+1}; \boldsymbol{R}^{(\tau)}\boldsymbol{V}_{\epsilon}^{(R)} \Big) \\ &\overset{(g)}{=} 0,
\end{aligned}
\end{equation}
where $\boldsymbol{V}_\epsilon^{(R)} \in \mathbb{F}_q^{d' \times \ell'}$ is the Vandermonde matrix formed by the eavesdropped nodes. Step (a) holds because, for all eavesdropped nodes $j \in \beta_{\mathrm{obs}}$, the set of hashes $\bigl\{\{h_{j_p\to j}^{(M)}\}_{j_p\in\mathcal{P}(j)}\bigr\}_{j\in\beta_{\mathrm{obs}}}$ and $\bigl\{\{h_{j_p\to j}^{(R)}\}_{j_p\in\mathcal{P}(j)}\bigr\}_{j\in\beta_{\mathrm{obs}}}$ can be deterministically computed from $\{\alpha_j\}_{j\in\beta_{\mathrm{obs}}}$, $\{\{\boldsymbol{M}^{(\tau)}\boldsymbol{v}_j\}_{j\in \beta_{\mathrm{obs}}}\}_{\tau \in [n]}$ already present in $\{\boldsymbol{R}^{(\tau)}\boldsymbol{V}_{\epsilon}^{(R)}\}_{\tau \in [n]}$, and $\{\boldsymbol{M}^{(\tau)}\}_{\tau \in [n]}$. Step (b) follows because $\{\alpha_j\}_{j \in \beta_{\mathrm{obs}}}$ are independently generated uniform random symbols, sharing no statistical correlation with the key or matrices. Step (c) follows by substituting the definition of $\boldsymbol{K}_{\mathcal{T}_i}$ and applying the chain rule of mutual information. Note that since $\{\boldsymbol{M}^{(\tau)}\}_{\tau=1}^n$ and $\{\boldsymbol{R}^{(\tau)}\}_{\tau=1}^n$ are all generated independently, $I\!\Big( [\boldsymbol{M}^{(1)} \boldsymbol{v}_{\mathcal{T}_i}^{(M)}, \dots, \boldsymbol{M}^{(n)} \boldsymbol{v}_{\mathcal{T}_i}^{(M)}]_{1:d'-\ell'-z+1} + [\boldsymbol{R}^{(1)} \boldsymbol{v}_{\mathcal{T}_i}^{(R)}, \dots, \boldsymbol{R}^{(n)} \boldsymbol{v}_{\mathcal{T}_i}^{(R)}]_{1:d'-\ell'-z+1} ;\, \{\boldsymbol{M}^{(\tau)}\}_{\tau=1}^n \Big) = 0$. Step (d) follows because conditioning on $\{\boldsymbol{M}^{(\tau)}\}_{\tau=1}^n$ makes the $\boldsymbol{M}$-related key vectors deterministic, allowing them to be subtracted and removed without altering the mutual information. Step (e) follows since the symmetric random matrices $\{\boldsymbol{R}^{(\tau)}\}$ are independent of $\{\boldsymbol{M}^{(\tau)}\}$, allowing the conditioning to be dropped. Step (f) follows since $\boldsymbol{R}^{(1)},\dots,\boldsymbol{R}^{(n)}$ are independently and uniformly generated over $\mathbb{F}_{q}^{d' \times d'}$, and thus the pairs $\bigl([\boldsymbol{R}^{(\tau)}\boldsymbol{v}_{\mathcal{T}_i}^{(R)}]_{1:d'-\ell'-z+1}, \boldsymbol{R}^{(\tau)}\boldsymbol{V}_\epsilon^{(R)}\bigr)$ are independent across $\tau$. Step (g) follows from \cite[Lemma 3.2]{eavsPaper} and the fact that the vector $\boldsymbol{v}_{\mathcal{T}_i}^{(R)}$ is independent of the column space of $\boldsymbol{V}_\epsilon^{(R)}$. Consequently, regardless of the observed node types, the scheme guarantees perfect secrecy against the active adversary.
\end{proof}

\bibliographystyle{IEEEtran}
\bibliography{references}

\end{document}